\documentclass[10pt]{amsart}
\pdfoutput=1
\usepackage{amsmath}
  \usepackage{graphics} %% add this and next lines if pictures should be in esp format
\usepackage{graphicx}  %% add this and next lines if pictures should be in esp format
  \usepackage{epsfig} %For pictures: screened artwork should be set up with an 85 or 100 line screen
  \usepackage{float}
  \usepackage{caption}
\usepackage{amsmath}
\usepackage{graphicx}
\usepackage{subcaption}
\usepackage{amssymb}

\newtheorem{theorem}{Theorem}[section]

\newtheorem{lemma}[theorem]{Lemma}

\newtheorem{conjecture}{Conjecture}

\theoremstyle{definition}

\newtheorem{remark}{Remark}

\title[ epibiont dynamics ] 
      {The effects of invasive epibionts on crab-mussel communities: a theoretical approach to understand mussel population decline}


\begin{document}
\maketitle

\centerline{\scshape  Jingjing Lyu$^{1}$, Linda A. Auker$^{2}$, Anupam Priyadarshi$^{3}$ and Rana D. Parshad$^{4}$
 }
\medskip
{\footnotesize
% please put the address of the first author
 \centerline{1) Department of Mathematics,}
 \centerline{Clarkson University,}
   \centerline{ Potsdam, New York 13699, USA.}
   \medskip
   \centerline{ 2) Department of Biology,}
 \centerline{St. Lawrence University,}
   \centerline{Canton, New York 13617, USA.}
    \medskip
   \centerline{ 3) Department of Mathematics,}
 \centerline{Institute of Science,}
  \centerline{Banaras Hindu University,}
   \centerline{Varanasi-221005 India.}
   \medskip
   \centerline{ 4) Department of Mathematics,}
 \centerline{Iowa State University,}
   \centerline{Ames, Iowa 50011, USA.}
   
 }

%%The abstract of your paper
\begin{abstract}
Blue mussels \emph{(Mytilus edulis)} are an important keystone species that have been declining in the Gulf of Maine. This could be attributed to a variety of complex factors such as \emph{indirect} effects due to invasion by epibionts, which remains unexplored mathematically.
Based on classical optimal foraging theory and anti-fouling defense mechanisms of mussels, we derive an ODE model for crab-mussel interactions in the presence of an invasive epibiont, \emph{Didemnum vexillum}. 
The dynamical analysis leads to results on stability, global boundedness and bifurcations of the model.
 Next, via optimal control methods we predict various ecological outcomes. Our results have key implications for preserving mussel populations in the advent of invasion by non-native epibionts. In particular they help us understand the changing dynamics of local predator-prey communities, due to indirect effects that epibionts confer. 
 \end{abstract}

 \section{Introduction}
 \label{1}
 \subsection{Background}
Blue mussels \emph{(Mytilus edulis)} are an ecologically and economically important species   \cite{FC03,FAO17,SS92}. They play several roles in marine ecosystems: as important prey for many species, such as crabs, shorebirds, sea stars, and gastropod molluscs  \cite{SS92,M83,S85}; as nutrient recyclers and pollution indicators \cite{WD92}; and as a keystone species, serving as habitat for benthic infaunal organisms \cite{SS92,TN85}. However, \emph{M. edulis} has declined in the Gulf of Maine by over 60$\%$ since the 1970s \cite{SD17}. Mussel post-larval settlement, consistent with this observation, has also declined \cite{A10}. The reasons for this decline are unclear, but are almost certainly complex. Thus a clearer understanding of the ecological factors that influence mussel populations is needed.

A primary cause of a species population decline is predation. Invasive predators, like the green crab (\emph{Carcinus maenas}) and the Asian shore crab (\emph{Hemigrapsus sanguineus}), readily prey on the blue mussel \cite{A14,DT04, FB06}. 
However, mussel size limits crab predation, with crabs consuming mussel prey below 70 mm in shell length \cite{E78}. Furthermore, mussels have also adapted to crab predation by thickening their shells in response to novel predator presence, in extremely short time periods \cite{FB06}. Furthermore, substrate complexity reduces predation on mussels as increasingly complex habitats provide refuge from crab predation \cite{FD02}. Thus while predation has put considerable pressure on mussel populations, rapidly evolving defense mechanisms, escape from predation via growth, and physical refuges have counteracted predator impacts. 

Though mussels do have the aforementioned protections against predation, they are still in decline. Curiously, in the 1970s, an introduced ascidian species \emph{Didemnum vexillum}, arrived in the Gulf of Maine \cite{DH07}. \emph{D. vexillum} is a colonial ascidian that is dominant as a competitor for substratum, prolifically laying down mat-like structures on any hard substrate \cite{B07}. Consequently, it acts as an epibiont (fouling organism) on \emph{M. edulis} \cite{A10,A14}. Epibionts impact predator-prey communities indirectly by affecting predation rates on basibionts. \emph{D. vexillum} in particular has chemical anti-predatory defenses. If a crab predator attempts to break off pieces of the \emph{D. vexillum} colony to reach  the mussel, \emph{D. vexillum} releases secondary metabolites and sulfuric acid that deters the crab \cite{B07}. This mechanism by which the epibiont protects the mussel from crab predation is known as associational resistance \cite{LW99}. While it appears to protect mussels from crab predation, \emph{D. vexillum} also negatively affects mussel fecundity and fitness, resulting in fewer progeny \cite{A10}.
Thus \emph{D. vexillum} has both a positive and negative impact on mussel populations. 

Given this complex relationship, we ask,

\begin{itemize}
\item Could the introduced epibiont  \emph{D. vexillum} change the predator-prey dynamics in an established local crab-mussel community? 
\item Could the net effect of positive and negative impacts from \emph{D. vexillum} epibiosis, lead to mussel population decline?
\end{itemize}

Our analysis is (to the best of our knowledge) the first mathematical investigation of predator-prey dynamics under pressure of fouling from epibionts in a crab-mussel community. Herein,

\begin{itemize}
\item  we derive a predator-prey model for crab-mussel interactions, given that clearly a certain size of mussel is preferred or ``optimal" for the crab,

\item we consider the effects of an invasive epibiont by meshing OFT and ant-fouling defense of the mussels,

 \item we model the effects of associational resistance and reduced fecundity, due to the epibiont, and

\item we next investigate dynamical aspects of the model, and use optimal control theory to predict various outcomes.
\end{itemize}
%In order to build a model for such an investigation, we draw from two sources: anti-fouling defence of the mussel and optimal foraging theory. These are \emph{intertwined} from our viewpoint.

 In adult mussels the protective periostracum, which inhibits epibiont settlement when present \cite{WL98} tends to wear off due to age, decay and abrasion. Consequently, the periostracum is more prevalent on newer regions of the shell, while absent on older regions \cite{HS93}. This means juvenile mussels are less likely to be overgrown with epibionts than are adult mussels. When crabs forage for mussels, they typically prefer a medium sized adult. But this preferred size tends to get easily overgrown by epibionts. Epibionts can alter the prey size choice of predators, including crabs, in experiments \cite{ B07,E03, T07, V06}, though this has not yet been tested with \emph{D. vexillum}. 
We hypothesize in the current manuscript that \emph{D. vexillum} can change the feeding preference of crabs away from mid-size adult mussels (that are easily overgrown and therefore less likely to be eaten) towards juvenile mussels (which are less likely to be overgrown and so are easier targets), even though the latter are not the crab's preferred food source. To elucidate our approach we survey some classical results from OFT.

%
%We derive our mathematical model based on this key assumption. To understand better how antifouling defense is coupled with optimal foraging theory (OFT) to derive our model, we first survey some classic results from OFT.
%
 \subsection{Optimal Foraging Theory}

Optimal foraging theory (OFT) predicts how animals behave when they forage for food. 
It is well known that predators optimize feeding strategies to maximize energy intake \cite{D77}.
Essentially,  predators try to gain the most energy from their prey by expending the least amount of energy in the hunting process. For a crab foraging for mussels this amounts to maximizing

\begin{equation}
 \frac{e}{h}  =  \frac{\mbox{energy gained from mussel intake} }{\mbox{handling time of mussel}}.
\end{equation}
This translates to a medium-sized adult mussel as the optimally preferred prey by adult crabs. While large mussels have a potentially high source of energy, they take a much longer time to open and consume than smaller mussels. Small mussels, conversely, take a short time to consume, but they offer very little reward. Even so, juvenile mussels are readily consumed by many species, including green crabs and dogwhelk \cite{A14}. 

 We refer the reader to the mathematical treatment by Krivan \cite{K96,K99}, who describes, via a three species ODE model, in which a predator hunts (optimally) for two prey species
 $u$ and $v$, where $u$ is favored to $v$. The term $u_{1}$ denotes the attack rate with which prey $u$ is hunted, and $u_{2}$ denotes the attack rate with which prey $v$ is hunted.
 The analysis presented in \cite{K99} draws from classic results in OFT and shows that, in order to maximize $ \frac{e}{h}$, the optimal pair of $(u_{1},u_{2})$ is given by $u_{1}=1, u_{2}=0$ if $u > u^{*}$ or $u_{1}=1, u_{2}=1$ if $u < u^{*}$, or $u_{1}=1, 0 < u_{2} < 1$ if $u = u^{*}$, where $u^{*}$ is the critical density for switching. 
 
  \subsubsection{OFT in the presence of epibionts}

Predators are known to switch prey if preferred prey drop below a threshold density \cite{BM01}. For example, fish species have been shown to switch habitats if foraging in one habitat becomes less fruitful than in another habitat \cite{WG81}. Theoretical studies also support this \cite{K96}. 
In our context, If $\frac{e_{1}}{h_{1}} > \frac{e_{2}}{h_{2}}$, the adult mussels are preferred to juveniles as the optimal prey for crabs. If $\frac{e_{1}}{h_{1}} < \frac{e_{2}}{h_{2}}$, juveniles mussels are optimal to hunt. 
%However, the crab's preference for adult mussels or juvenile mussels depends on the epibiont density. Crabs prefer to prey on adults without epibionts, while the choice will shift to juveniles, under high epibiont density.  

\begin{conjecture}
An increase in epibiont density will cause crabs to switch from adult mussels to juveniles even though the juvenile is less preferred and the adult mussel density remains high.
\end{conjecture}
Essentially, if one considers a predator-prey model with these species (crab-mussel-epibiont), there are two limiting cases
\begin{itemize} 
\item There is no epibiont ($e=0$), in this setting $\frac{e_{1}}{h_{1}} > \frac{e_{2}}{h_{2}}$, so $u_{1} = 1, u_{2} = 0$;
\item The epibiont achieves its carrying capacity $K$, that is, $e=K$, in this setting, high epibiont density causes the handling time $h_{1} \gg 1$, thus $\frac{e_{1}}{h_{1}} < \frac{e_{2}}{h_{2}}$, and so the crab switches to juveniles, and now $u_{1} = 0, u_{2} = 1$.
\end{itemize} 
To this end we first split the mussel class into adults and juveniles (we assume the juveniles have the protective periostracum whereas the adults do not).
A crab species preying on two separate classes of adult and juveniles mussels (with adults being the preferred food type), places this in a classic one predator-two prey setting \cite{K96}. Our hypotheses for OFT are as follows:
\begin{enumerate}
\item In the absence of epibionts, crab-mussel interactions follow classical OFT. That is, adult mussels will be attacked with rate 1, whilst the less preferred juveniles will not be attacked ($u_{1}=1, u_{2}=0$). We claim this is the only optimal strategy for the crab, as long as $d_{1} > \frac{e_{2}}{h_{2}}$, where $d_1$ is the death rate of prey type 1.

\item There is a change to the classical case, under pressure of epibiosis from \emph{D. vexillum}.

\item If a crab were presented with a preferred adult mussel overgrown with \emph{D. vexillum}, it would switch to a prey of a less optimal size, even if the overall adult mussel density was high (assuming their was uniform overgrowing of all adult mussels).

\item The switch would be to juvenile mussels, which we know are almost never overgrown because of their intact protective peristrocum. That is in the $(e=K)$ case, we have ($u_{1}=0, u_{2}=1$). We claim this is the only optimal strategy for the crab, as long as $d_{1} > \frac{e_{1}}{h_{1}}$.  

\item This will in turn directly affect the feedback loop to the adult mussel population, given that juveniles are transitioning to adults.
\end{enumerate}
The above is rigorously shown in appendix section \ref{AP6}.
\begin{figure}[H]
 \includegraphics[scale=0.16]{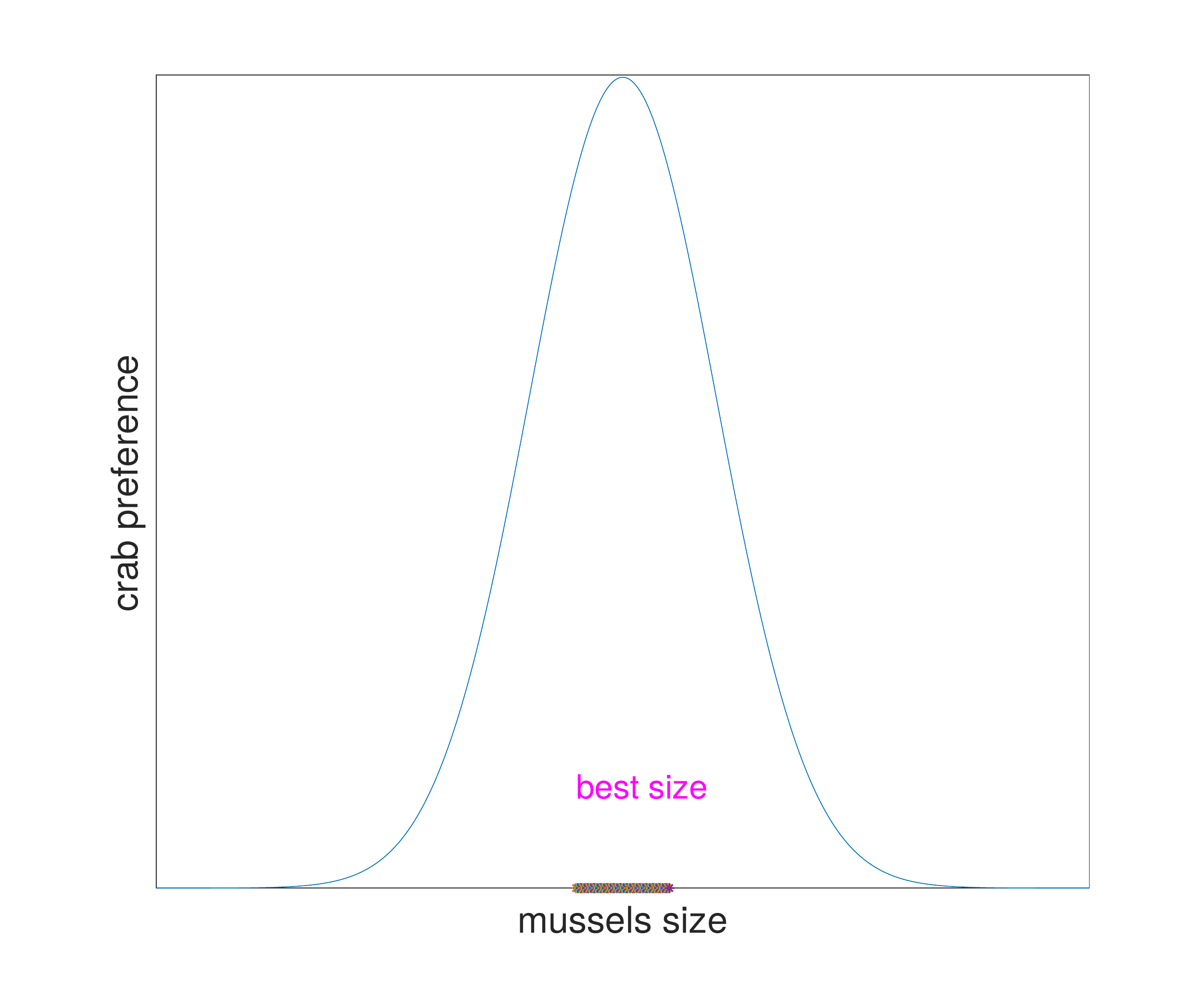}
           \includegraphics[scale=0.35]{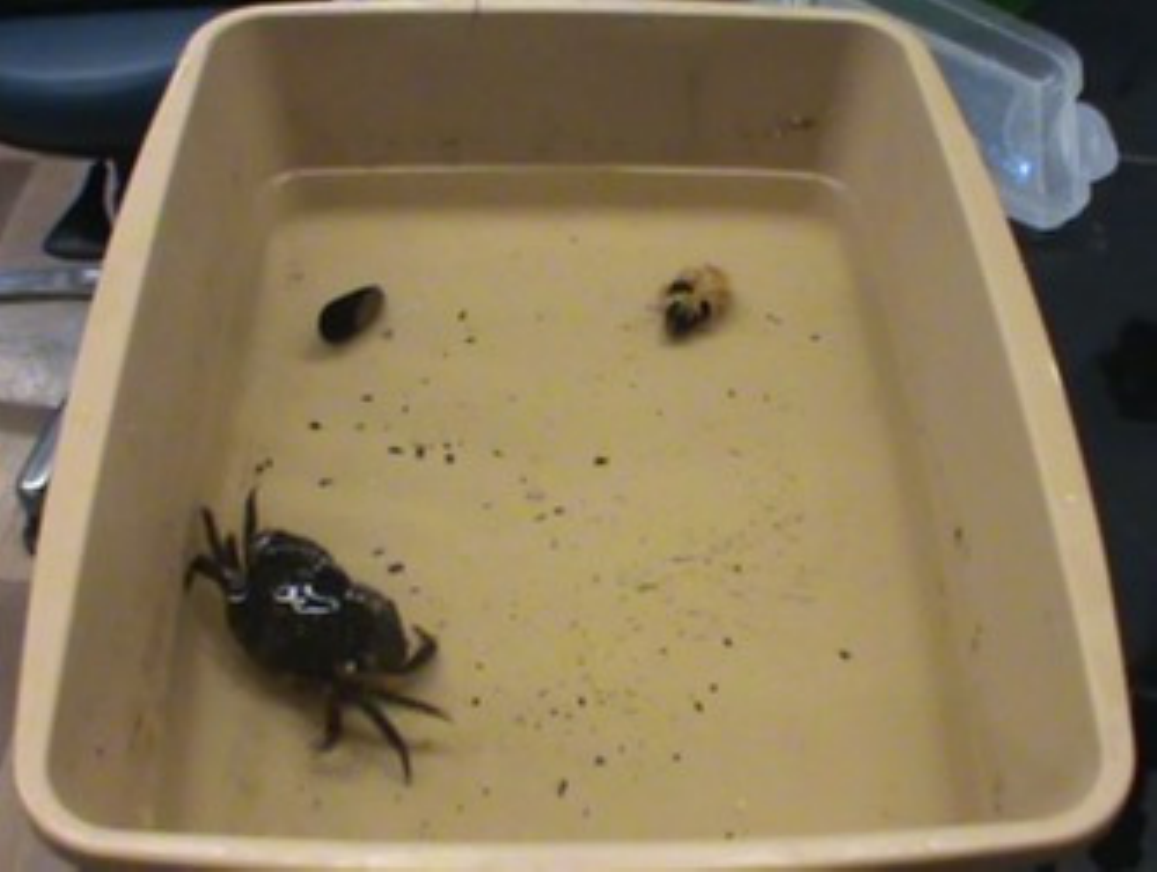}
            \includegraphics[scale=0.16]{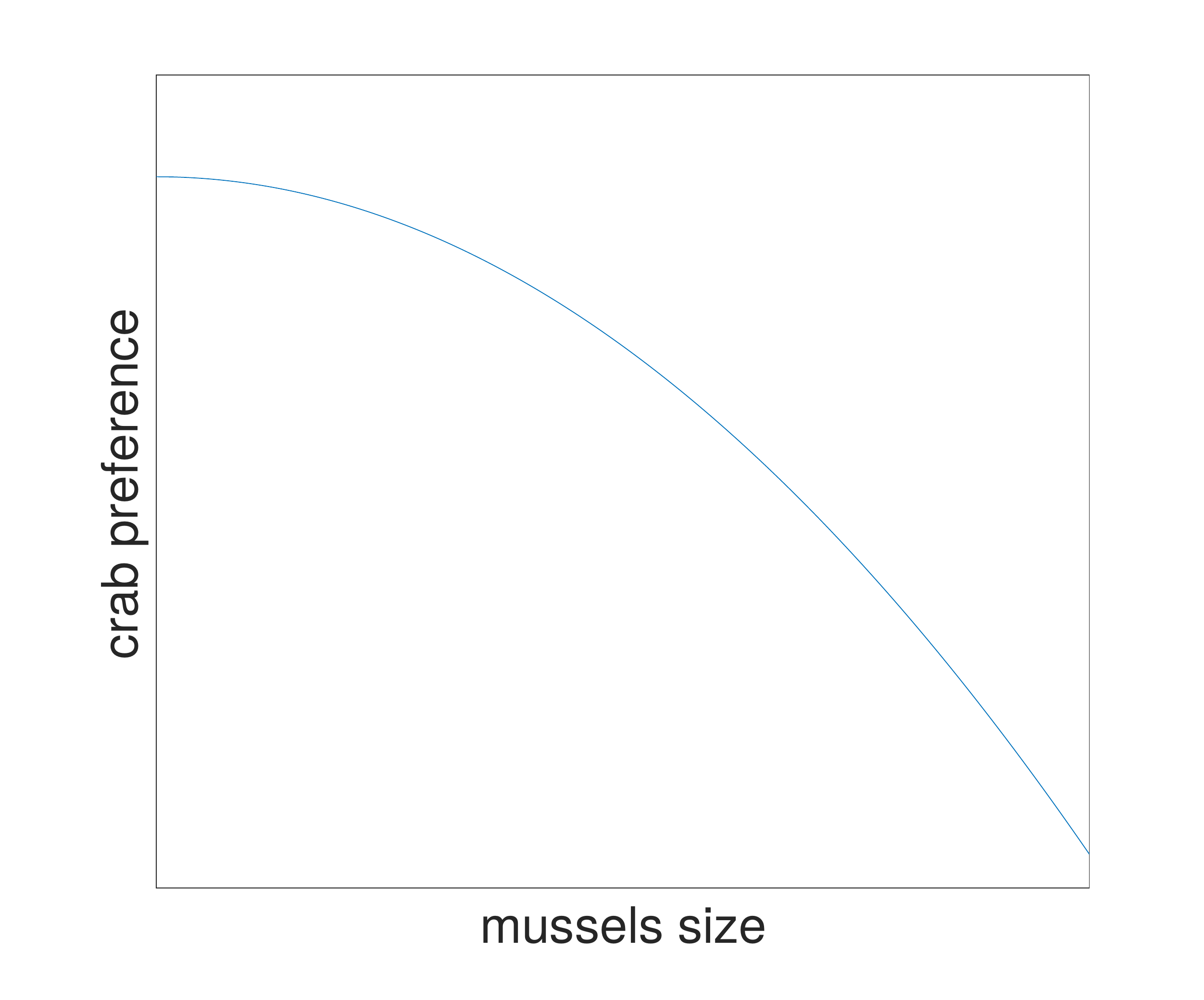}
 \caption{ Classic crab preference (without epibiont) shown in left panel. Crabs were presented with clean and overgrown mussels \cite{A10, A14}, and the handling times of the overgrown ones were up to 6 times greater than the clean ones (seen in center panel). These experimental results \cite{A10, A14} lead us to conjecture that crabs will \emph{switch} to the smaller uncovered juvenile mussels under pressure of epibiosis, shown in right panel.}
   \label{fig2}  
\end{figure} 
\section{Mathematical Formulation}
Our goal is to derive a mathematical model that best captures our hypotheses.
To this end, we make the following assumptions,
 
 \begin{enumerate}
\item An epibiont has invaded into a local crab-mussel community, and is growing logistically. It will eventually reach a critical carrying capacity.

\item We model the pressure from epibiosis, in terms of the attack rates $u_{1}, u_{2}$. That is we assume these are \emph{dependent on the epibiont density}.  As a simple first approach we assume 

\begin{equation*}
u_{1}(e)=  \frac{K-e}{K}, \ u_{2}(e) =  \frac{e}{K}.
\end{equation*}
Thus, without any epibiont presence ($e=0$), the adult mussel is the only prey eaten  and the juvenile is not eaten at all because there is not enough energy gain for the effort involved, so $u_{1}=1, u_{2}=0$, in line with classical theory \cite{K99}. However, this starts to change as the epibiont starts to overgrow the mussels. When the epibiont is at carrying capacity, $e=K$, we assume the adult mussels are completely overgrown, and thus is not consumed at all. The crab switches completely to juveniles, so that $u_{1}=0, u_{2}=1$.

\item  We model the decreasing fecundity in mussels due to epibiont cover by considering a growth rate $a=a(e)$. We consider $a=a(e)=a\left(\frac{K-\frac{e}{2}}{K}\right)$. Hence as epibionts get to carrying capacity $e=K$, this growth rate is cut in half and becomes $a/2$.

\item We assume the intraspecies competition is present only in adult mussels, and not juveniles \cite{L16}.

\item We assume the search rates $\lambda_{1}, \lambda_{2}$ to be the same, and normalized to 1, so $\lambda_{1}=\lambda_{2}=1$.

\end{enumerate}

Based on the above assumptions, we have the following system of differential equations,

\begin{equation}
\label{eq:1a1nj}
\begin{split}
\frac{dC}{dt}= & -d_{1}C + e_{1}u_{1}(e)\frac{ M_{A}}{1+h_{1}u_{1}(e) M_{A} + h_{2}u_{2}(e) M_{J}}C \\
& +e_{2}u_{2}(e)  \frac{ M_{J}}{1+h_{1}u_{1}(e) M_{A} + h_{2}u_{2}(e) M_{J}}C,
\end{split}
\end{equation}
\begin{equation}
\label{eq:2a1nj}
\frac{dM_{A}}{dt}=  bM_{J} - \delta_{1}M^{2}_{A} - u_{1}(e)\frac{ M_{A}}{1+h_{1}u_{1}(e) M_{A} + h_{2}u_{2}(e) M_{J}}C ,
\end{equation}

\begin{equation}
\label{eq:3a1nj}
\frac{dM_{J}}{dt}=  a(e) M_{A} - bM_{J} - u_{2}(e)  \frac{ M_{J}}{1+h_{1}u_{1}(e) M_{A} + h_{2}u_{2}(e) M_{J}}C,
\end{equation}

\begin{equation}
\label{eq:4a1nj}
\frac{de}{dt}=  b_{1}e(1-\frac{e}{K}).
\end{equation}

where
\begin{equation}
\label{eq:eee}
u_{1}(e)=  \frac{K-e}{K}, \ u_{2}(e)=   \frac{e}{K}, \  a(e)=a\left(\frac{K-\frac{e}{2}}{K}\right).
\end{equation}

\noindent
with positive initial conditions $C(0)=C_{0}, M_{A}(0)=M_{A0}, M_{J}(0)=M_{J0}, e(0)=e_{0}$. These responses are for the range $0 \leq e \leq K$. 

Here $C, M_{A}, M_{J}$ are the densities of crabs, adult mussels and juvenile mussels population at a given time $t$ respectively. The population density of \emph{D. vexillum} is $e$, while $d_{1}$ is the mortality rate of the crab, $e_{1}, e_{2}$ is the energy gain to the crab from preying on the adult mussel and juvenile mussel respectively, $h_{1}, h_{2}$ are the handling time of the adult mussel and juvenile mussel respectively, $b$ is the rate at which juveniles leave the juvenile class and become adults, $a$ is the rate at which juveniles are produced, $\delta_{1}$ measures the intraspecific competition among adult mussels, $b_1$ is the intrinsic rate of growth of the epibiont population, and  $K$ is its carrying capacity. 

\section{Dynamical Analysis}
\subsection{Boundedness}

The equation for $e$ is bounded trivially by $K$. Addition of \eqref{eq:1a1nj}-\eqref{eq:3a1nj}, and given the fact that $e_{1} < 1$ and $e_{2} < 1$, yields:

\begin{equation}
\frac{d(C+M_{A}+M_{J})}{dt} \leq -d_{1}C + a M_{A} - \delta_{1}M^{2}_{A} \leq a M_{A} - \delta_{1}M^{2}_{A}.
 \end{equation}
Thus, comparison with a logistic ODE yields: 

\begin{equation}
C+M_{A}+M_{J} \leq \frac{a}{\delta_{1}},
 \end{equation}
and thus we can state the following theorem:

\begin{theorem}
\label{thm:cmek}
Consider the crab-mussel system \eqref{eq:1a1nj}-\eqref{eq:4a1nj}. The solutions $ (C,M_{A},M_{J},e)$, satisfy the following uniform bounds
\begin{equation}
||C||_{\infty} \leq K_{1}, ||M_{A}||_{\infty} \leq K_{1}, ||M_{J}||_{\infty} \leq K_{1}, ||e||_{\infty} \leq K_{1},
\end{equation}
for any initial conditions $(C(0), M_{A}(0), M_{J}(0), e(0)) \in L^{\infty}$,
where 
\\
$K_{1} = \max (K, \frac{a}{\delta_{1}})$.
\end{theorem}

\subsection{Equilibrium and Local Stability with no epibiont}
\noindent
We now consider the existence and stability of the equilibrium for the system when there is no epibiont($e=0$). The system is simplified as

\begin{equation}
\label{eq:1a1n112}
\frac{dC}{dt}=  -d_{1}C + e_{1}\frac{ M_{A}}{1+h_{1} M_{A} }C,
\end{equation}
\begin{equation}
\label{eq:2a1n112}
\frac{dM_{A}}{dt}=  bM_{J} - \delta_{1}M^{2}_{A} - \frac{ M_{A}}{1+h_{1}M_{A}}C, 
\end{equation}

\begin{equation}
\label{eq:3a1n112}
\frac{dM_{J}}{dt}=  a M_{A} - bM_{J}.
\end{equation}

\noindent
 Two equilibria, $ (0,0,0)$ and $(0,\frac{a}{\delta_{1}}, \frac{a^{2}}{\delta_{1}b})$, on the boundary and one interior equilibrium $(C^{*},  M_{A} ^{*}, M_{J} ^{*})$.  It is easy to see $ (0,0,0)$ is unstable. And  $(0,\frac{a}{\delta_{1}}, \frac{a^{2}}{\delta_{1}b})$ is globally stable if $ 0<e_{1}-d_{1} h_{1}<\frac{d_{1}\delta_{1}}{a }$ and unstable if $ e_{1}-d_{1} h_{1}>\frac{d_{1}\delta_{1}}{a }$. The interior equilibrium is given by
\begin{equation}
\label{eql_c}
C^{*}=\frac{e_{1}(a(e_{1}-d_{1}h_{1})-d_{1}\delta_{1})}{(e_{1}-d_{1}h_{1})^{2}},
\end{equation}

\begin{equation}
\label{eql_a}
 M_{A} ^{*}=\frac{d_{1}}{e_{1}-d_{1}h_{1}},
\end{equation}

\begin{equation}
\label{eql_j}
M_{J}^{*}=\frac{a d_{1}}{b(e_{1}-d_{1}h_{1})}.
\end{equation}

\noindent
Note that $ C^{*}>0,  M_{A} ^{*}>0$ and$  M_{J} ^{*}>0 $ if 
\begin{equation}
\label{eq:feasibility2}
  e_{1}-d_{1} h_{1}>\frac{d_{1}\delta_{1}}{a}.
\end{equation}

We next state the following theorem

\begin{theorem}
\label{thm:cme0}
Consider the crab-mussel system  \eqref{eq:1a1nj}-\eqref{eq:4a1nj}, without the presence of an epibiont, that is when $e=0$. There exists an interior steady state $ (C^{*},M_{A}^{*},M_{J}^{*})$, which is locally asymptotically stable under the following criteria,

\begin{equation}
\label{no_epibiont_stability}
\frac{d_{1} \delta_{1}}{a }< e_{1}-d_{1} h_{1}< \frac{d_{1}\delta_{1}}{a }+\frac{d_{1} \delta_{1} e_{1}}{a h_{1} }, \ d_{1} > \frac{e_{2}}{h_{2}}.
\end{equation}

\end{theorem}

The proof is relegated to the appendix section \ref{AP2}. 

\begin{remark}
Note, the second condition $d_{1} > \frac{e_{2}}{h_{2}}$ is not a result of the Routh-Hurwitz criterion, rather it follows from lemma \ref{lem:sm1} in appendix section \ref{AP6}. We enforce it so that the attack rates should be as predicted via classical OFT.
\end{remark}

 \subsection{Equilibrim and Local Stability Analysis with Epibiont}
\label{3}
The system \eqref{eq:1a1nj}-\eqref{eq:4a1nj} has five possible equilibria. There is one in the interior of the positive octant $(C^{*},M_{A}^{*},M_{J}^{*},e^{*})$,
and four on the boundary, $(0,0,0,0), (0,0,0,K), (0,\frac{a}{\delta_{1}},\frac{a^{2}}{b\delta_{1}},0) $ and $ (0,\frac{a}{2\delta_{1}},\frac{a^{2}}{4b\delta_{1}},K)$. It is easy to check that the equilibria with no epiboint, $(0,0,0,0)$ and  $(0,\frac{a}{\delta_{1}},\frac{a^{2}}{b\delta_{1}},0)$, are unstable. Furthermore, $(0,\frac{a}{2\delta_{1}},\frac{a^{2}}{4b\delta_{1}},K) $ is stable if  $0< e_{2}-d_{1}h_{2}<\frac{4 b d_{1} \delta_{1}}{a^{2}}$  and unstable if $ e_{2}-d_{1}h_{2}>\frac{4 b d_{1} \delta_{1}}{a^{2}}$. In fact, it is common that prey exist in a stable state in the absence of the predator. Finally, $(0,0,0,K)$ is also unstable. 
%This is plausible as epibionts (at least in our setting) will not exist if there are no available basibionts (mussels).  
%The equilibria on the boundary are not of interest here, and 
We will focus on the interior equilibrium. 

Consider the interior equilibrium, i.e.  $ (C^{*},M_{A}^{*},M_{J}^{*},e^{*}).$ It is easy to see $ e^{*}=K$ in the equilibrium state.Then we have $ u_{1}(e^*)=0, u_{2}(e^*)=1, a(e^*)=\frac{a}{2}$. To get $ (C^{*},M_{A}^{*},M_{J}^{*},e^{*})$ explicitly, it is equivalent to solve the following equations:

\begin{equation}
-d_{1}C  + e_{2}  \frac{ M_{J}}{1+ h_{2} M_{J}}C=0,
\end{equation}
\begin{equation}
 bM_{J} - \delta_{1}M^{2}_{A}  =0,
\end{equation}

\begin{equation}
\frac{a}{2} M_{A} - bM_{J} -  \frac{ M_{J}}{1+ h_{2} M_{J}}C=0,
\end{equation}

\begin{equation}
e=K.
\end{equation}

Thus the interior equilibrium is given by 
\begin{equation}
C^{*}=\frac{1}{2} \frac{a e_{2} \sqrt{\frac{b d_{1}}{(e_{2} -d_{1} h_{2}) \delta_{1}}}}{d_{1}} - \frac{e_{2} b}{e_{2}-d_{1} h_{2}},
\end{equation}

\begin{equation}
 M_{A} ^{*}=\sqrt{\frac{b d_{1}}{(e_{2}-d_{1}h_{2})\delta_{1}}},
\end{equation}

\begin{equation}
M_{J}^{*}=\frac{d_{1}}{e_{2}-d_{1} h_{2}},
\end{equation}

\begin{equation}
e^{*}=K.
\end{equation}

\noindent
Note that $M_{A}^{*} >0 $ and $ M_{J}^{*}>0 $ if $e_{2}-d_{1} h_{2}>0. $ And $C^{*}>0 $ if $ e_{2}-d_{1} h_{2}> \frac{4 b d_{1} \delta_{1}}{a^{2}}$. Therefore, the feasibility criteria for this system is 

\begin{equation}
\label{feasibility_epibiont}
 e_{2}-d_{1} h_{2} > \frac{4 b d_{1} \delta_{1}}{a^{2}}.
\end{equation}

We next state the following theorem

\begin{theorem}
\label{thm:cmek}
Consider the crab-mussel system \eqref{eq:1a1nj}-\eqref{eq:4a1nj}, when the epibiont has reached equilibrium, that is $e=K$. There exists an interior steady state $ (C^{*},M_{A}^{*},M_{J}^{*},K)$, which is locally asymptotically stable under the following criteria,
\begin{equation}
\quad \frac{4bd_{1}\delta_{1}}{a^{2}}<e_{2}-d_{1}h_{2}<\frac{16 b d_{1} \delta_{1}}{a^{2}}, \ d_{1} > \frac{e_{1}}{h_{1}}.
\end{equation}
\end{theorem}

The proof relies on the Routh-Hurwitz criterion \cite{PL13}, and is relegated to the appendix section \ref{AP1}.

\begin{remark}
Note that the condition $d_{1} > \frac{e_{1}}{h_{1}}$ is not a result of the Routh-Hurwitz criterion, rather it follows from lemma \ref{lem:sm1} in appendix section \ref{AP6}. In a sense we enforce it so that the attack rates should be as predicted via classical OFT.
\end{remark}

\subsection{Global stability}

We now derive some results on the global stability of the internal equilibrium.
\begin{theorem}
\label{thm:global}
Consider the model \eqref{eq:1a1n112}-\eqref{eq:3a1n112}. There exists an $\epsilon > 0$, s.t. the internal equilbrium point, $(C^{*}, M_A^{*}, M_{J}^{*})$, is globally asymptotically stable under the following parametric restriction 
\begin{equation}
\frac{d_{1} \delta_{1}}{a}<e_{1}-d_{1}h_{1}<\frac{d_{1}}{\epsilon}, \quad 0<\epsilon < \frac{a}{\delta_{1}},\quad  \frac{1}{2}<e_{1} < 1.
\end{equation}
\end{theorem}
\begin{proof}
Consider $M_{A}=\overline{M_{A}}+\epsilon$, Under this transformation we have the following transformed system 
\begin{equation}
\label{new1}
\frac{dC}{dt}=-d_{1}C  + e_{1}  \frac{ \overline{M_{A}}+\epsilon}{1+ h_{1} (\overline{M_{A}}+\epsilon)}C,
\end{equation}

\begin{equation}
\label{new2}
\frac{d\overline{M_{A}}}{dt}= bM_{J} - \delta_{1}(\overline{M_{A}}+\epsilon)^{2}-\frac{\overline{M_{A}}+\epsilon}{1+h_{1}(\overline{M_{A}}+\epsilon)}C,
\end{equation}

\begin{equation}
\label{new3}
\frac{dM_{J}}{dt}=a (\overline{M_{A}}+\epsilon)- bM_{J}.
\end{equation}
It is enough to show the new system \eqref{new1}-\eqref{new3} is globally asymptotically stable. The equilibrium of \eqref{new1}-\eqref{new3} is given by 
\begin{equation}
\begin{split}
& C^{*}=\frac{e_{1} a}{d_{1}} (\overline{M_{A} }^{*}+\epsilon)-\frac{e_{1}\delta_{1}}{d_{1}}(\overline{M_{A} }^{*}+\epsilon)^{2},\\
& \overline{M_{A} }^{*}=\frac{d_{1}}{e_{1}-d_{1}h_{1}}-\epsilon,\\
& M_{J}^{*}=\frac{a d_{1}}{b(e_{1}-d_{1}h_{1})}.
\end{split}
\end{equation}
Note that solutions to this new system are feasible if 
\begin{equation}
\frac{d_{1}\delta_{1}}{a}<e_{1}-d_{1} h_{1}<\frac{d_{1}}{\epsilon}, \quad 0<\epsilon<\frac{a}{\delta_{1}}.
\end{equation}
We define the following Lyapunov function,
\begin{equation}
V(C,\overline{M_{A}^{*}},M_{J})=C+\overline{M_{A}^{*}}+M_{J}.
\end{equation}
Note that $V \geq 0$, because of the positivity of the solutions. Furthermore, $V$ is radially unbounded.  Now consider 
\begin{equation}
\begin{split}
\frac{dV}{dt}&=\frac{dC}{dt}+\frac{d\overline{M_{A}}}{dt}+\frac{dM_{J}}{dt}\\
& = -d_{1}C+ (e_{1} -1) \frac{ \overline{M_{A}}+\epsilon}{1+ h_{1} (\overline{M_{A}}+\epsilon)}C-\delta_{1}(\overline{M_{A}}+\epsilon)^{2}+a (\overline{M_{A}}+\epsilon)\\
& <-d_{1}C-\delta_{1}(\overline{M_{A}}+\epsilon)^{2}+a (\overline{M_{A}}+\epsilon)\\
& =-d_{1}C-\delta_{1} (\overline{M_{A}})^{2}-2\delta_{1}\epsilon \overline{M_{A}}-\delta_{1} \epsilon^{2}+a \overline{M_{A}} +a\epsilon\\
&=-d_{1} C -\delta_{1} (\overline{M_{A}}-\frac{a}{2\delta_{1}})^{2}-\delta_{1}(\epsilon-\frac{a}{2 \delta_{1}})^{2}+\frac{a^{2}}{2\delta_{1}}-2\delta_{1} \epsilon \overline{M_{A} }\\
&=-\delta_{1} [(\overline{M_{A}}-\frac{a}{2\delta_{1}})^{2}+(\epsilon-\frac{a}{2 \delta_{1}})^{2}]+\frac{a^{2}}{2\delta_{1}}-2\delta_{1} \epsilon \overline{M_{A} }- d_{1}C.
\end{split}
\end{equation}
We hope $\frac{a^{2}}{2\delta_{1}}<\delta_{1} [(\overline{M_{A}}-\frac{a}{2\delta_{1}})^{2}+(\epsilon-\frac{a}{2 \delta_{1}})^{2}]+2\delta_{1} \epsilon \overline{M_{A} }+d_{1}C$. Since $2[(\overline{M_{A}}-\frac{a}{2\delta_{1}})^{2}+(\epsilon-\frac{a}{2 \delta_{1}})^{2}] \geq (\overline{M_{A}}-\epsilon)^{2}$, it is enough to show 
 
\begin{equation}
\begin{split}
\Longrightarrow \frac{a^{2}}{\delta_{1}^{2}} &< (\overline{M_{A}}-\epsilon)^{2}+2\epsilon \overline{M_{A} }+\frac{2d_{1}C}{\delta_{1}}\\
&=(\overline{M_{A}})^{2}+\epsilon^{2}+\frac{2d_{1}C}{\delta_{1}}.
\end{split}
\end{equation}
Since $\overline{M_{A} }$ is bounded by $\frac{a}{\delta_{1}}-\epsilon$ from the feasibility conditions, we show 
\begin{equation}
\Longrightarrow \frac{a^{2}}{\delta_{1}^{2}}-\epsilon^{2}<(\frac{a}{\delta_{1}}-\epsilon)^{2}+\frac{2d_{1}C}{\delta_1},
\end{equation}

\begin{equation}
\Longrightarrow \frac{a}{\delta_{1}} <\epsilon+\frac{d_{1}C}{ \epsilon \delta_{1}}.
\end{equation}
Since $\epsilon+\frac{d_{1}C}{\epsilon \delta_{1}} \geq 2\sqrt{\frac{d_{1}C}{\delta_{1}}}$, it is enough to show $\frac{a}{\delta_{1}} <2\sqrt{\frac{d_{1}C}{\delta_{1}}}$. Due to $ C^{*}=\frac{e_{1} a}{d_{1}} (\overline{M_{A} }^{*}+\epsilon)-\frac{e_{1}\delta_{1}}{d_{1}}(\overline{M_{A} }^{*}+\epsilon)^{2}=\frac{e_{1} a}{d_{1}} M_{A} ^{*}-\frac{e_{1}\delta_{1}}{d_{1}} (M_{A} ^{*})^{2}$, it is enough to show 

\begin{equation}
\begin{split}
\frac{a^2}{\delta_{1}} &< 4d_{1}C=4 e_{1} a M_{A}-4e_{1}\delta_{1}M_{A}^{2}\\
&=-4e_{1}\delta_{1} (M_{A}-\frac{a}{2\delta_{1}})^{2}+\frac{2e_{1}a^{2}}{\delta_{1}}.
\end{split}
\end{equation}
Therefore we only need to show 

\begin{equation}
\frac{a^2}{\delta_{1}}  < max(-4e_{1}\delta_{1} (M_{A}-\frac{a}{2\delta_{1}})^{2}+\frac{2e_{1}a^{2}}{\delta_{1}}) = \frac{2e_{1}a^{2}}{\delta_{1}},
\end{equation}
and this requires we have $ \frac{1}{2} < e_{1} < 1$.Then the system \eqref{eq:1a1n112}-\eqref{eq:3a1n112} is globally stable if 
\begin{equation}
\frac{d_{1} \delta_{1}}{a}<e_{1}-d_{1}h_{1}<\frac{d_{1}}{\epsilon}, \quad 0 < \epsilon < \frac{a}{\delta_{1}},\quad  \frac{1}{2} < e_{1} < 1.
\end{equation}
\end{proof}

\begin{figure}[H]
        \begin{minipage}[b]{0.48\linewidth}
            \centering
            \includegraphics[width=\textwidth]{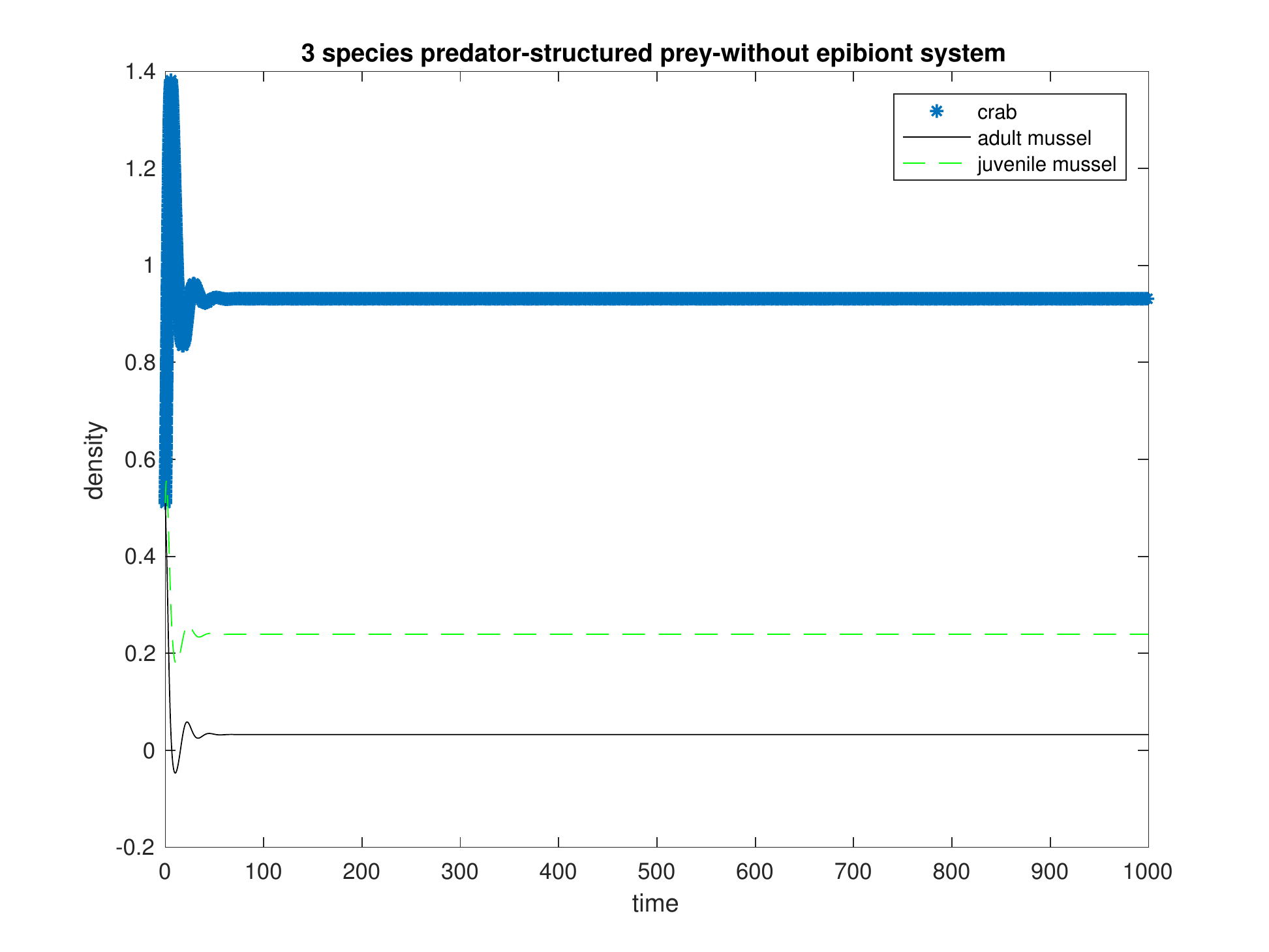}
            \subcaption{ }
             \label{fig:fig1_1}
        \end{minipage}
        \hspace{0.10cm}
        \begin{minipage}[b]{0.48\linewidth}
            \centering
            \includegraphics[width=\textwidth]{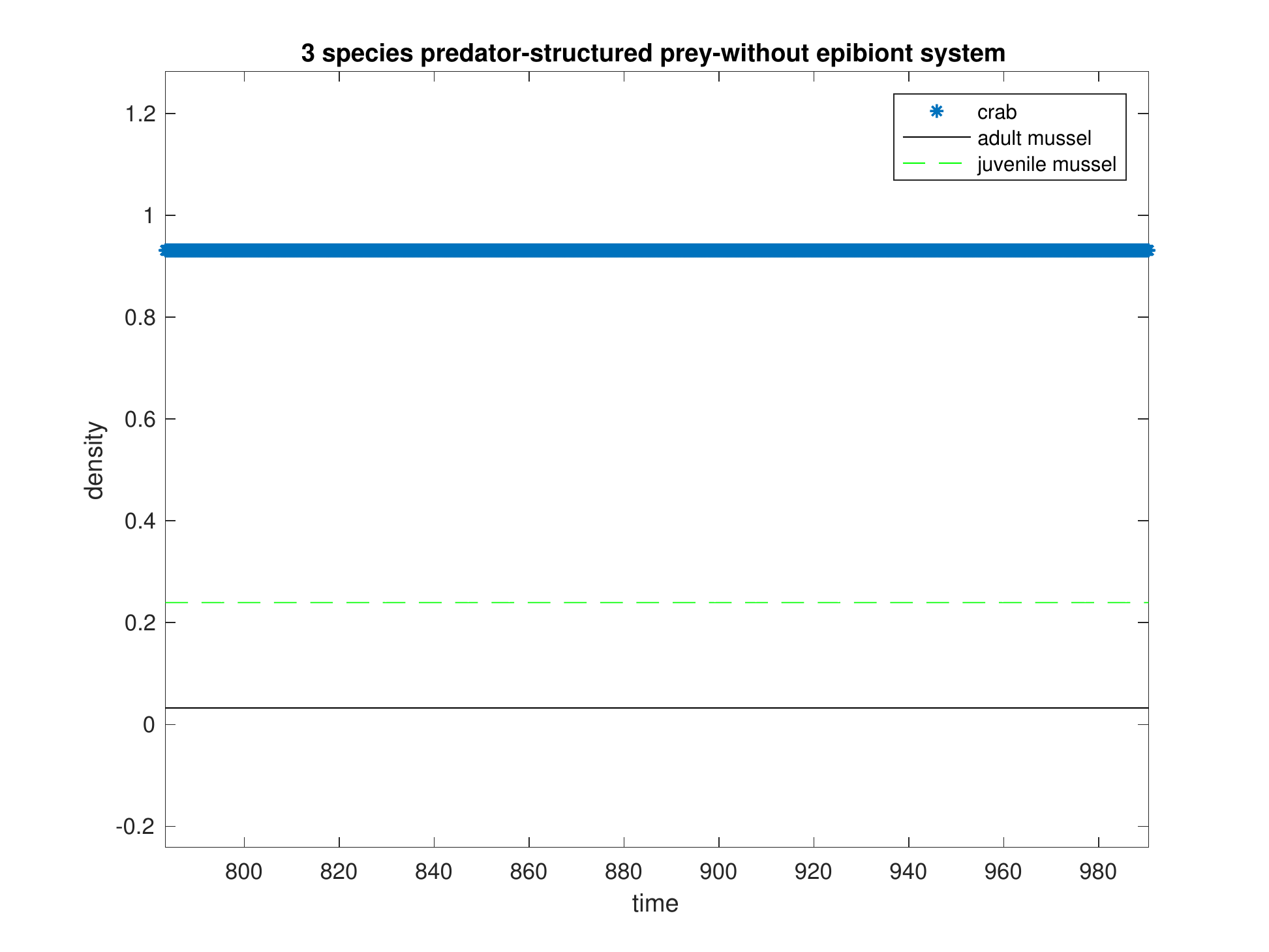}
            \subcaption{ }
            \label{fig:fig1_2}
        \end{minipage}
             \label{fig1}
\caption{ The above figures verify theorem \ref{thm:global}. We consider the parameters $e_1=0.9, e_2=0.01, d_1=0.2,b=1,h_1=0.2, h_2=0.1, \delta_{1}=0.6,a=0.8,\epsilon=0.2.$(A) The initial condition $(C_{0}, M_{A0},M_{J0})$=(0.5 0.5 0.5) (B) The initial condition $(C_{0}, M_{A0},M_{J0})$=(500 500 500) (zooming in time scale). They both reach a stable level (0.9312    0.0326    0.2394).}
\end{figure} 

\begin{remark}
Note, although we prove global stability (under certain parametric restrictions) for the case without epibiont $(e=0)$, it is easily proven using the same approach as above for the $(e=K)$ case by just replacing $M_A=\overline{M_A}+\epsilon$ and defining a new Lyapunov function as $V(C,\overline{M_A},M_J,e)=C+\overline{M_A}+M_J+e$.
\end{remark}

\section{Hopf Bifurcation}
Now we will investigate the Hopf bifurcation for the system in terms of parameter $ a $. In this paper, we will follow the method developed by Liu \cite{L94}. 
Firstly, let us consider the system  \eqref{eq:1a1nj}-\eqref{eq:3a1nj}, without the presence of an epibiont ($e=0$), that is when $e=0$. The Hopf bifurcation at $ a=a_{*} $ can occur if $ A_{2}(a_{*}),  A_{0}(a_{*}),$  and $ \phi(a_{*})=A_{2}(a_{*})A_{1}(a_{*})-A_{0}(a_{*})$ are smooth functions of $a $ in an open interval of $ a_{*} \in 
\mathbf{ R} $ such that:
\begin{enumerate} 

\item  $ A_{1}(a_{*})>0,  A_{0}(a_{*})>0,$ and $ \phi(a_{*})=A_{2}(a_{*})A_{1}(a_{*})-A_{0}(a_{*})=0$.

\item $\frac{d\phi(a)}{da}|_{a=a_{*}} \not=0$.

\end{enumerate}

We check the above, in appendix section \ref{AP3}, to state the following theorem

\begin{theorem}
\label{thm:he0}
Under the condition \eqref{eq:feasibility2}, there is a simple Hopf bifurcation of the positive equilibrium point $(C^{*}, M_{A}^{*},M_{J}^{*})$ of model system \eqref{eq:1a1n112}-\eqref{eq:3a1n112} at some critical value of parameter $a_{*}$ given by  \eqref{eq:hf1} and  \eqref{eq:hf2}.
\end{theorem}
%Now we will numerically demonstrate the Hopf bifurcation. The following parameter set will be used: $d_{1}=0.1,e_{1}=0.9,e_{2}=0.5, h_{1}=1.0,h_{2}=0.2, b=0.5,  \delta_{1}=0.1.$
%\begin{figure}[H]
%           \includegraphics[scale=0.3]{hopf_a.eps}
%            \label{fig:fig2}  
%            \caption{Here we see a plot of $\phi(a)$ versus $a$ for the other parameters as above. We clearly see that $\frac{d\phi(a)}{da}|_{a=0.237}=0$. The root occurs at approximately $a=0.237$ or $a=0.513$.}
%\end{figure} 

\begin{figure}[H]
        \begin{minipage}[b]{0.480\linewidth}
            \centering
            \includegraphics[width=\textwidth]{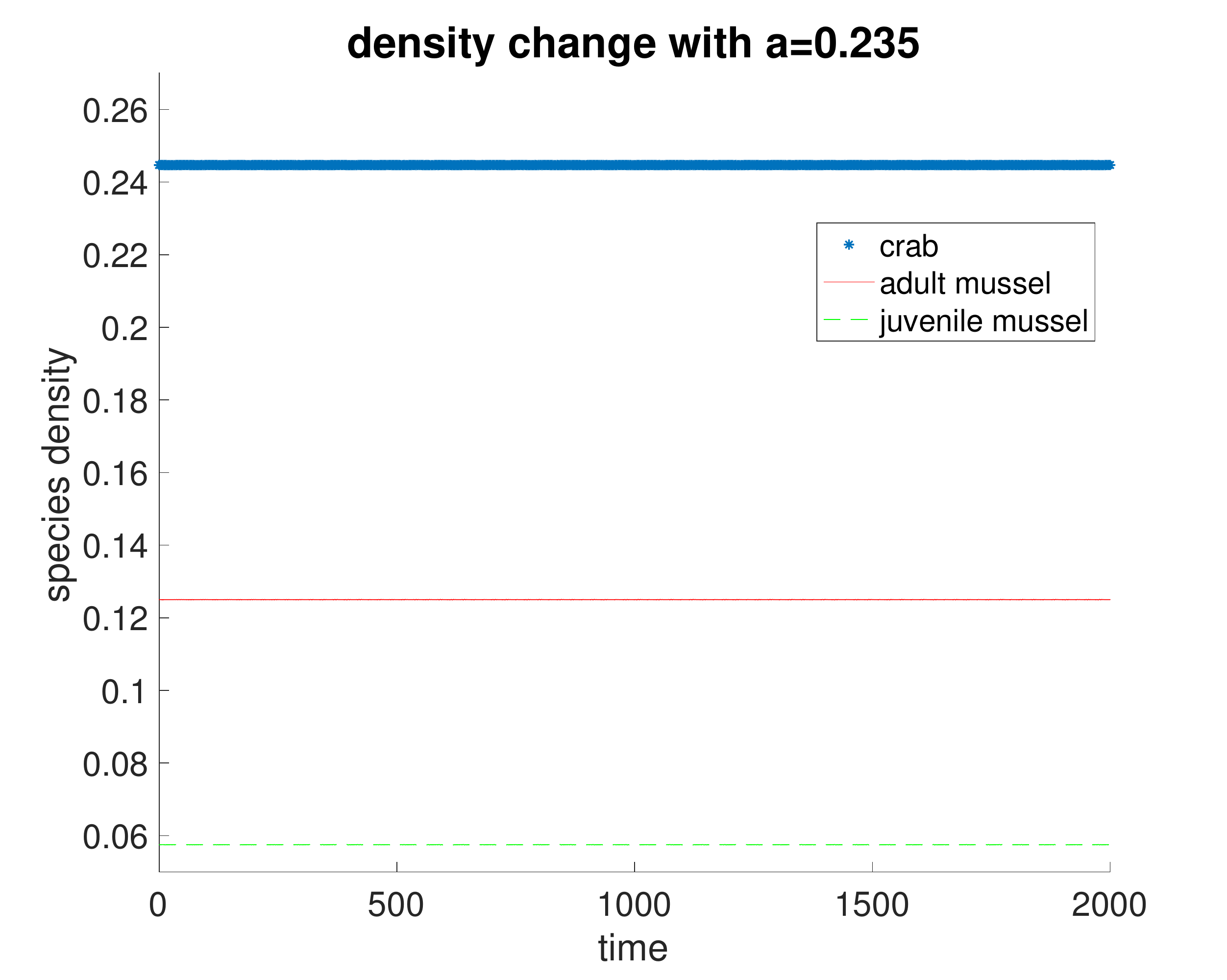}
            \subcaption{ }
             \label{fig:fig3_1}
        \end{minipage}
        \hspace{0.10cm}
        \begin{minipage}[b]{0.480\linewidth}
            \centering
            \includegraphics[width=\textwidth]{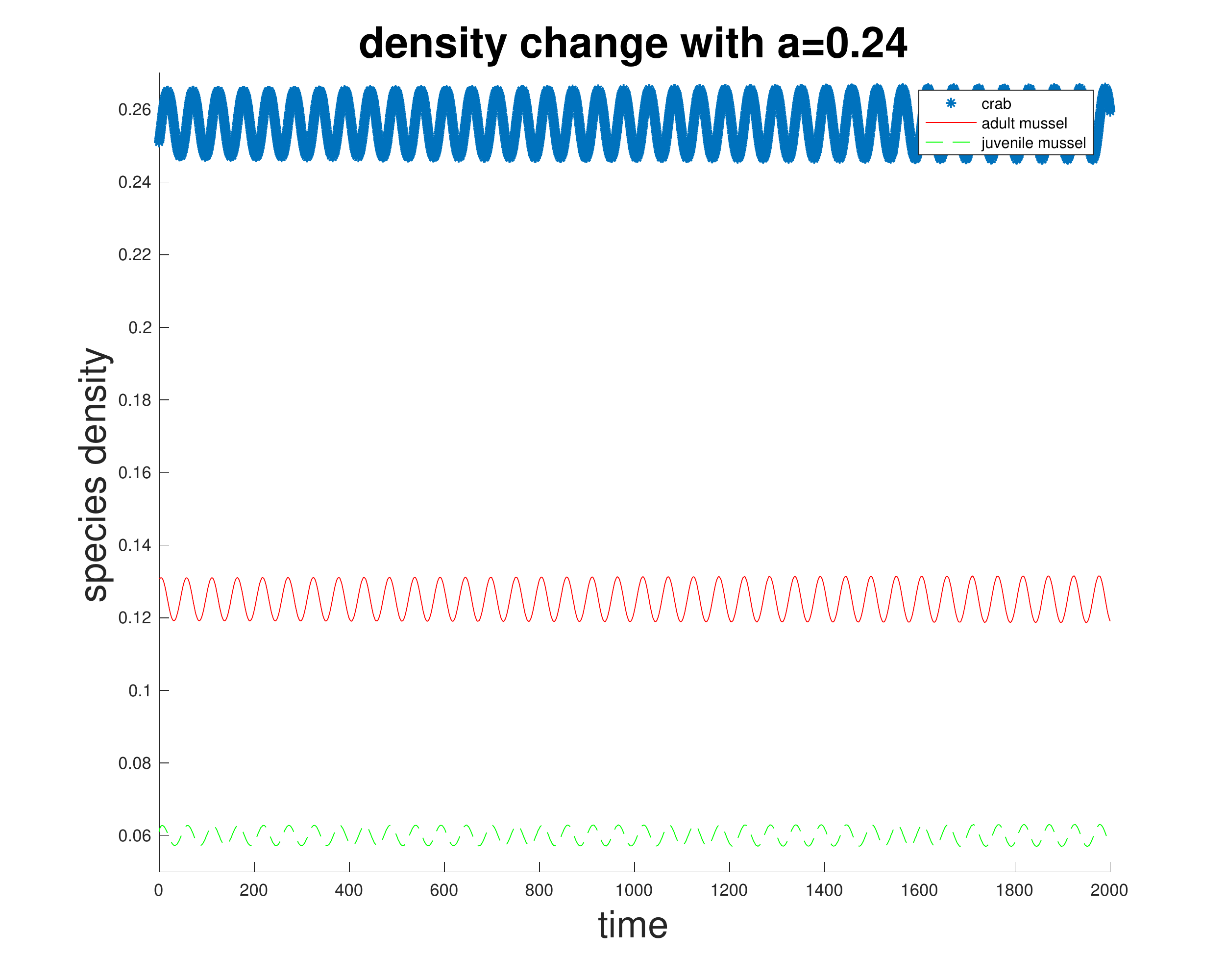}
            \subcaption{ }
            \label{fig:fig3_2}
        \end{minipage}
 \label{fig3}
\caption{Here we demonstrate the species density change with time. We see in $(A)$, the population of the species are stable when $a=0.235$, while in $(B)$ occurence of a Hopf bifurcation has lead to population cycles.}
\end{figure}

\section{Optimal Control}

In this section our goal is to investigate mechanisms in our crab-mussel-epibiont system, that, if controlled, could lead to optimal levels of crab or mussel densities.
 We assume that the attack rates $u_{1}, u_{2}$ are not known \emph{a priori} and enter the system as time-dependent controls. They no longer depend on the epibiont density. Instead we assume that the handling time depends on the epibiont density $e$ in the following way,                                                                                                                                                                                                                                                                                                                                                                                                                                                                                                                                                                                                                                                                                                                                                                                                                                                                                                                                                                                                                                                                                                                                                                                                                                                                                                                                                                                                                                                                                                                                                                                                                                                                                                                                              

where
\begin{equation}
\label{eq:eej}
h_{1}(e)=  1+  \frac{e}{K}, \  a(e)=a\left(\frac{K-\frac{e}{2}}{K}\right).
\end{equation}
These responses are for the range $0 \leq e \leq K$. Increase in epibiont density still negatively effects mussel fecundity and the handling time for adult mussels increases with increasing epibiont density.

This has a twofold advantage. We can visualise the system from the crab's point of view. That is, the crab can "optimally" control its attack rate, to reach the best possible population density. Also we can visualise the system from the mussel's point of view. That is, the mussel can induce defenses or other mechanisms, that would alter the attack rate of the crab, thus enabling the mussel population density to reach optimum levels. Our model takes the following form,

\begin{equation}
\label{eq:1a1njjo}
\begin{split}
C'= & -d_{1}C + e_{1}u_{1}(t)\frac{ M_{A}}{1+h_{1}(e)u_{1}(t) M_{A} + h_{2}u_{2}(t) M_{J}}C  \\
& +e_{2}u_{2}(t)  \frac{ M_{J}}{1+h_{1}(e)u_{1}(t) M_{A} + h_{2}u_{2}(t) M_{J}}C,
\end{split}
\end{equation}
\begin{equation}
\label{eq:2a1njjo}
M_{A}'=  bM_{J} - \delta_{1}M^{2}_{A} - u_{1}(t)\frac{ M_{A}}{1+h_{1}(e)u_{1}(t) M_{A} + h_{2}u_{2}(t) M_{J}}C,
\end{equation}

\begin{equation}
\label{eq:3a1jjno}
M_{J}'=  a(e) M_{A} - bM_{J} - u_{2}(t)  \frac{ M_{J}}{1+h_{1}(e)u_{1}(t) M_{A} + h_{2}u_{2}(t) M_{J}}C,
\end{equation}

\begin{equation}
\label{eq:4a1njjo}
e'=  b_{1}e(1-\frac{e}{K}).
\end{equation}

 We next derive optimal strategies for three objective functions, where we maximize both crab and mussel populations. 
To simplify the calculation, we will consider the case when $e=K$, which is when the epibiont achieves carrying capacity. In this case our system reduces to

\begin{equation}
\label{eq:1a1njj}
\begin{split}
C'=  & -d_{1}C + e_{1}u_{1}(t)\frac{ M_{A}}{1+2u_{1}(t) M_{A} + h_{2}u_{2}(t) M_{J}}C  \\
& + e_{2}u_{2}(t)  \frac{ M_{J}}{1+2u_{1}(t) M_{A} + h_{2}u_{2}(t) M_{J}}C,
\end{split}
\end{equation}
\begin{equation}
\label{eq:2a1njj}
M_{A}'=  bM_{J} - \delta_{1}M^{2}_{A} - u_{1}(t)\frac{ M_{A}}{1+2u_{1}(t) M_{A} + h_{2}u_{2}(t) M_{J}}C, 
\end{equation}

\begin{equation}
\label{eq:3a1jjn}
M_{J}'=  \frac{a}{2} M_{A} - bM_{J} - u_{2}(t)  \frac{ M_{J}}{1+2u_{1}(t) M_{A} + h_{2}u_{2}(t) M_{J}}C.
\end{equation}

\subsection{Maximizing crab denisty w.r.t. attacking rates}
To maximize density of the crab, the density of juvenile mussels (thus leading to more adult mussels, its favored food) should be maximized. Crab attack rates should be miminized on the juvenile mussels, as they are less preferred by the crab.
Thus we choose the following objective functional,
\begin{equation}
  J_{1}(u_{1},u_{2})= \int^{T}_{0} (C+M_{J} - \frac{1}{2}u_{2}^{2} ) dt,
\end{equation}

\begin{center}
$s.t. $ \eqref{eq:1a1njj}- \eqref{eq:3a1jjn} and $C(t_{0})=C_{0}, M_{A}(t_{0})=M_{A_{0}}, M_{J}(t_{0})=M_{J_{0}}$. 
\end{center}
and we search for the optimal controls in the set $U$ where
\begin{equation}
 U= \{ (u_{1},u_{2})|u_{i} \ \mbox{measurable}, \  0 \leq u_{1} \leq 1, 0 \leq u_{2} \leq 1, \  t \in [0,T], \ \forall T\}.
\end{equation}
The goal is to seek an optimal $(u^{*}_{1}, u^{*}_{2})$ s.t.,

\begin{equation}
  J_{1}(u^{*}_{1}, u^{*}_{2})= \underset{(u_{1},u_{2})}{\max} \int^{T}_{0} (C+M_{J} - \frac{1}{2}u_{2}^{2})  dt.
\end{equation}

We can state the following existence theorem,

\begin{theorem}
\label{thm:oc1}
Consider the optimal control problem \eqref{eq:1a1njj}-\eqref{eq:3a1jjn}. There exists $(u^{*}_{1}, u^{*}_{2}) \in U$ s.t.

\begin{equation}
  J_{1}(u^{*}_{1}, u^{*}_{2})= \underset{(u_{1},u_{2}) \in U}{\max} \int^{T}_{0} (C+M_{J} - \frac{1}{2}u_{2}^{2})  dt.
\end{equation}
\end{theorem}

\begin{proof}
The compactness (closed and bounded in ODE case) of the functional $J_{1}$ follows from the global boundedness of the state variables via theorem \ref{thm:cmek}, and the boundedness assumption on the controls. Also the functional $J_{1}$ is concave in the argument $u_{2}$. This is easily verified via standard application \cite{LJ07}. These in conjunction give the existence of an optimal control via application of classical one predator-two prey theory \cite{FR75}.
\end{proof}

 In order to derive necessary conditions on the optimal control, we use Pontryagin's maximum principle (PMP). 
The Hamiltonian for our problem is given by
\begin{equation}
H=C+M_{J} - \frac{1}{2}u_{2}^2 +\lambda_{1}  C' +\lambda_{2} M_{A}' +\lambda_{3} M_{J}'.
\end{equation}

We use the Hamiltonian to find a differential equation of the adjoint $\lambda_{i}, i=1,2,3$.
\begin{equation}
\begin{split}
\lambda_{1}'(t)=&-\lambda_{{1}} \left( -d_{{1}}+{\frac {M_{{A}}e_{{1}}u_{{1}}+M_{{J}}
e_{{2}}u_{{2}}}{M_{{A}}h_{{1}}u_{{1}}+M_{{J}}h_{{2}}u_{{2}}+1}}
 \right) \\
& +{\frac {\lambda_{{2}}u_{{1}}M_{{A}}}{M_{{A}}h_{{1}}u_{{1}}+M
_{{J}}h_{{2}}u_{{2}}+1}}+ {\frac {\lambda_{{3}}u_{{2}}M_{{J}}}{M_{{A}}h
_{{1}}u_{{1}}+M_{{J}}h_{{2}}u_{{2}}+1}}-1,\\
\lambda_{2}'(t)=&-\lambda_{{1}} \left( {\frac {e_{{1}}u_{{1}}C}{M_{{A}}h_{{1}}u_{{1}}+M
_{{J}}h_{{2}}u_{{2}}+1}}-{\frac { \left( M_{{A}}e_{{1}}u_{{1}}+M_{{J}}
e_{{2}}u_{{2}} \right) Ch_{{1}}u_{{1}}}{ \left( M_{{A}}h_{{1}}u_{{1}}+
M_{{J}}h_{{2}}u_{{2}}+1 \right) ^{2}}} \right) \\ 
&- \lambda_{{2}} \left( -
2\,\delta_{{1}}M_{{A}}-{\frac {u_{{1}}C}{M_{{A}}h_{{1}}u_{{1}}+M_{{J}}
h_{{2}}u_{{2}}+1}}+{\frac {{u_{{1}}}^{2}M_{{A}}Ch_{{1}}}{ \left( M_{{A
}}h_{{1}}u_{{1}}+M_{{J}}h_{{2}}u_{{2}}+1 \right) ^{2}}} \right) \\
& -\lambda_{{3}} \left( \frac{a}{2}+{\frac {u_{{2}}M_{{J}}Ch_{{1}}u_{{1}}}{
 \left( M_{{A}}h_{{1}}u_{{1}}+M_{{J}}h_{{2}}u_{{2}}+1 \right) ^{2}}}
 \right), \\
 \lambda_{3}'(t)=&- \lambda_{{1}} \left( {\frac {e_{{2}}u_{{2}}C}{M_{{A}}h_{{1}}u_{{1}}
+M_{{J}}h_{{2}}u_{{2}}+1}}-{\frac { \left( M_{{A}}e_{{1}}u_{{1}}+M_{{J
}}e_{{2}}u_{{2}} \right) Ch_{{2}}u_{{2}}}{ \left( M_{{A}}h_{{1}}u_{{1}
}+M_{{J}}h_{{2}}u_{{2}}+1 \right) ^{2}}} \right) \\
& -\lambda_{{2}}\left( b+{\frac {u_{{1}}M_{{A}}Ch_{{2}}u_{{2}}}{ \left( M_{{A}}h_{{1}
}u_{{1}}+M_{{J}}h_{{2}}u_{{2}}+1 \right) ^{2}}} \right) \\
& -\lambda_{{3}}\left( -b-{\frac {u_{{2}}C}{M_{{A}}h_{{1}}u_{{1}}+M_{{J}}h_{{2}}u_{{2
}}+1}}+{\frac {{u_{{2}}}^{2}M_{{J}}Ch_{{2}}}{ \left( M_{{A}}h_{{1}}u_{
{1}}+M_{{J}}h_{{2}}u_{{2}}+1 \right) ^{2}}} \right) -1,
\end{split}
\end{equation}

\noindent
with the transversality condition given as
\begin{equation}
\lambda_{1}(T)=\lambda_{2}(T)=\lambda_{3}(T)=0.
\end{equation}

Considering the optimality conditions, the Hamiltonian function is differentiated with respect to control variables $ u_{1} $ and $ u_{2}$  resulting in
\begin{equation}
\begin{split}
\frac{\partial H}{\partial u_{1}}=& \lambda_{{1}} \left( {\frac {M_{{A}}e_{{1}}C}{M_{{A}}h_{{1}}u_{{1}}+M_{{J}}h_{{2}}u_{{2}}+1}}-{\frac { \left( M_{{A}}e_{{1}}u_{{1}}+M_{{J}}e
_{{2}}u_{{2}} \right) CM_{{A}}h_{{1}}}{ \left( M_{{A}}h_{{1}}u_{{1}}+M
_{{J}}h_{{2}}u_{{2}}+1 \right) ^{2}}} \right) +\\
& \lambda_{{2}} \left( -{\frac {M_{{A}}C}{M_{{A}}h_{{1}}u_{{1}}+M_{{J}}h_{{2}}u_{{2}}+1}}+{
\frac {u_{{1}}{M_{{A}}}^{2}Ch_{{1}}}{ \left( M_{{A}}h_{{1}}u_{{1}}+M_{
{J}}h_{{2}}u_{{2}}+1 \right) ^{2}}} \right) +\\ 
& \lambda_{3} {\frac {u_{{2}}M_{{J}}CM_{{A}}h_{{1}}}{ \left( M_{{A}}h_{{1}}u_{{1}}+M_{{J}}h_{{2}
}u_{{2}}+1 \right) ^{2}}},\\
\frac{\partial H}{\partial u_{2}}=& \lambda_{{1}} \left( {\frac {M_{{J}}e_{{2}}C}{h_{{1}}u_{{1}}M
_{{A}}+h_{{2}}u_{{2}}M_{{J}}+1}}-{\frac { \left( e_{{1}}u_{{1}}M_{{A}}
+e_{{2}}u_{{2}}M_{{J}} \right) CM_{{J}}h_{{2}}}{ \left( h_{{1}}u_{{1}}
M_{{A}}+h_{{2}}u_{{2}}M_{{J}}+1 \right) ^{2}}} \right) +\\ 
&\lambda_{2} {\frac {u_{{1}}M_{{A}}CM_{{J}}h_{{2}}}{ \left( h_{{1}}u_{{1}}M_{{
A}}+h_{{2}}u_{{2}}M_{{J}}+1 \right) ^{2}}}+\\
& \lambda_{{3}} \left( -{
\frac {M_{{J}}C}{h_{{1}}u_{{1}}M_{{A}}+h_{{2}}u_{{2}}M_{{J}}+1}}+{
\frac {u_{{2}}{M_{{J}}}^{2}Ch_{{2}}}{ \left( h_{{1}}u_{{1}}M_{{A}}+h_{
{2}}u_{{2}}M_{{J}}+1 \right) ^{2}}} \right) -u_{{2}}.
\end{split}
\end{equation}

We find a characterization of $u_{1}^{*}$ by considering three cases:
\begin{equation}
\begin{split}
\frac{\partial H}{\partial u_{1}}<0 & \Rightarrow u_{1}^{*}=0,\\
\frac{\partial H}{\partial u_{1}}=0 & \Rightarrow u_{1}^{*}=u_{1_{1}} \quad  s.t. \quad \frac{\partial H}{\partial u_{1}}
\bigg| _{u_{1_{1}}}=0,\\
\frac{\partial H}{\partial u_{1}}>0 & \Rightarrow u_{1}^{*}=1.
\end{split}
\end{equation}

When the control is at the upper bound,$u_{1_{1}}$ is strictly greater than 1. When the control is at the lower bound, the solution of $u_{1_{1}}$ is strictly less than 0. Similarly for $u_{2}^{*}$. Thus a compact way of writing the optimal control is 
\begin{equation}
\label{eq:occ}
\begin{split}
u_{1}^{*} & =min(1, max(0, u_{1_{1}})),\\
u_{2}^{*} & =min(1, max(0, u_{2_{1}})),
\end{split}
\end{equation}

\noindent
where $u_{1_{1}}$  and $u_{2_{1}}$ are given by 
\begin{equation}
\begin{split}
\label{eq:occ2}
u_{1_{1}}=& \frac{w_{1}}{w_{2}},\\
u_{2_{1}}=& {\frac {-e_{{1}}\lambda_{{1}}+\lambda_{{2}}}{M_{{J}} \left( e_{{1}}h_{
{2}}\lambda_{{1}}-e_{{2}}h_{{1}}\lambda_{{1}}+h_{{1}}\lambda_{{3}}-h_{
{2}}\lambda_{{2}} \right) }}.
\end{split}
\end{equation}

\noindent
with
\begin{equation}
\begin{split}
\label{eq:occ3}
w_{1}=& C{M_{{J}}}^{2}{e_{{1}}}^{3}{h_{{2}}}^{3}{\lambda_{{1}}}^{3}-3\,C{M_{{J
}}}^{2}{e_{{1}}}^{2}e_{{2}}h_{{1}}{h_{{2}}}^{2}{\lambda_{{1}}}^{3}+3\,
C{M_{{J}}}^{2}e_{{1}}{e_{{2}}}^{2}{h_{{1}}}^{2}h_{{2}}{\lambda_{{1}}}^
{3}-\\
& C{M_{{J}}}^{2}{e_{{2}}}^{3}{h_{{1}}}^{3}{\lambda_{{1}}}^{3}+3\,C{M
_{{J}}}^{2}{e_{{1}}}^{2}h_{{1}}{h_{{2}}}^{2}{\lambda_{{1}}}^{2}\lambda
_{{3}}-3\,C{M_{{J}}}^{2}{e_{{1}}}^{2}{h_{{2}}}^{3}{\lambda_{{1}}}^{2}
\lambda_{{2}}-\\
& 6\,C{M_{{J}}}^{2}e_{{1}}e_{{2}}{h_{{1}}}^{2}h_{{2}}{
\lambda_{{1}}}^{2}\lambda_{{3}}+6\,C{M_{{J}}}^{2}e_{{1}}e_{{2}}h_{{1}}
{h_{{2}}}^{2}{\lambda_{{1}}}^{2}\lambda_{{2}}+3\,C{M_{{J}}}^{2}{e_{{2}
}}^{2}{h_{{1}}}^{3}{\lambda_{{1}}}^{2}\lambda_{{3}}-\\
& 3\,C{M_{{J}}}^{2}{
e_{{2}}}^{2}{h_{{1}}}^{2}h_{{2}}{\lambda_{{1}}}^{2}\lambda_{{2}}+3\,C{
M_{{J}}}^{2}e_{{1}}{h_{{1}}}^{2}h_{{2}}\lambda_{{1}}{\lambda_{{3}}}^{2
}-6\,C{M_{{J}}}^{2}e_{{1}}h_{{1}}{h_{{2}}}^{2}\lambda_{{1}}\lambda_{{2
}}\lambda_{{3}}+\\
& 3\,C{M_{{J}}}^{2}e_{{1}}{h_{{2}}}^{3}\lambda_{{1}}{
\lambda_{{2}}}^{2}-3\,C{M_{{J}}}^{2}e_{{2}}{h_{{1}}}^{3}\lambda_{{1}}{
\lambda_{{3}}}^{2}+6\,C{M_{{J}}}^{2}e_{{2}}{h_{{1}}}^{2}h_{{2}}\lambda
_{{1}}\lambda_{{2}}\lambda_{{3}}-\\
& 3\,C{M_{{J}}}^{2}e_{{2}}h_{{1}}{h_{{2
}}}^{2}\lambda_{{1}}{\lambda_{{2}}}^{2}+C{M_{{J}}}^{2}{h_{{1}}}^{3}{
\lambda_{{3}}}^{3}-3\,C{M_{{J}}}^{2}{h_{{1}}}^{2}h_{{2}}\lambda_{{2}}{
\lambda_{{3}}}^{2}+\\
& 3\,C{M_{{J}}}^{2}h_{{1}}{h_{{2}}}^{2}{\lambda_{{2}}
}^{2}\lambda_{{3}}-C{M_{{J}}}^{2}{h_{{2}}}^{3}{\lambda_{{2}}}^{3}+e_{{
1}}e_{{2}}{h_{{1}}}^{2}{\lambda_{{1}}}^{2}-e_{{1}}{h_{{1}}}^{2}\lambda
_{{1}}\lambda_{{3}}-\\
& e_{{2}}{h_{{1}}}^{2}\lambda_{{1}}\lambda_{{2}}+{h_
{{1}}}^{2}\lambda_{{2}}\lambda_{{3}},\\
w_{2}=& M_{{A}}{h_{{1}}}^{2} ( {e_{{1}}}^{2}h_{{2}}{\lambda_{{1}}}^{2}-e_
{{1}}e_{{2}}h_{{1}}{\lambda_{{1}}}^{2}+e_{{1}}h_{{1}}\lambda_{{1}}
\lambda_{{3}}-2\,e_{{1}}h_{{2}}\lambda_{{1}}\lambda_{{2}}+e_{{2}}h_{{1
}}\lambda_{{1}}\lambda_{{2}}-\\
& h_{{1}}\lambda_{{2}}\lambda_{{3}}+h_{{2}}
{\lambda_{{2}}}^{2} ).
\end{split}
\end{equation}

We can thus state the following theorem,

\begin{theorem}
\label{thm:oc11}
An optimal control $(u^{*}_{1}, u^{*}_{2}) \in U$ for the system  \eqref{eq:1a1njj}-\eqref{eq:3a1jjn} that maximises the objective functional $J_{1}$ is characterised by \eqref{eq:occ}.
\end{theorem}

\subsection{Maximizing mussel density w.r.t. attacking rates}

To maximinize mussel density, the attack rate on the adult mussels should be minimized.
We choose the following objective function,
\begin{equation}
  J_{2}(u_{1},u_{2})= \int^{T}_{0} (M_{A}+M_{J} - \frac{1}{2}u_{1}^{2} ) dt,
\end{equation}

\begin{center}
$s.t. $ \eqref{eq:1a1njj}- \eqref{eq:3a1jjn} and $C(t_{0})=C_{0}, M_{A}(t_{0})=M_{A_{0}}, M_{J}(t_{0})=M_{J_{0}}$. 
\end{center}
and we search for the optimal controls in the set $U$. Where,
\begin{equation}
 U= \{ (u_{1},u_{2})|u_{i} \ \mbox{measurable}, \  0 \leq u_{1} \leq 1, 0 \leq u_{2} \leq 1, \  t \in [0,T], \ \forall T\}.
\end{equation}

We can state the following existence theorem,

\begin{theorem}
\label{thm:oc2}
Consider the optimal control problem \eqref{eq:1a1njj}-\eqref{eq:3a1jjn}. There exists $(u^{*}_{1}, u^{*}_{2}) \in U$ s.t.

\begin{equation}
  J_{2}(u^{*}_{1}, u^{*}_{2})= \underset{(u_{1},u_{2}) \in U}{\max} \int^{T}_{0} (M_{A}+M_{J} - \frac{1}{2}u_{1}^{2})  dt.
\end{equation}
\end{theorem}

The proof is similar to theorem \ref{thm:oc1}.

We can next state 
\begin{theorem}
\label{thm:oc22}
An optimal control $(u^{*}_{1}, u^{*}_{2}) \in U$ for the system  \eqref{eq:1a1njj}-\eqref{eq:3a1jjn} that maximizes the objective function $J_{2}$ is characterised by 

\begin{equation}
\begin{split}
u_{1}^{*} & =min(1, max(0, u_{1_{2}})),\\
u_{2}^{*} & =min(1, max(0, u_{2_{2}})).
\end{split}
\end{equation}
\end{theorem}

For the details of the proof of the above necessary conditions and forms of $u_{1_{2}},u_{2_{2}}$ the reader is refered to the appendix section \ref{AP4}.

\subsection{Maximizing mussel density w.r.t. intraspecific competition rate}

In this approach we view the competition coefficient as a control. To reach certain optimal population densities, the mussels would maximise the densities of both adult and juvenile groups whilst minimising intraspecific competition. To this end our system reduces to

\begin{equation}
\label{eq:1a1njd}
C'=  -d_{1}C + e_{1}u_{1}\frac{ M_{A}}{1+2u_{1}(t) M_{A} + h_{2}u_{2}(t) M_{J}}C  + e_{2}u_{2}  \frac{ M_{J}}{1+2u_{1} M_{A} + h_{2}u_{2} M_{J}}C,
\end{equation}
\begin{equation}
\label{eq:2a1njd}
M_{A}'=  bM_{J} - \delta_{1}(t)M^{2}_{A} - u_{1} \frac{ M_{A}}{1+2u_{1} M_{A} + h_{2}u_{2} M_{J}}C,
\end{equation}

\begin{equation}
\label{eq:3a1jjd}
M_{J}'=  \frac{a}{2} M_{A} - bM_{J} - u_{2}  \frac{ M_{J}}{1+2u_{1} M_{A} + h_{2}u_{2} M_{J}}C.
\end{equation}

We choose the following objective function,
\begin{equation}
  J_{3}(\delta_1)= \int^{T}_{0} (M_{A}+M_{J} - \frac{1}{2}\delta_{1}^{2} ) dt,
\end{equation}

\begin{center}
$s.t. $ \eqref{eq:1a1njd}- \eqref{eq:3a1jjd} and $C(t_{0})=C_{0}, M_{A}(t_{0})=M_{A_{0}}, M_{J}(t_{0})=M_{J_{0}}$. 
\end{center}
and we search for the optimal controls in the set $U_1$. Where,
\begin{equation}
 U_1= \{ \delta_1 | \delta_1 \mbox{measurable}, \  0 \leq \delta_1 \leq \infty, \  t \in [0,T], \ \forall T\}.
\end{equation}

We can state the following existence theorem,

\begin{theorem}
\label{thm:oc2}
Consider the optimal control problem \eqref{eq:1a1njd}- \eqref{eq:3a1jjd}. There exist $(u^{*}_{1}, u^{*}_{2}) \in U$ s.t.

\begin{equation}
  J_{3}(\delta_1)= \underset{(u_{1},u_{2}) \in U}{\max} \int^{T}_{0} (M_{A}+M_{J} - \frac{1}{2}\delta_{1}^{2})  dt.
\end{equation}
\end{theorem}

The proof is similar to theorem \ref{thm:oc1}.

We can next state 
\begin{theorem}
\label{thm:oc33}
An optimal control $(u^{*}_{1}, u^{*}_{2}) \in U$ for the system  \eqref{eq:1a1njd}- \eqref{eq:3a1jjd} that maximises the objective function $J_{3}$ is characterised by 

\begin{equation}
\delta_{1}^{*}=max(0,-M_{A}^{2} \lambda_{2}).
\end{equation}
\end{theorem}

For the details of the proof of the above necessary conditions and forms of $u_{1_{2}},u_{2_{2}}$ the reader is refered to the appendix section \ref{AP5}. 
%% stop here
\subsection{Numerical Simulations}
In this subsection, we investigate via numerical simulation and compare the species' population of the control system \eqref{eq:1a1njjo}- \eqref{eq:4a1njjo} and the classical system \eqref{eq:1a1nj}-\eqref{eq:4a1nj} under the epibiont achieving the carrying capacity ($e=K$). Since the solutions of the states and adjoint equations are a prior bounded and concavity in the controls holds, the optimal controls exist by using a result from Fleming and Rishel[Chap III, Theorem 2.1, pp 63]\cite{FR75}. Forward-Backward Sweep iteration algorithms are used for numerical simulations. The following parameter set is chosen 

\begin{equation}
\label{parasets_2}
 d_{1}=0.1,e_{1}=0.9,e_{2}=0.5, h_{1}=2.0, h_{2}=1.0, b=0.2,  \delta_{1}=0.1, a=0.3
\end{equation}

 \begin{figure}[H]
        \begin{minipage}[b]{0.480\linewidth}
            \centering
            \includegraphics[width=\textwidth]{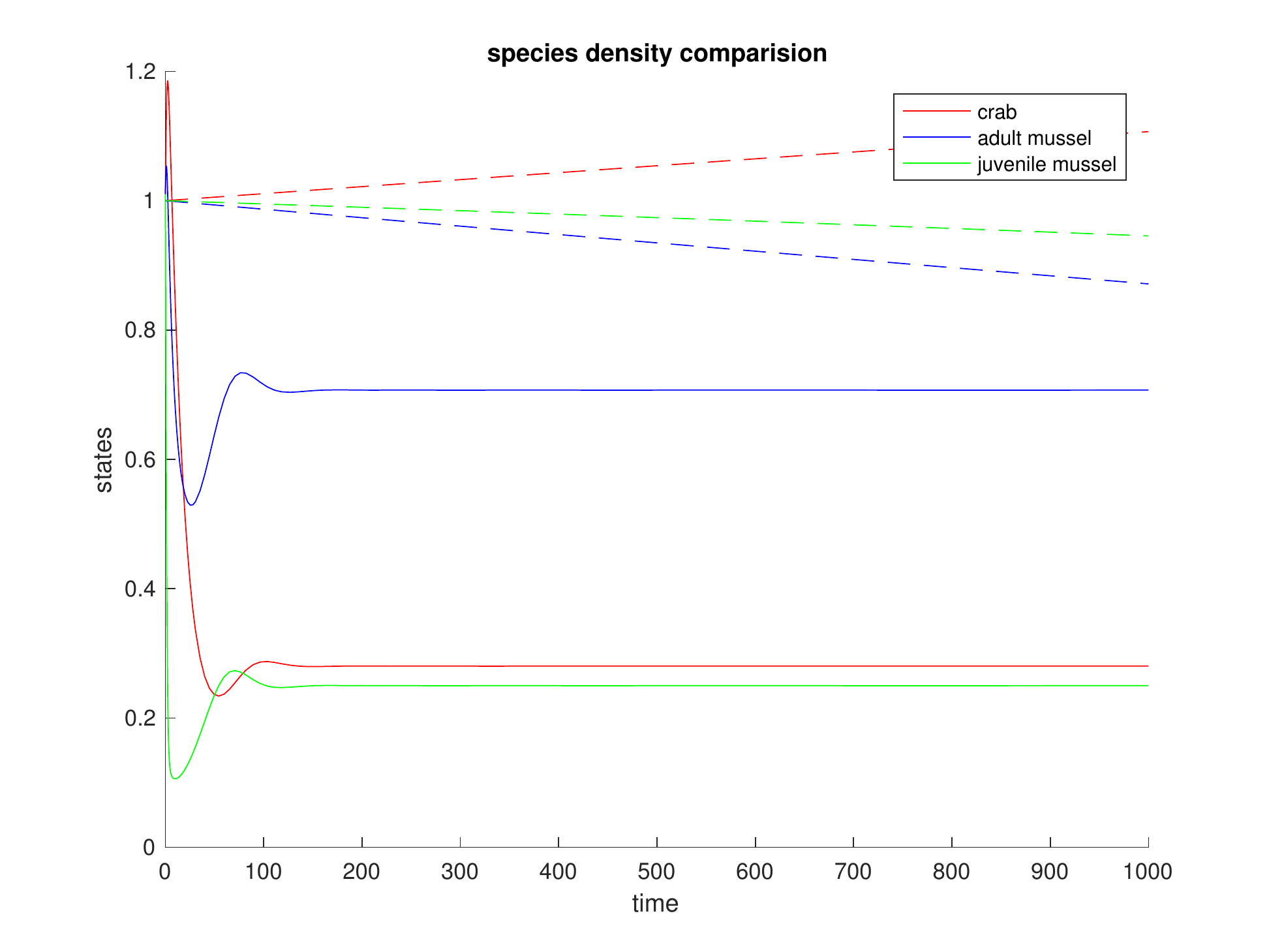}
            \subcaption{ }
             \label{fig:fig5_1}
        \end{minipage}
        \hspace{0.10cm}
        \begin{minipage}[b]{0.480\linewidth}
            \centering
            \includegraphics[width=\textwidth]{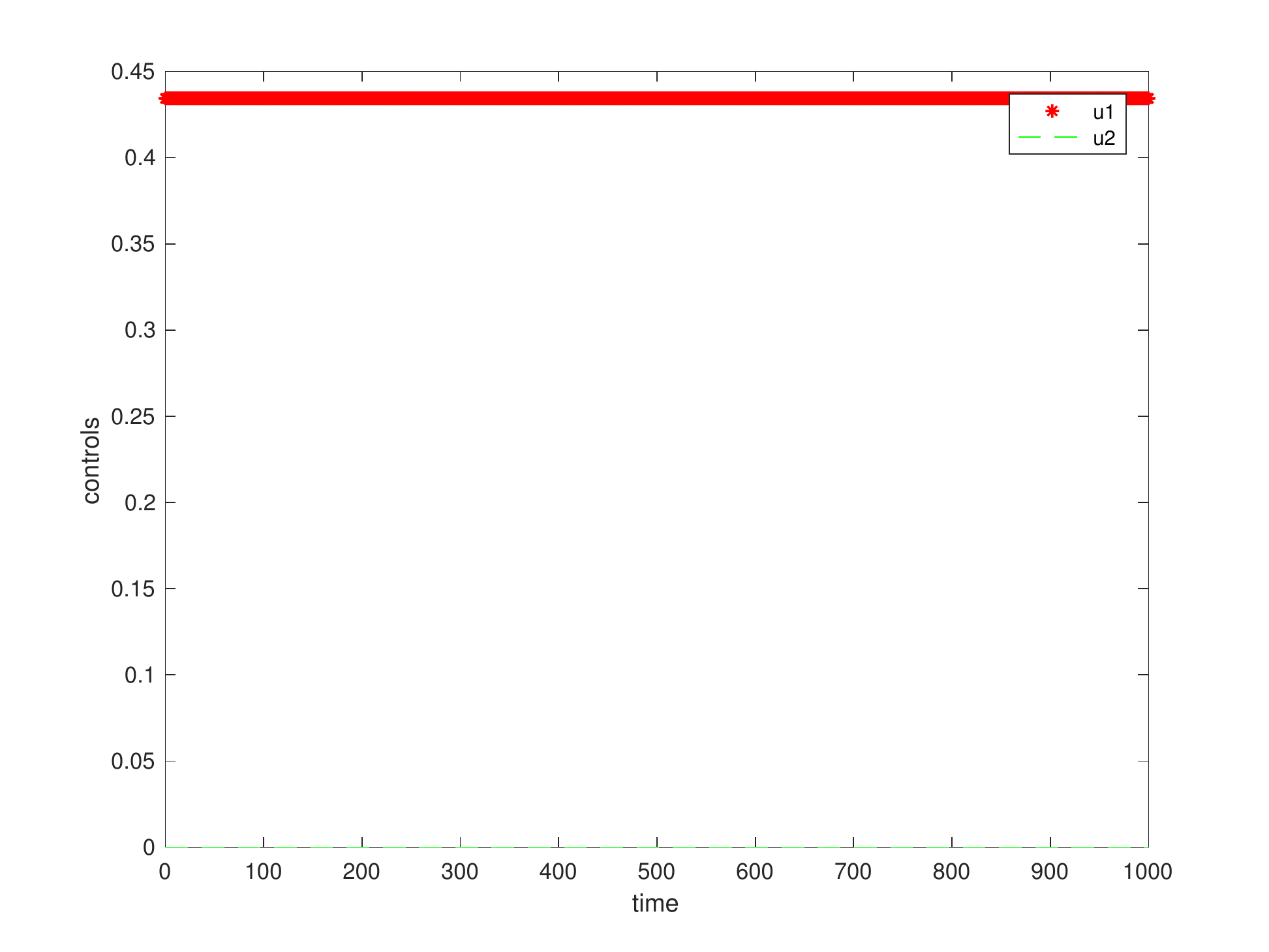}
            \subcaption{ }
            \label{fig:fig5_2}
        \end{minipage}
        \label{fig5}
            %\end{minipage}
       
\caption{(A) Solid curves are the density change for each species of the system \eqref{eq:1a1nj}-\eqref{eq:4a1nj} under $e=K$ and the dashed line are the optimal state of the control system \eqref{eq:1a1njjo}- \eqref{eq:4a1njjo} for the objective function $J_{1}(u_1,u_2)$ (B) Optimal controls of  $J_{1}(u_1,u_2)$ with the above parameter set \eqref{parasets_2}. }
\end{figure} 

We set $h_1=2$ since $h_{1}=1+\frac{e}{K}$, however, if we just assume $h_{1}$ as a constant and increase $h_{1}$ and keep other parameters the same, we found the optimal control $u_{2}$ always to be 0, and $u_{1}$ decreases and gradually become stable. In fact, when $h_{1}$ achieves to some critical value, $u_{1}$ begins to increase slightly. \\ % hide, disperse, refuge

\begin{figure}[H]
\label{fig:fig6}
        \begin{minipage}[b]{0.480\linewidth}
            \centering
            \includegraphics[width=\textwidth]{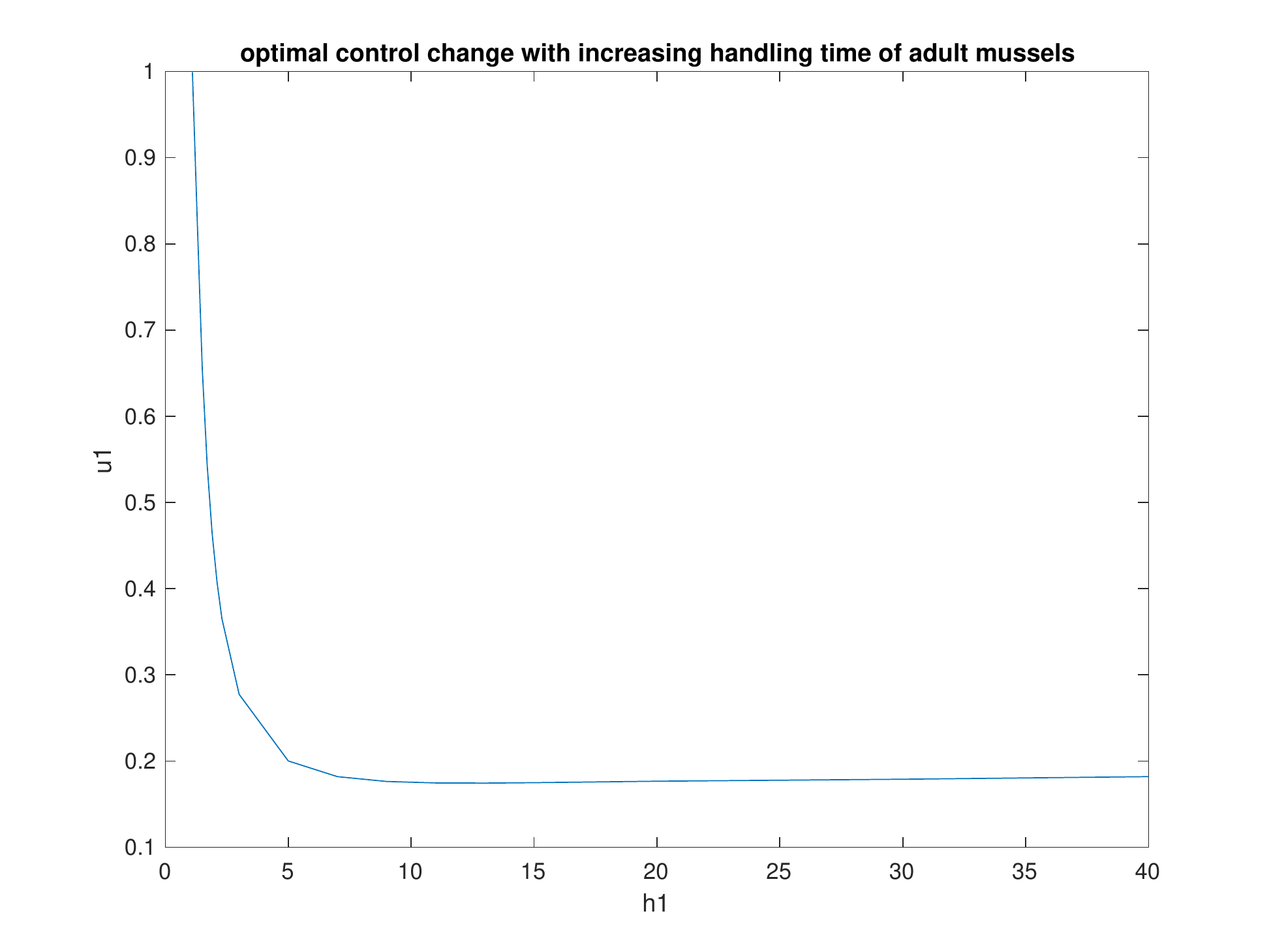}
            \subcaption{ }
             \label{fig:fig6_1}
        \end{minipage}
        \hspace{0.10cm}
        \begin{minipage}[b]{0.480\linewidth}
            \centering
            \includegraphics[width=\textwidth]{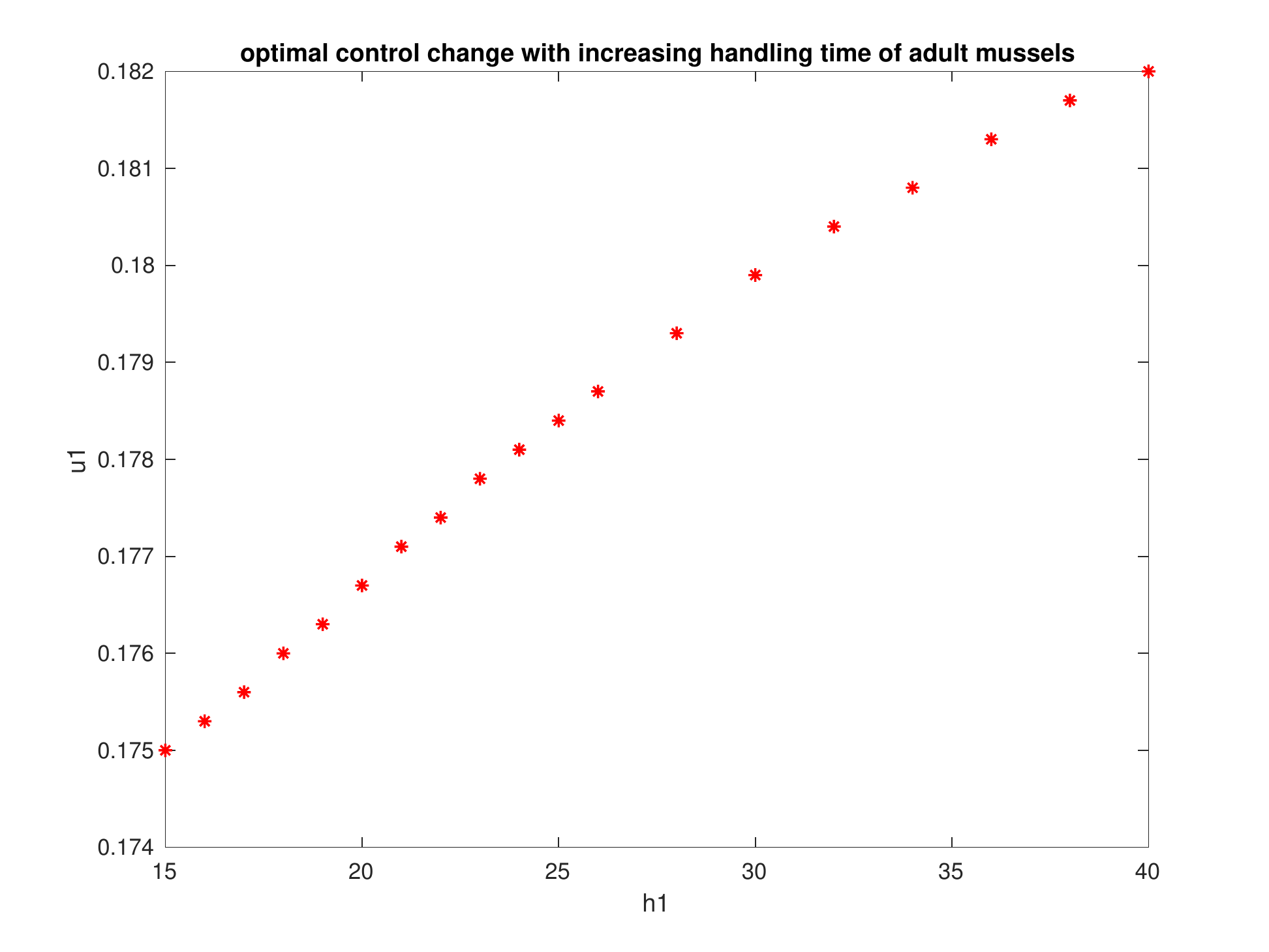}
            \subcaption{ }
             \label{fig:fig6_2}
        \end{minipage}
        \label{fig6}
\caption{(A) Optimal control $u_{1}$ changes with increasing $h_1$ (B) $u_{1}$ increases slightly with large $h_1$ }
\end{figure} 

\begin{figure}[H]
        \begin{minipage}[b]{0.480\linewidth}
            \centering
            \includegraphics[width=\textwidth]{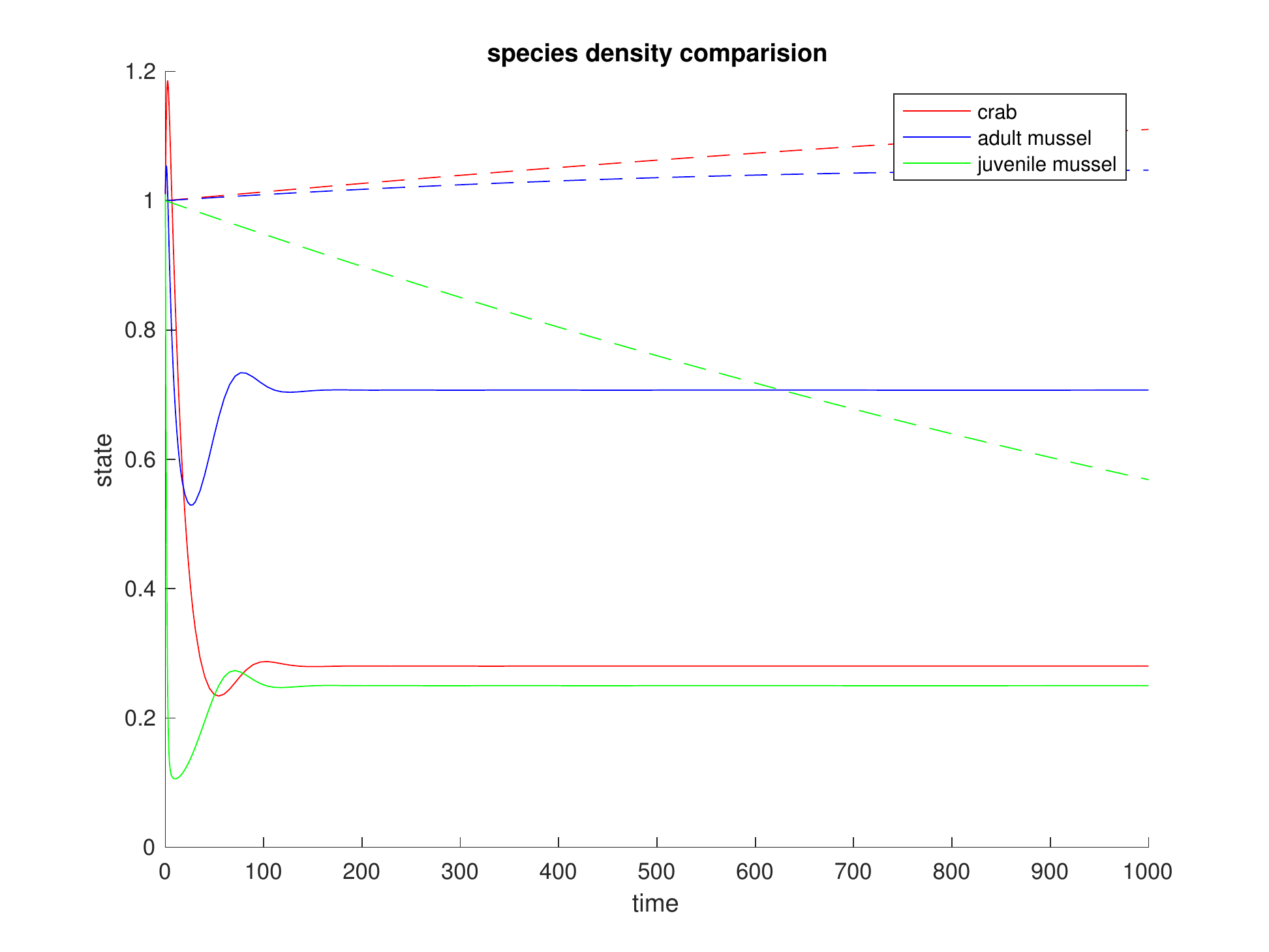}
            \subcaption{ }
             \label{fig:fig7_1}
        \end{minipage}
        \hspace{0.10cm}
        \begin{minipage}[b]{0.480\linewidth}
            \centering
            \includegraphics[width=\textwidth]{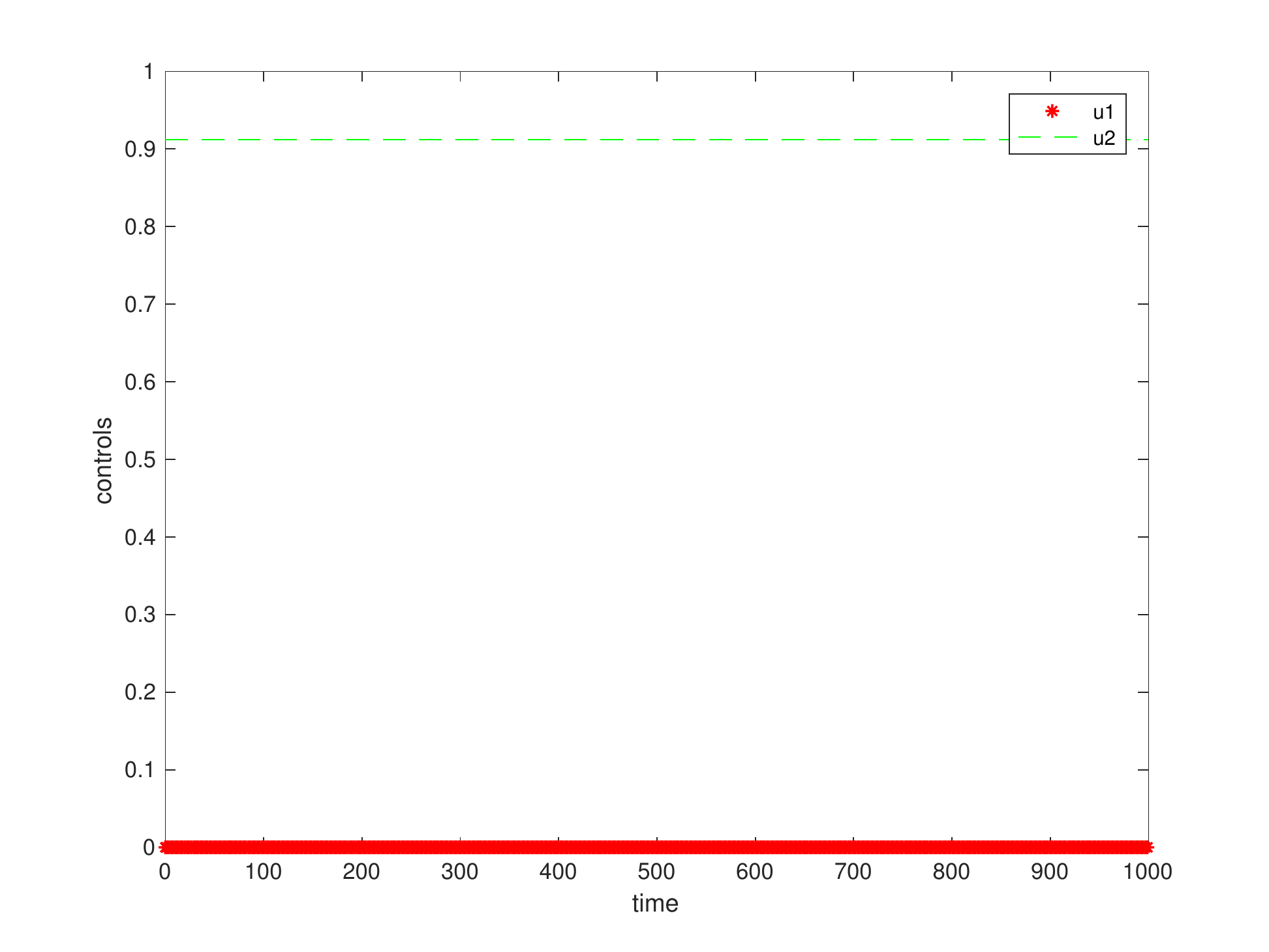}
            \subcaption{ }
            \label{fig:fig7_2}
        \end{minipage}
            \label{fig7}
   \caption{(A) Solid curves are the density change for each species of the system \eqref{eq:1a1nj}-\eqref{eq:4a1nj} under $e=K$ and the dashed line are the optimal state of the control system \eqref{eq:1a1njjo}- \eqref{eq:4a1njjo} for the objective function $J_{2}(u_1,u_2)$ (B) Optimal controls for$J_{2}(u_1,u_2)$ with the ablove parameter set shown in \eqref{parasets_2}. \\}
\end{figure} 

%\begin{figure}[H]
%            \centering
%            \includegraphics[width=0.6\textwidth]{u2_change.eps}
%\caption{In this simulation we look at how $u_{2}$ changes w.r.t. $h_{1}$. Here $h_2=1$. We want to see the change in the control $u_{2}$ as $h_{1}$ increases.
%The control $u_{1}=0$ no matter how large $h_1$ is. What we notice is that $u_2$ suddenly goes down to 0 from 1, at a critical value $h^{*}_1=2.1$.}
% \label{fig:fig8_1}
%\end{figure} 

\begin{figure}[H]
        \begin{minipage}[b]{0.480\linewidth}
            \centering
            \includegraphics[width=\textwidth]{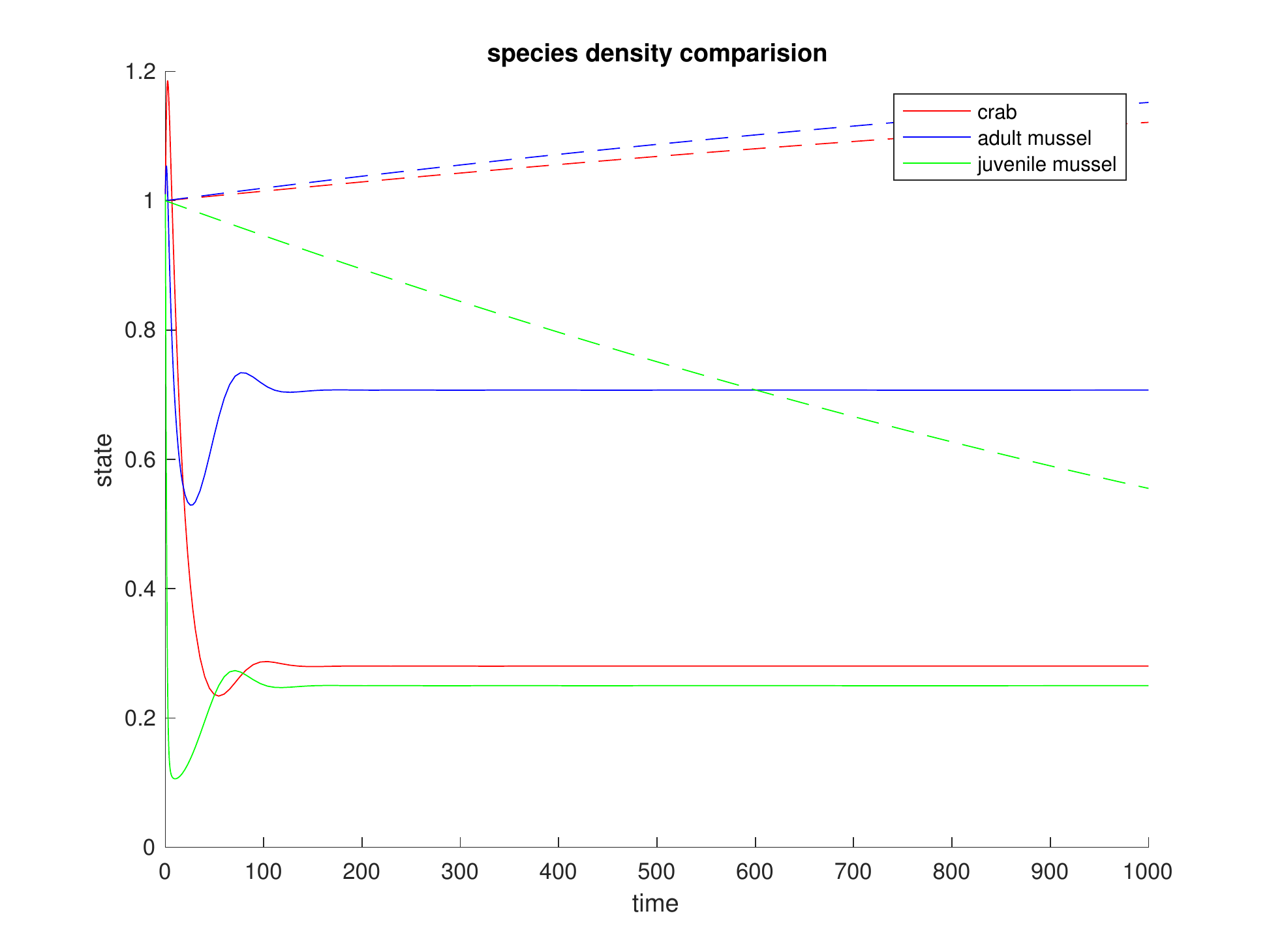}
            \subcaption{ }
             \label{fig:fig9_1}
        \end{minipage}
        \hspace{0.10cm}
        \begin{minipage}[b]{0.480\linewidth}
            \centering
            \includegraphics[width=\textwidth]{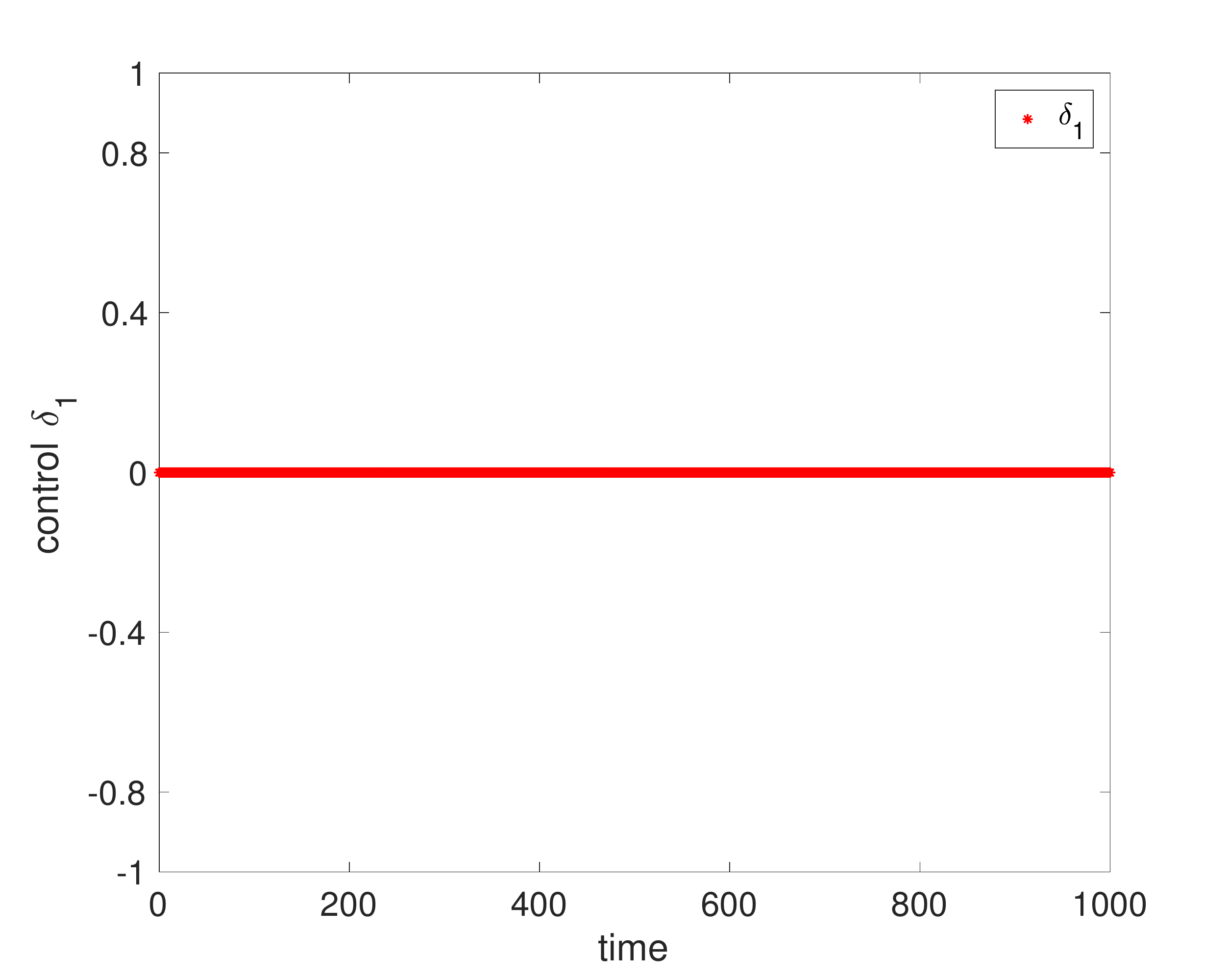}
            \subcaption{ }
            \label{fig:fig9_2}
        \end{minipage}
            \label{fig9}
\caption{(A) Solid curves are the density change for each species of the system \eqref{eq:1a1nj}-\eqref{eq:4a1nj} under $e=K$ and the dashed line are the optimal state of the control system \eqref{eq:1a1njjo}- \eqref{eq:4a1njjo} for the objective function $J_{3}(\delta_{1})$  (B) The optimal control for the objective function $J_3$ is awalys to be $\delta_1=0$. }
\end{figure}

Based on $e=K$, as for system \eqref{eq:1a1nj}-\eqref{eq:4a1nj} , the optimal foraging strategies are $u_1=0,u_2=1$. However, for the control system \eqref{eq:1a1njjo}- \eqref{eq:4a1njjo} , $J_1(u_1,u_2)$ will be maximized when $u_1=0.4343, u_2=0$; $J_{2}(u_1,u_2)$ will be maximized when $ u_1=0,u_2=0.9121$ and $J_3(\delta_1)$ will be maximized when $\delta_1=0$ with data set we mentioned. \\

\section{Discussion and Conclusion}

Epibiotic invasive species often have anti-predator defenses that are behavioral, chemical, or mechanical  \cite{IR06,WL02}, giving them a survival advantage in a novel habitat because potential predators avoid using them as a food source \cite{CR04,PP02}.  While this provides a benefit to the basibiont, it impacts other members of the community, including predators of the basibiont as they may show lower preference for basibionts that are overgrown by invasive epibionts \cite{A14,WH95}. However, the effects of epibionts on basibionts are not always positive. Many times the epibiont may attract predators resulting in consumption of the epibiont, which automatically leads to consumption of the basibiont. This is refered to in the literature as ``shared doom" \cite{LW99}. Epibionts can also negatively affect basibiont fecundity and fitness, resulting in fewer offspring \cite{A10}. In essence, invasion of predator-prey communities by epibionts is complex, and warrants a thorough mathematical investigation of their impact on predator-prey interactions and populations.

Population cycles are common in predator-prey communities, and although these are possible in our model without epibionts, extensive numerical simulations indicate that at carrying capacity $e=K$, a Hopf bifurcation is \emph{not} possible. This points to the epibiont having a stabilizing influence in that it can eliminate population oscillations. A rigorous proof of this is an interesting future direction.
Within our study, theorem \ref{thm:global} tells us that if the energy gain from the adult mussel is in a certain critical region $0 < e_{1} <1$, then one has global stability; even very large perturbations would still allow the system to return to its base state. 

Our central question focuses on the effect of the introduced epibiont on the population densities of the local crab-mussel communities. Could high epibiont density lead to lower mussel populations (and so subsequently lower crab populations)? To answer this we compare the equilibrium levels of the juvenile mussel population, ``no epibiont" case versus ``epibiont reaches carrying capacity" case. If the epibionts do have an adverse effect then we would have

$M^{*}_{J}|_{(e=K)} < M^{*}_{J}|_{(e=0)}$. Comparing these yields,

\begin{equation}
\label{eq:e11}
e_{1} + h_{2} a < e_{2} + h_{1} b.
\end{equation}

Although we know $e_{1} > e_{2}$, under high epibiont density $(e=K)$, we have $h_{1} >> h_{2}$, thus even if $a < b$, \eqref{eq:e11} could easily hold meaning that there is an adverse effect on the juvenile mussel density via epibiont presence, leading to fewer adults subsequently, and so epibionts could clearly be a factor in mussel population declines as seen via data from the Gulf of Maine \cite{SD17}. Such decline could eventually lead to crab population decline as well, if the crab species is a specialist on mussels. However, the effects of epibionts on mussel fecundity could also be a cause of predator decline. In order to understand the effects of epibiosis, we investigate the equilibrium density of the crab populations for the ``no epibiont" case versus the ``epibiont reaches carrying capacity" case. What we note is 

\begin{equation}
\label{eq:c11}
C^{*}_{J}|_{(e=0)} = \frac{e_{1}(a(e_{1}-d_{1}h_{1})-d_{1}\delta_{1})}{(e_{1}-d_{1}h_{1})^{2}}, \ C^{*}_{J}|_{(e=K)} = \frac{1}{2} \frac{a e_{2} \sqrt{\frac{b d_{1}}{(e_{2} -d_{1} h_{2}) \delta_{1}}}}{d_{1}} - \frac{e_{2} b}{e_{2}-d_{1} h_{2}}.
\end{equation}

Clearly, as epibiont cover reduces mussel fecundity from $a$, to $a/2$, this directly affects the crab population. In the $(e=K)$ case there is an increase by a factor of only $\frac{a}{2}$, as opposed to a factor of $a$ in the $(e=0)$ case. Thus reduced fecundity in mussels due to epibiont cover, can also reduce crab populations as well.

We assume logistic growth in the epibiont density. Although in the Gulf of Maine epibiont density fluctuates seasonally, our model could be a useful predictive tool in periods where logistic growth is seen. In locations such as Japan and New Zealand, the epibiont \emph{D. vexillum} grows logistically \cite{A10} due to water temperatures staying above the threshold for \emph{D. vexillum} viability. 

We also use optimal control theory to visualise various optimal scenarios to maximize each crab and mussel densities. Herein, we change the problem slightly, and assume the attack rates $u_{1}, u_{2}$ are not known \emph{a priori}, but are time-dependent. Our objective is to explore various scenarios that a species of crab or mussel might attempt to optimise, by manipulating the attack rates. Epibionts are assumed to be present, and their effect is modeled via increasing the handling time $h_{1}$ of adult mussels, as epibiont density increase. 

Simulations suggest (Fig. \ref{fig:fig5_1}) that even under high epibiont density (which in this scenario amounts to doubled $h_{1}$), the crab should \emph{not} attack juvenile mussels, but attempt to attack adults. Fig. \ref{fig:fig5_2} demonstrates, that what is optimal for the mussel is if $u_{1}=0$, and so the adult mussel must induce defenses to reduce $u_{1}$, even if it realistically cannot drive that rate to zero. This confirms experimental results of rapid shell thickening by mussels, seen via \cite{FB06}. 
Fig. \ref{fig:fig6_1} looks at the attack rate on adut mussels $u_{1}$, as $h_{1}$ changes. Here, we are trying to maximize mussel populations, and $u_{1}$ decreases as $h_{1}$ increases, as expected. However, $u_{1}$ is approximately $0.17$; that is, it does not change significantly if handling times become very large. Curiously, it goes up ever so slightly as shown in Fig. \ref{fig:fig6_2}. This likely corresponds to the adult mussel thickening its shell just enough to increase handling time by the crab.

Fig. \ref{fig:fig8_1} looks at the attack rate on juvenile mussels $u_{2}$, as $h_{1}$ changes. Here again we are trying to maximize mussel populations. When handling time on adults is low, juvenile mussels are protected from crab predation due to the predators preference for larger mussels. However, when handling time on adults is high (greater than 2.1 in this simulation), it is likely that crabs would switch to the juvenile; therefore, juveniles must be able to disperse or seek refuge in order to bring attack rates on the juvenile mussels to zero (in turn maximizing their population size). Young mussels drift in the water column until they reach a size of approximately 2.5mm, then they settle on a filamentous algal substrate \cite{SS92}. Some mussel species settle on algal substrate until they are 30 mm in length \cite{M95}. This substrate acts as  refuge and must be available for juveniles in order to maximize mussel populations. However, with degradation of suitable habitat, the opportunities for escape from crab populations becomes diminished. Major disturbance events, either natural or anthropogenic, in conjunction with invasion by substrate-smothering colonial species and voracious predators, are likely to decrease opportunities for escape. Endeavors to model predator-prey systems incorporating prey refuge may yield surprising results on stability \cite{KP15}, \cite{PQB16}. Thus it would be very interesting to model refuge effects for the juvenile mussels herein. 

As a future direction in modelling the crab-mussel-epibiont interaction, we would also like to examine interference effects \cite{G07, GG16, PBU16}. This effect has often seen to be stabilising \cite{LB15}, and thus modeling interference among the crab population, at high epibiont density is also realistic. Note, Theorem \ref{thm:oc33} suggests that eliminating intraspecific competition among mussels is optimal from their point of view, and yields a maximum density, if there was \emph{no} competition present. Future modelling endeavors may also investigate if high epibiont cover promotes cannibalism in crabs. That is, under high epibiont cover of adult mussels, would a crab prefer to  cannibalise its own conspecifics \cite{B16,BL16}, rather than switching to juvenile mussels? Another interesting future direction is to look at the foraging of crabs as they move in and out of patches containing mussels, some of which might be protected by mussel farmers, akin to marine protected areas \cite{PI16, PID16}. A spatially explicit approach to this end, modeling a changing habitat based on mussel density, would also be interesting \cite{BF14}.

The empirical literature shows that while epibionts alter the prey choice of predators, including crabs and sea stars \cite{E03, T07, V06, B07}, there are no prey switching experiments using \emph{D. vexillum}. Our goal is to provide firm modeling grounds for the scope of such experiments in the future. Thus a logical next step for empirical studies is to conduct experiments with \emph{D. vexillum} to confirm our switching hypothesis, as well as look at switching scenarios under varying levels of overgrowth (with both living and artificial epibionts such as in \cite{E03}). An interesting research question therein would be to ask if one sees the inverted parabola shaped curve, typical of OFT scenarios when measuring crab size versus mussel preference. If our switching hypothesis is confirmed, this should not be the case, as smaller size juvenile mussels should be preferred to adults under heavy epibiont cover. All in all we hope our results will help devise suitable strategies and measures that will enable a boost in dwindling mussel populations, particularly as new complexities arise in ecosystems, driven by rapid increase in invasions.

\section{Acknowledgements}
JL and RP would like to acknowledge valuable support from the NSF via DMS-1715377 and DMS-1839993.

\section{Appendix}
\subsection{Optimal Strategy in our setting}
\label{AP6}
Here, we give a rigorous reasoning for our switching hypothesis. 
If we following standard OFT, we can consider a fitness function 
\begin{equation}
R(u_{1},u_{2})=\frac{e_{1}u_{1}M_{A}}{1+h_{1}u_{1}M_{A}+h_{2}u_{2}M_{J}} +\frac{e_{1}u_{1}M_{J}}{1+h_{1}u_{1}M_{A}+h_{2}u_{2}M_{J}}.
\end{equation}
We endeavor to maximize $R(u_{1},u_{2})$, the net rate of energy intake during foraging. The optimal strategy for a crab (according to classical OFT) relies on the density of mussels. That is for each $(M_{A},M_{J})$ , we get a set of optimal controls $S(M_{A},M_{J})$ known as the strategy map.

\begin{equation}
S(M_{A},M_{J})=\{(u_{1},u_{2}) | R(u_{1},u_{2})=\max \limits_{0 \leq p_{1},p_{2} \leq 1}  R(p_{1},p_{2})\}.
\end{equation}

This is \eqref{eq:1a1nj}-\eqref{eq:4a1nj}, which is actually a control system with controls $(u_{1},u_{2})$ relying on the state of the system. Now we look for controls belonging to the strategy map $S(M_{A},M_{J})$. Then we calculate the derivatives of $S(M_{A},M_{J})$ to investigate the maximizing controls $u_{1}$ and $u_{2}$. 

\begin{equation}
\frac{\partial R}{\partial u_{1}}=\frac{M_{A}e_{1}+M_{A}M_{J}u_{2}(e_{1}h_{2}-e_{2}h_{1})}{(1+h_{1}u_{1}M_{A}+h_{2}u_{2}M_{J})^{2}},
\end{equation}

\begin{equation}
\frac{\partial R}{\partial u_{2}}=\frac{M_{J}(e_{2}-M_{A}u_{1}(e_{1}h_{2}-e_{2}h_{1}))}{(1+h_{1}u_{1}M_{A}+h_{2}u_{2}M_{J})^{2}}.
\end{equation}
The sign of $\frac{\partial R}{\partial u_{1}}$ and $\frac{\partial R}{\partial u_{2}}$ depend on the $e_{1}h_{2}-e_{2}h_{1}$.

A tricky point here is that attack rates depend critically on the density of adult and juvenile mussels. That is of $(u_{1} = 1, u_{2} = 0)$, or $(u_{1} = 0, u_{2} = 1)$ are feasible as attack rates if
the mussel densities are above a certain density. However
 if $M_{A}$, or $M_{J}$ fall below a certain critical level, theory predicts that the less preferred prey should \emph{also} be attacked, and one might have a situation of $(u_{1} = 1, u_{2} = 1)$. What we show next, is that if certain parametric restrictions are met, $(u_{1} = 1, u_{2} = 0)$, or $(u_{1} = 0, u_{2} = 1)$ are the \emph{only} optimal choices for the crab, irrespective of mussel density.

\begin{lemma}
\label{lem:sm1}
Consider \eqref{eq:1a1nj}-\eqref{eq:4a1nj}. If $e=0$, and $d_{1} > \frac{e_{2}}{h_{2}}$ then $(u_{1} = 1, u_{2} = 0)$ is the only optimal choices for the crab. Whereas if $e=K$, and $d_{1} > \frac{e_{1}}{h_{1}}$$(u_{1} = 0, u_{2} = 1)$ is the only optimal choices for the crab.
\end{lemma}

\begin{proof}
If $\frac{e_{1}}{h_{1}} > \frac{e_{2}}{h_{2}}$, $\frac{\partial R}{\partial u_{1}}>0$, the maximum of $R(u_{1},u_{2})$ is thusly achieved for $ u_{1}=1$.  And since the sign of $\frac{\partial R}{\partial u_{2}}$ does not depend on $u_{2}$ it follows if $ \frac{\partial R}{\partial u_{2}} \not= 0, R(u_{1},u_{2})$ will be maximized either with $u_{2}=0$ or $u_{2}=1$.Then we get the strategy map
\begin{equation}
S(M_{A},M_{J}) = \left\{
    \begin{array}{lll}
        (1,1) & \quad $if $   M_{A}<\frac{e_{2}}{e_{1}h_{2}-e_{2}h_{1}}, \\
       (1,0) & \quad $if $  M_{A}>\frac{e_{2}}{e_{1}h_{2}-e_{2}h_{1}}, \\
       (1,u_{2}), 0 \leq u_{2} \leq 1 & \quad $if $  M_{A}=\frac{e_{2}}{e_{1}h_{2}-e_{2}h_{1}}.
   \end{array}
 \right.
\end{equation}

Now $M^{*}_{A} = \frac{d_{1}}{e_{1}-d_{1}h_{1}}$ from the earlier stability calculations. We note, 

\begin{equation}
M^{*}_{A} = \frac{d_{1}}{e_{1}-d_{1}h_{1}} > \frac{e_{2}}{e_{1}h_{2}-e_{2}h_{1}},
\end{equation}

\noindent
as long as $d_{1} > \frac{e_{2}}{h_{2}}$, and if this is enforced $(u_{1} = 1, u_{2} = 0)$ is the only optimal strategy for the crab.

If $\frac{e_{1}}{h_{1}} < \frac{e_{2}}{h_{2}}$ in order to maximize $R(u_{1},u_{2})$, we need $u_{2}=1$. The strategy map will switch to
\begin{equation}
S(M_{A},M_{J}) = \left\{
    \begin{array}{lll}
        (1,1) & \quad $if $   M_{J}<\frac{e_{1}}{e_{2}h_{1}-e_{1}h_{2}}, \\
       (0,1) & \quad $if $  M_{J}>\frac{e_{1}}{e_{2}h_{1}-e_{1}h_{2}}, \\
       (u_{1},1), 0 \leq u_{1} \leq 1 & \quad $if $  M_{J}=\frac{e_{2}}{e_{1}h_{2}-e_{2}h_{1}}.\\\\
   \end{array}
 \right.
\end{equation}
Now $M^{*}_{J} = \frac{d_{1}}{e_{2}-d_{1} h_{2}}$ from the earlier stability calculations. We note, 

\begin{equation}
M^{*}_{J} = \frac{d_{1}}{e_{2}-d_{1} h_{2}} > \frac{e_{1}}{e_{2}h_{1}-e_{1}h_{2}}
\end{equation}

\noindent
as long as $d_{1} > \frac{e_{1}}{h_{1}}$, and if this is enforced $(u_{1} = 0, u_{2} = 1)$ is again, the only optimal strategy for the crab.

\subsection{Proof of theorem \ref{thm:cme0}}
\label{AP2}
The Jacobian matrix about  $(C^{*},  M_{A} ^{*}, M_{J} ^{*})$ of system \eqref{eq:1a1n112}-\eqref{eq:3a1n112}, without epiboint, is given by 
\begin{equation} \label{J__2} 
\mathrm{J}=\left[ \begin{array}{ccc}
0 & J_{12} & 0 \\ 
J_{21} & J_{22} & J_{23}  \\ 
0 & J_{32} & J_{33} \end{array}
\right] 
\end{equation} 

where
\begin{equation}
J_{12}=a (e_{1}-d_{1}h_{1})-d_{1} \delta_{1},
\end{equation}

\begin{equation}
J_{21}=-\frac{d_{1}}{e_{1}},
\end{equation}

\begin{equation}
J_{22}=-\frac{a (e_{1}-d_{1}h_{1})^{2}+d_{1}^{2} \delta_{1} h_{1}+d_{1} \delta_{1} e_{1}}{e_{1}  (e_{1}-d_{1} h_{1})},
\end{equation}

\begin{equation}
J_{23}=b,
\end{equation}

\begin{equation}
J_{32}=a,
\end{equation}

\begin{equation}
J_{33}=-b.
\end{equation}

The characteristic equation is 
\begin{equation}
\lambda^{3}+A_{2} \lambda^{2}+A_{1} \lambda+A_{0}=0,
\end{equation}
with
\begin{equation}
A_{2}=-J_{33}-J_{22},
\end{equation}

\begin{equation}
A_{1}=-J_{12}J_{21}+J_{22}J_{33}-J_{23}J_{32},
\end{equation}
and
\begin{equation}
A_{0}=J_{33}J_{21}J_{12}.
\end{equation}
It follows from the Routh-Hurwitz stability criteria that all eigenvalues have negative real part if 
\begin{equation}
A_{2}>0, \quad A_{0}>0, \quad A_{2}A_{1}>A_{0}.
\end{equation}

It is obvious that the first two conditions are always satisfied under feasibility condition \eqref{eq:feasibility2}. Furthermore, $A_{2}A_{1}-A_{0}>0$ if $ J_{23}J_{32}-J_{22}J_{33}<0$.
\begin{equation}
J_{23}J_{32}-J_{22}J_{33}=-\frac{b d_{1} (a d_{1} h_{1}^{2} -a e_{1} h_{1} +d_{1} \delta_{1} h_{1}+\delta_{1} e_{1})}{e (e_{1}-d_{1}h_{1})} <0.
\end{equation}

It is enough to solve $ b d_{1} (a d_{1} h_{1}^{2} -a e_{1} h_{1} +d_{1} \delta_{1} h_{1}+\delta_{1} e_{1})> 0 $, that is,  $ e_{1}-d_{1} h_{1}< \frac{d_{1}\delta{1}}{a }+\frac{d_{1} \delta_{1} e_{1}}{a h_{1} }. $ 

Therefore, the system \eqref{eq:1a1n112}-\eqref{eq:3a1n112} is asymptotically stable if 
\begin{equation}
 \frac{d_{1}\delta_{1}}{a }<e_{1}-d_{1} h_{1}< \frac{d_{1}\delta_{1}}{a }+\frac{d_{1} \delta_{1} e_{1}}{a h_{1} }.
\end{equation}

\subsection{Proof of Theorem \ref{thm:cmek}}
\label{AP1}
The equilibrium state, of the system \eqref{eq:1a1nj}-\eqref{eq:4a1nj}, for the epiboint is $e=K$. At the interior equilibrium state, the parameters $u_1=0, u_2=1$ and $a(e)=\frac{a}{2}$. Since  $e$ will not effect the solution of $ C, M_{A} $ and $ M_{J}$ once $u_1,u_2$ and $a(e)$ are determined, then it is enough to nvestigate the following three dimension system with the equilibrium $(C^{*}, M_{A}^{*},M_{J}^{*})=(\frac{1}{2} \frac{a e_{2} \sqrt{\frac{b d_{1}}{(e_{2} -d_{1} h_{2}) \delta_{1}}}}{d_{1}} - \frac{e_{2} b}{e_{2}-d_{1} h_{2}},\sqrt{\frac{b d_{1}}{(e_{2}-d_{1}h_{2})\delta_{1}}},\frac{d_{1}}{e_{2}-d_{1} h_{2}})$ when $ e=K.$

\begin{equation}
\label{equa_hf_epi_1}
\frac{dC}{dt}=  -d_{1}C +  e_{2}  \frac{ M_{J}}{1+ h_{2} M_{J}}C,
\end{equation}

\begin{equation}
\label{equa_hf_epi_2}
\frac{dM_{A}}{dt}=  bM_{J} - \delta_{1}M^{2}_{A}, 
\end{equation}

\begin{equation}
\label{equa_hf_epi_3}
\frac{dM_{J}}{dt}=  aM_{A} - bM_{J} -   \frac{ M_{J}}{1 +h_{2}M_{J}}C.
\end{equation}

The Jacobian matrix about   $ (C^{*},M_{A}^{*},M_{J}^{*})$ is 
\begin{equation} \label{GrindEQ__16_} 
\mathrm{J}=\left[ \begin{array}{ccc}
0 & 0 & J_{13}  \\ 
0 & J_{22} & J_{23}  \\ 
J_{31} & J_{32} & J_{33} 
\end{array}
\right] 
\end{equation} 

where
\begin{equation}
\label{Jaco_13}
J_{13}=\frac{1}{2} \frac{[a (e_{2}-d_{1} h_{2}) \sqrt{\frac{b d_{1}}{(e_{2} -d_{1} h_{2}) \delta_{1}}} -2 b d_{1}] (e_{2}-d_{1} h_{2} )}{d_{1}},
\end{equation}

\begin{equation}
\label{Jaco_22}
J_{22}=-2\delta_{1} \sqrt{\frac{b d_{1}}{(e_{2} -d_{1} h_{2}) \delta_{1}}},
\end{equation}

\begin{equation}
\label{Jaco_23}
J_{23}=b,
\end{equation}

\begin{equation}
\label{Jaco_31}
J_{31}=-\frac{d_{1}}{e_{2}},
\end{equation}

\begin{equation}
\label{Jaco_32}
J_{32}=\frac{a}{2},
\end{equation}

\begin{equation}
\label{Jaco_33}
J_{33}=-\frac{1}{2} \frac{ a(e_{2}-d_{1} h_{2})^{2} \sqrt{\frac{b d_{1}}{(e_{2} -d_{1} h_{2}) \delta_{1}}}+2 b d_{1}^{2} h_{2}}{e_{2} d_{1}}.
\end{equation}

Since all the parameters are positive, it is obvious that $ J_{22}<0,  J_{23}>0, J_{31}<0, J_{32}>0,$ and $J_{33}<0$. Under the feasibility condition \eqref{feasibility_epibiont}, $J_{13}>0$. And the characteristic equation is given by 

\begin{equation}
\lambda^{3}+B_{2} \lambda^{2}+B_{1} \lambda+B_{0}=0,
\end{equation}

where 
\begin{equation}
B_{2}=-J_{33}-J_{22},
\end{equation}

\begin{equation}
B_{1}=-J_{13}J_{31}+J_{22}J_{33}-J_{23}J_{32},
\end{equation}

\begin{equation}
B_{0}=J_{13}J_{22}J_{31},
\end{equation}

By Routh Hurwitz stability criteria, all eigenvalues have negative real part if 
\begin{equation}
B_{0}>0, B_{1}>0, B_{2}>0, B_{2}B_{1}-B_{0}>0.
\end{equation}

It is easy to check $B_{2}>0 $ and $B_{0}>0$ under the feasility criterion \eqref{feasibility_epibiont}. And $B_{1}>0 $ if $J_{22}J_{33}-J_{23}J_{32}>0$.

\begin{equation}
\begin{split}
J_{22}J_{33}-J_{23}J_{32} &=\frac{(4 \sqrt{\frac{bd_{1}}{\delta_{1}(e_{2}-d_{1}h_{2})}}d_{1}\delta_{1}h_{2}-2ad_{1}h_{2}+ae_{2})b}{2e_{2}d_{1}}\\
&=\frac{(4 \sqrt{\frac{bd_{1}}{\delta_{1}(e_{2}-d_{1}h_{2})}}d_{1}\delta_{1}h_{2}-ad_{1}h_{2}-ad_{1}h_{2}+ae_{2})b}{2e_{2}d_{1}}\\
&=\frac{(4 \sqrt{\frac{bd_{1}}{\delta_{1}(e_{2}-d_{1}h_{2})}}d_{1}\delta_{1}h_{2}-ad_{1}h_{2}+a(e_{2}-d_{1}h_{2}))b}{2e_{2}d_{1}}.
\end{split}
\end{equation}

 To make $J_{22}J_{33}-J_{23}J_{32}>0$, it is enough to show $4 \sqrt{\frac{bd_{1}}{\delta_{1}(e_{2}-d_{1}h_{2})}}d_{1}\delta_{1}h_{2}-ad_{1}h_{2}>0$, which gives us $e_{2}-d_{1}h_{2}<\frac{16bd_{1}\delta{1}}{a^{2}}$. Furthermore, 
\begin{equation}
\begin{split}
B_{2}B_{1}-B_{0} &=J_{13}J_{31}J_{33}-J_{22}^{2}J_{33}+J_{22}J_{23}J_{32}-J_{22}J_{33}^{2}+J_{23}J_{32}J_{33}\\
&=J_{13}J_{31}J_{33}+J_{22}(J_{23}J_{32}-J_{22}J_{33})+J_{33}(J_{23}J_{32}-J_{22}J_{33}).\\
\end{split}
\end{equation}

Since $J_{22}<0$ and $J_{33}<0$, $J_{22}J_{33}-J_{23}J_{32}>0$ implies $B_{2}B_{1}-B_{0}>0$. Thus, the system  \eqref{eq:1a1nj}-\eqref{eq:4a1nj} is asymptotically stable if 
\begin{equation}
\label{3_local stability_epiboint}
\frac{4bd_{1}\delta{1}}{a^{2}}<e_{2}-d_{1}h_{2}<\frac{16bd_{1}\delta_{1}}{a^{2}}.
\end{equation}
%%%%%%%
%%%%%%%%%
%%%%%%%%%%

\subsection{Proof of theorem \ref{thm:he0}}
\label{AP3}
Now let $a$, the growth rate of juvenile mussels, as the bifurcation parameter. Therefore, if condition \eqref{eq:feasibility2} holds, $ A_{0}(a_{*})$ are always positive.$ A_{2}(a_{*})>0$ if $ e_1 -d_1 h_1 <\frac{d_1 \delta_1}{a}+\frac{d_1 \delta_1 e_1}{ah_1}$.  And  $ \phi(a_{*})=A_{2}(a_{*})A_{1}(a_{*})-A_{0}(a_{*})=0 $ if 
\begin{equation}
\label{eq:hf1}
a_{*}=\frac{f_1}{f_2},
\end{equation}

where 
\begin{equation}
\label{eq:hf2}
\begin{split}
f_1 &= 4b \delta_{1}^{2} h_{1}^{5}  (M_{A}^{*})^{7} +(2 b^{2} \delta_{1} h_{1}^{5}+20b \delta_{1}^{2} h_{1}^{4} (M_{A}^{*})^{6}+(10b^{2} \delta_{1} h_{1}^{4} +40 b\delta_{1}^{2} h_{1}^{3}) (M_{A}^{*})^{5}+\\
  &\quad  (4C^{*} b \delta_{1} h_{1}^{3}+20b^{2} \delta_{1} h_{1}^{3} +2C^{*}\delta_{1} e_{1} h_{1}^{2} +40 b \delta_{1}^{2} h_{1}^{2}) (M_{A}^{*})^{4}+\\
  &\quad (C^{*} b^{2} h_{1}^{3} +12C^{*}b\delta_{1} h_{1}^{2} +20b^{2} \delta_{1} h_{1}^{2}+4C^{*}\delta_{1} e_{1} h_{1}+20b\delta_{1}^{2} h_{1}) (M_{A}^{*})^{3}+\\
  &\quad (3C^{*} b^{2} h_{1}^{2} +12 C^{*} b\delta_{1} h_{1}+10 b^{2} \delta_{1} h_{1}+2 C^{*}\delta_{1} e_{1}+4 b\delta_{1}^{2})  (M_{A}^{*})^{2}+\\
  &\quad ((C^{*})^{2} b h_{1} +3C^{*} b^{2} h_{1} +(C^{*})^{2} e_{1} +4C^{*} b \delta_{1} +2b^{2}\delta_{1}) M_{A}^{*}+(C^{*})^{2}b+C^{*}b^{2},\\
f_2 &=b (M_{A}^{*}  h_{1}+1)^{3} (2 (M_{A}^{*})^{3} \delta_{1} h_{1}^{2}+ b (M_{A}^{*})^{2} h_{1}^{2}+4(M_{A}^{*})^{2} \delta_{1} h_{1} +2 b M_{A}^{*} h_{1} +\\
& \quad 2M_{A}^{*} \delta_{1} +C^{*}+b),
\end{split} 
\end{equation}

and $C^{*},M_{A}^*$ are given by \eqref{eql_c} and \eqref{eql_a}. 

Furthermore, it is easy to verify that 
\begin{equation}
\begin{split}
\frac{d\phi(a)}{da}|_{a=a_{*}} &=-\frac{(b (M_{A}^{*})^{3} h_{1}^{3}+3b  (M_{A}^{*})^{2}h_{1}^{2} +3b M_{A}^{*}h_{1}+b)(2 (M_{A}^{*})^{2}\delta_{1}h_{1}^{2}}{(M_{A}^{*}+1)^{5}}\\
&  \quad -\frac{(bh_{1}^{2}+4\delta_{1}h_{1})(M_{A}^{*})^{2}+(2bh_{1}+2\delta_{1})M_{A}^{*}+C^{*}+b) }{(M_{A}^{*}+1)^{5}}\\
&<0\\
&\not=0.
\end{split}
\end{equation}

\subsection{Proof of theorem \ref{thm:oc22}}
\label{AP4}
The Hamiltonian of the system is given by
\begin{equation}
H=M_{A}+M_{J} -\frac{1}{2} u_{1}^{2} +\lambda_{1}  C' +\lambda_{2} M_{A}' +\lambda_{3} M_{J}'.
\end{equation}

We use the Hamiltonian to find a differential equation of the adjoint $\lambda_{i}, i=1,2,3$.
\begin{equation}
\begin{split}
\lambda_{1}'(t)=-& \lambda_{{1}} \left( -d_{{1}}+{\frac {M_{{A}}e_{{1}}u_{{1}}+M_{{J}}e_
{{2}}u_{{2}}}{M_{{A}}h_{{1}}u_{{1}}+M_{{J}}h_{{2}}u_{{2}}+1}} \right) 
+\\
& {\frac {\lambda_{{2}}u_{{1}}M_{{A}}}{M_{{A}}h_{{1}}u_{{1}}+M_{{J}}h_{
{2}}u_{{2}}+1}}+{\frac {\lambda_{{3}}u_{{2}}M_{{J}}}{M_{{A}}h_{{1}}u_{
{1}}+M_{{J}}h_{{2}}u_{{2}}+1}},\\
\lambda_{2}'(t)=-& \lambda_{{1}} \left( {\frac {e_{{1}}u_{{1}}C}{M_{{A}}h_{{1}}u_{{1}}
+M_{{J}}h_{{2}}u_{{2}}+1}}-{\frac { \left( M_{{A}}e_{{1}}u_{{1}}+M_{{J
}}e_{{2}}u_{{2}} \right) Ch_{{1}}u_{{1}}}{ \left( M_{{A}}h_{{1}}u_{{1}
}+M_{{J}}h_{{2}}u_{{2}}+1 \right) ^{2}}} \right) -\\
& \lambda_{{2}} \left( -2\,\delta_{{1}}M_{{A}}-{\frac {u_{{1}}C}{M_{{A}}h_{{1}}u_{{1}
}+M_{{J}}h_{{2}}u_{{2}}+1}}+{\frac {{u_{{1}}}^{2}M_{{A}}Ch_{{1}}}{
 \left( M_{{A}}h_{{1}}u_{{1}}+M_{{J}}h_{{2}}u_{{2}}+1 \right) ^{2}}}
 \right) -\\
& \lambda_{{3}} \left( a/2+{\frac {u_{{2}}M_{{J}}Ch_{{1}}u_{{1
}}}{ \left( M_{{A}}h_{{1}}u_{{1}}+M_{{J}}h_{{2}}u_{{2}}+1 \right) ^{2}
}} \right) -1,\\
\lambda_{3}'(t)=-& \lambda_{1} \left( {\frac {e_{{2}}u_{{2}}C}{M_{{A}}h_{{1}}u_{{1}}
+M_{{J}}h_{{2}}u_{{2}}+1}}-{\frac { \left( M_{{A}}e_{{1}}u_{{1}}+M_{{J
}}e_{{2}}u_{{2}} \right) Ch_{{2}}u_{{2}}}{ \left( M_{{A}}h_{{1}}u_{{1}
}+M_{{J}}h_{{2}}u_{{2}}+1 \right) ^{2}}} \right) -\\
& \lambda_{2}
 \left( b+{\frac {u_{{1}}M_{{A}}Ch_{{2}}u_{{2}}}{ \left( M_{{A}}h_{{1}
}u_{{1}}+M_{{J}}h_{{2}}u_{{2}}+1 \right) ^{2}}} \right) -\\
& \lambda_{3}
 \left( -b-{\frac {u_{{2}}C}{M_{{A}}h_{{1}}u_{{1}}+M_{{J}}h_{{2}}u_{{2
}}+1}}+{\frac {{u_{{2}}}^{2}M_{{J}}Ch_{{2}}}{ \left( M_{{A}}h_{{1}}u_{
{1}}+M_{{J}}h_{{2}}u_{{2}}+1 \right) ^{2}}} \right) -1,
\end{split}
\end{equation}

with the transversality condition gives as
\begin{equation}
\lambda_{1}(T)=\lambda_{2}(T)=\lambda_{3}(T)=0
\end{equation}

By solving 

\begin{equation}
\begin{split}
0=\frac{\partial H}{\partial u_{1}}=& \lambda_{{1}} \left( {\frac {M_{{A}}e_{{1}}C}{M_{{A}}h_{{1}}u
_{{1}}+M_{{J}}h_{{2}}u_{{2}}+1}}-{\frac { \left( M_{{A}}e_{{1}}u_{{1}}
+M_{{J}}e_{{2}}u_{{2}} \right) CM_{{A}}h_{{1}}}{ \left( M_{{A}}h_{{1}}
u_{{1}}+M_{{J}}h_{{2}}u_{{2}}+1 \right) ^{2}}} \right) +\\
& \lambda_{{2}}\left( -{\frac {M_{{A}}C}{M_{{A}}h_{{1}}u_{{1}}+M_{{J}}h_{{2}}u_{{2}}
+1}}+{\frac {u_{{1}}{M_{{A}}}^{2}Ch_{{1}}}{ \left( M_{{A}}h_{{1}}u_{{1
}}+M_{{J}}h_{{2}}u_{{2}}+1 \right) ^{2}}} \right) +\\ 
& \lambda_{3}{\frac {u_{{2}}M_{{J}}CM_{{A}}h_{{1}}}{ \left( M_{{A}}h_{{1}}u_{{1}}+M_{{J}}
h_{{2}}u_{{2}}+1 \right) ^{2}}}-u_{{1}},\\
0=\frac{\partial H}{\partial u_{2}}=& \lambda_{{1}} \left( {\frac {M_{{J}}e_{{2}}C}{M_{{A}}h_{{1}}u_{{1}}+M_{{J}}h_{{2}}u_{{2}}+1}}-{\frac { \left( M_{{A}}e_{{1}}u_{{1}}+M_{{J}}e
_{{2}}u_{{2}} \right) CM_{{J}}h_{{2}}}{ \left( M_{{A}}h_{{1}}u_{{1}}+M
_{{J}}h_{{2}}u_{{2}}+1 \right) ^{2}}} \right) +\\ 
& \lambda_{2} {\frac {u_
{{1}}M_{{A}}CM_{{J}}h_{{2}}}{ \left( M_{{A}}h_{{1}}u_{{1}}+M_{{J}}h_{{
2}}u_{{2}}+1 \right) ^{2}}}+\\
& \lambda_{{3}} \left( -{\frac {M_{{J}}C}{M_
{{A}}h_{{1}}u_{{1}}+M_{{J}}h_{{2}}u_{{2}}+1}}+{\frac {u_{{2}}{M_{{J}}}
^{2}Ch_{{2}}}{ \left( M_{{A}}h_{{1}}u_{{1}}+M_{{J}}h_{{2}}u_{{2}}+1
 \right) ^{2}}} \right).
\end{split}
\end{equation}

And $ u_{1_{2}} $ and $ u_{2_{2}}$ equal to 
\begin{equation}
\begin{split}
u_{1_{2}}= & {\frac {e_{2}\lambda_{1}-\lambda_{{3}}}{M_{{A}} \left( e_{
{1}}h_{{2}}\lambda_{{1}}-e_{{2}}h_{{1}}\lambda_{{1}}+h_{{1}}\lambda_{{
3}}-h_{{2}}\lambda_{{2}} \right) }},\\
u_{2_{2}}= & \frac{w_{1}}{w_{2}},
\end{split}
\end{equation}

where 

\begin{equation}
\begin{split}
w_{1}=& C{M_{{A}}}^{2}{e_{{1}}}^{3}{h_{{2}}}^{3}{\lambda_{{1}}}^{3}-3\,C{M_{{A
}}}^{2}{e_{{1}}}^{2}e_{{2}}h_{{1}}{h_{{2}}}^{2}{\lambda_{{1}}}^{3}+3\,
C{M_{{A}}}^{2}e_{{1}}{e_{{2}}}^{2}{h_{{1}}}^{2}h_{{2}}{\lambda_{{1}}}^
{3}-\\
& C{M_{{A}}}^{2}{e_{{2}}}^{3}{h_{{1}}}^{3}{\lambda_{{1}}}^{3}+3\,C{M
_{{A}}}^{2}{e_{{1}}}^{2}h_{{1}}{h_{{2}}}^{2}{\lambda_{{1}}}^{2}\lambda
_{{3}}-3\,C{M_{{A}}}^{2}{e_{{1}}}^{2}{h_{{2}}}^{3}{\lambda_{{1}}}^{2}
\lambda_{{2}}-\\
& 6\,C{M_{{A}}}^{2}e_{{1}}e_{{2}}{h_{{1}}}^{2}h_{{2}}{
\lambda_{{1}}}^{2}\lambda_{{3}}+6\,C{M_{{A}}}^{2}e_{{1}}e_{{2}}h_{{1}}
{h_{{2}}}^{2}{\lambda_{{1}}}^{2}\lambda_{{2}}+3\,C{M_{{A}}}^{2}{e_{{2}
}}^{2}{h_{{1}}}^{3}{\lambda_{{1}}}^{2}\lambda_{{3}}-\\
& 3\,C{M_{{A}}}^{2}{
e_{{2}}}^{2}{h_{{1}}}^{2}h_{{2}}{\lambda_{{1}}}^{2}\lambda_{{2}}+3\,C{
M_{{A}}}^{2}e_{{1}}{h_{{1}}}^{2}h_{{2}}\lambda_{{1}}{\lambda_{{3}}}^{2
}-6\,C{M_{{A}}}^{2}e_{{1}}h_{{1}}{h_{{2}}}^{2}\lambda_{{1}}\lambda_{{2
}}\lambda_{{3}}+\\
& 3\,C{M_{{A}}}^{2}e_{{1}}{h_{{2}}}^{3}\lambda_{{1}}{
\lambda_{{2}}}^{2}-3\,C{M_{{A}}}^{2}e_{{2}}{h_{{1}}}^{3}\lambda_{{1}}{
\lambda_{{3}}}^{2}+6\,C{M_{{A}}}^{2}e_{{2}}{h_{{1}}}^{2}h_{{2}}\lambda
_{{1}}\lambda_{{2}}\lambda_{{3}}-\\
& 3\,C{M_{{A}}}^{2}e_{{2}}h_{{1}}{h_{{2
}}}^{2}\lambda_{{1}}{\lambda_{{2}}}^{2}+C{M_{{A}}}^{2}{h_{{1}}}^{3}{
\lambda_{{3}}}^{3}-3\,C{M_{{A}}}^{2}{h_{{1}}}^{2}h_{{2}}\lambda_{{2}}{
\lambda_{{3}}}^{2}+\\
& 3\,C{M_{{A}}}^{2}h_{{1}}{h_{{2}}}^{2}{\lambda_{{2}}
}^{2}\lambda_{{3}}-C{M_{{A}}}^{2}{h_{{2}}}^{3}{\lambda_{{2}}}^{3}-e_{{
1}}e_{{2}}{h_{{2}}}^{2}{\lambda_{{1}}}^{2}+e_{{1}}{h_{{2}}}^{2}\lambda
_{{1}}\lambda_{{3}}+\\
& e_{{2}}{h_{{2}}}^{2}\lambda_{{1}}\lambda_{{2}}-{h_
{{2}}}^{2}\lambda_{{2}}\lambda_{{3}},\\
w_{2}= &  M_{J}h_{2}^{2} ( e_{1} e_{2} h_{2} \lambda_{1}^{2}-
e_{2}^{2} h_{1} \lambda_{1}^{2}-e_{1} h_{2} \lambda_{1}
\lambda_{3}+2\,e_{2} h_{1} \lambda_{1} \lambda_{3}-e_{2}h_{2} \lambda_{1}\lambda_{2}-h_{1} \lambda_{3}^{2}+\\
& h_{2}\lambda
_{2} \lambda_{3} ). 
\end{split}
\end{equation}

So that the optimal controls for $J_{2}(u_1,u_2)$ is 
\begin{equation}
\begin{split}
u_{1}^{*} & =min(1, max(0, u_{1_{2}})),\\
u_{2}^{*} & =min(1, max(0, u_{2_{2}})).
\end{split}
\end{equation}

\subsection{Proof of theorem \ref{thm:oc33}}
\label{AP5}
The Hamiltonian of our problem is given by
\begin{equation}
H=M_{A}+M_{J} -\frac{1}{2} \delta_{1}^{2} +\lambda_{1}  C' +\lambda_{2} M_{A}' +\lambda_{3} M_{J}'.
\end{equation}

The differential equations for $\lambda_{1}'(t),\lambda_{2}'(t),\lambda_{3}'(t)$, are standard and are derived as in Theorem \ref{thm:oc22}.

 the transversality condition is

\begin{equation}
\lambda_{1}(T)=\lambda_{2}(T)=\lambda_{3}(T)=0.
\end{equation}

Considering $\frac{\partial H}{\partial \delta_{1}}=-{M_{{A}}}^{2}\lambda_{{2}}-\delta_{{1}}$, we derive the optimal control for $J_{3}(\delta_{1})$

\begin{equation}
\delta_{1}^{*}=max(0,-M_{A}^{2} \lambda_{2}).
\end{equation}

\end{proof}

\begin{figure}[H]
            \centering
            \includegraphics[width=0.6\textwidth]{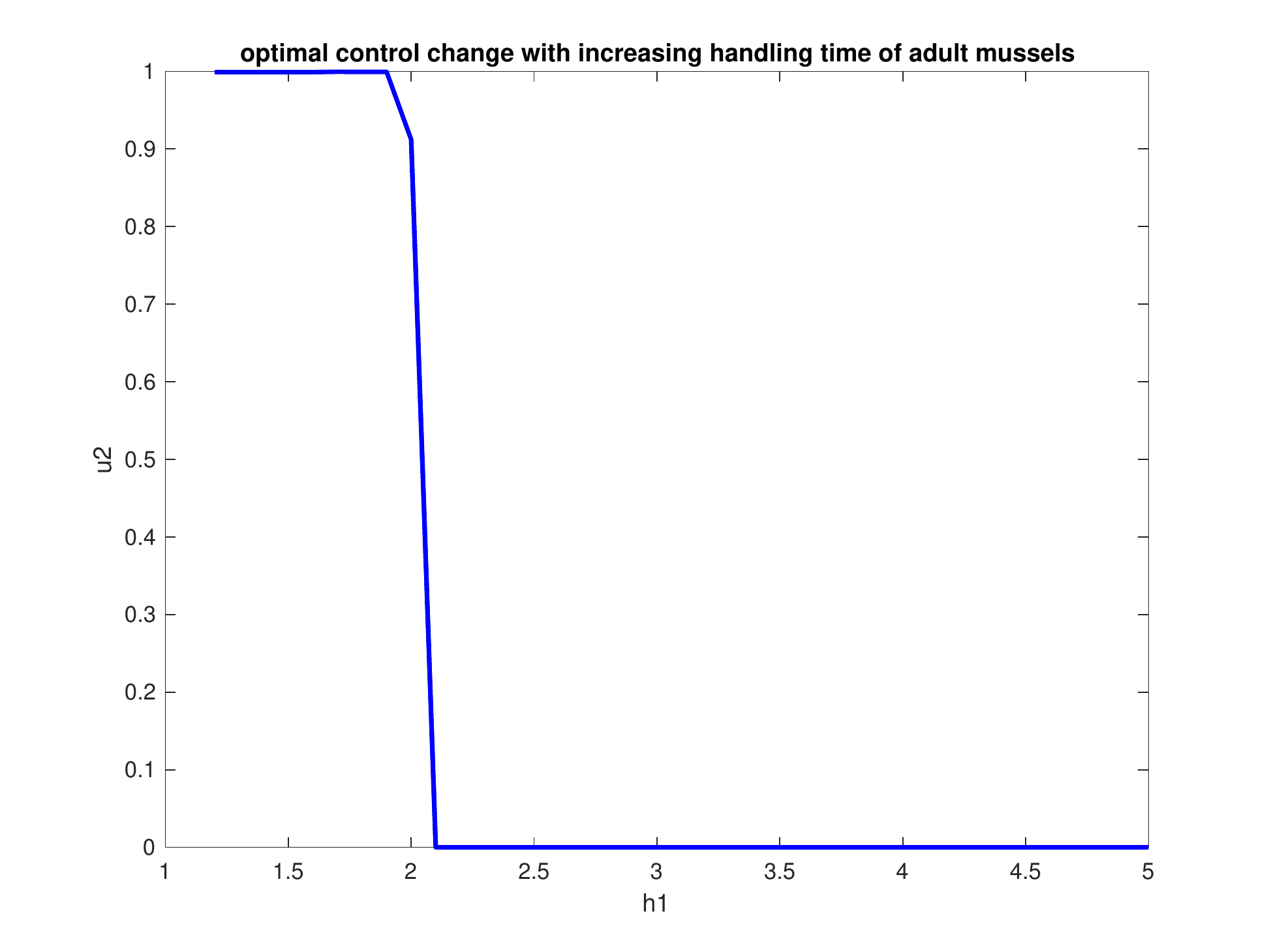}
\caption{In this simulation we look at how $u_{2}$ changes w.r.t. $h_{1}$. Here $h_2=1$. We want to see the change in the control $u_{2}$ as $h_{1}$ increases.
The control $u_{1}=0$ no matter how large $h_1$ is. What we notice is that $u_2$ suddenly goes down to 0 from 1, at a critical value $h^{*}_1=2.1$.}
 \label{fig:fig8_1}
\end{figure}

\section{Numerical Explorations of Bifurcations of alternate models}

In the case $e=K$, we do not see a Hopf bifurcation numerically. It is worthwhile considering certain alternate models for the epibiont dynamics as future work.
We motivate this via considering the following model,

\begin{equation}
\label{eq:1a1nj} 
\begin{split}
\frac{dC}{dt}= & -d_{1}C + e_{1}u_{1}(e)\frac{
M_{A}}{1+h_{1}u_{1}(e) M_{A} + h_{2}u_{2}(e) M_{J}}C  \\
& +e_{2}u_{2}(e)  \frac{ M_{J}}{1+h_{1}u_{1}(e) M_{A} + h_{2}u_{2}(e)
M_{J}}C,
\end{split}
\end{equation}
\begin{equation}
\label{eq:2a1nj} \frac{dM_{A}}{dt}=  bM_{J} - \delta_{1}M^{2}_{A} -
u_{1}(e)\frac{ M_{A}}{1+h_{1}u_{1}(e) M_{A} + h_{2}u_{2}(e) M_{J}}C,
\end{equation}

\begin{equation}
\label{eq:3a1nj} \frac{dM_{J}}{dt}=  a(e) M_{A} - bM_{J} - u_{2}(e)
\frac{ M_{J}}{1+h_{1}u_{1}(e) M_{A} + h_{2}u_{2}(e) M_{J}}C,
\end{equation}

\begin{equation}
\label{eq:4a1nj} \frac{de}{dt}=  b_{1}e(1-\frac{e}{k_1 M_A+k_2}).
\end{equation}

where
\begin{equation}
\label{eq:eee} u_{1}(e)=  \frac{K-e}{K}, \ u_{2}(e)=   \frac{e}{K},
\  a(e)=a\left(\frac{K-\frac{e}{2}}{K}\right), K=k_1M_A+k_2.
\end{equation}

\noindent with positive initial conditions $C(0)=C_{0},
M_{A}(0)=M_{A0}, M_{J}(0)=M_{J0}, e(0)=e_{0}$. These responses are
for the range $0 \leq e \leq K$.

The only change here to crab-mussel system \eqref{eq:1a1nj}-\eqref{eq:4a1nj} is that we assume the carrying capacity of the epibiont is density dependent, and depends primarily on the adult mussel density that is $K=K=k_1 M_A+k_2$. Here $k_{2}$ represents alternate substrate that the epibiont can grow on.

%To observe the qualitative behavior, extensive numerical simulations
%are carried out. For one-parameter bifurcation diagram we use
%"continuation algorithm" in which given a stationary solution
%$x_{t}(\alpha_{0})$ at $\alpha=\alpha_{0}$ of the system
%\begin{equation}
%f(x,\alpha)=0
%\end{equation}
%The aim is to estimate the new stationary point at
%$\alpha=\alpha_{0}+\vartriangle\alpha$ starting from the older one.
%The continuation algorithms aim to trace the loci of the solutions
%of algebraic equations given a first tentative value. The
%algorithm is used by many dynamical software such as XPPAUT,
%MATCONT, DYNAMICS etc.,. Results have been verified by using the
%MATLAB/MATHEMATICA code developed for this system.
%
%

%\section{Hopf Bifurcation and Periodic Solutions}

The four dimensional system has $11$ parameters with four dependent
variables. The following parameters are used in numerical simulations:

\begin{equation}
\label{par1} 
\begin{split}
&e_1=0.8, e_2=0.5, d_1=0.4, a=4, b_1=2, b=0.5,\\
& h_1=2,h_2=1, \delta_1=0.2, k_1=0.1, k_2=0.3.
\end{split}
\end{equation}

The system evolve the stable limit cycles for parameter set \ref{par1}. Time series for all species is shown in the figure \ref{fig:p1}, while limit cycles in 2-D phase space are shown in fig \ref{fig:p2}, \ref{fig:p3} and \ref{fig:p4}. To observe the more qualitative behavior of the model, one-parameter bifurcation diagram is drawn with respect to parameter $ d_{1}$   and parameter $a$ in the figures .   A supercritical Hopf bifurcation occurs at $ d_{1}=0.3567$ which emanates stable limit cycles. There is another supercritical hopf bifurcation at  $ d_{1}=0.444$. Between these two Hopf bifurcation, model has periodic solutions. After second Hopf bifurcation point model has stable solutions but crab populations are going to extinct. The dynamics is shown in one-parameter bifurcation diagram fig \ref{fig:p5}. The qualitative dynamics  has been also obtained for range of parameter $ a$ drawn in the fig \ref{fig:p6}. Initially, for low parameter value $a< 1.265$ the crab population is too low but as parameter $a$ increases, model exhibits stable coexistence. Further, it  undergoes through supercritical Hopf bifurcation at parameter $ a=3.386$ which emanates stable limit cycles (green filled circle). 

The parameter region has been obtained by drawing two-parameter $(a, d_{1})$ bifurcation diagram in the fig \ref{fig:p7}. The parameter region for which one species goes extinction is shown by shaded region (extreme left), region for which stable coexistence is possible shown by red and region for which perioidc solution is possible shown by blue color in the diagram fig \ref{fig:p7}. Another two-parameter  $( d_{1}, \delta)$ bifurcation diagram is drawn in the fig \ref{fig:p8}. The parameter region for  which stable coexistence occurs and region for which perioidc solution is possible is depicted in the diagram fig \ref{fig:p8}. 

As shown by these bifurcation graphs, model has periodic solutions for biologically feasible choice of parameters and one can find the Hopf bifurctaion point for each of the parameters used in the model.

These results show that a Hopf bifurcation is possible, if one considers a density dependent carrying capacity for the epibiont. These results are robust in nature as different sets of parameters will yield the same qualitative behavior. The periodicity in the system is beneficial for harvesting  and coexistence of all the species involved.   

%\begin{figure}[h!]
 % \includegraphics[scale=1.7]{edit_ad}
  %\caption{The birds}
  %\label{fig:birds}
%\end{figure}

\begin{figure}
\begin{subfigure}{.45\textwidth}
  \centering
  \includegraphics[angle=00,width=.75\linewidth]{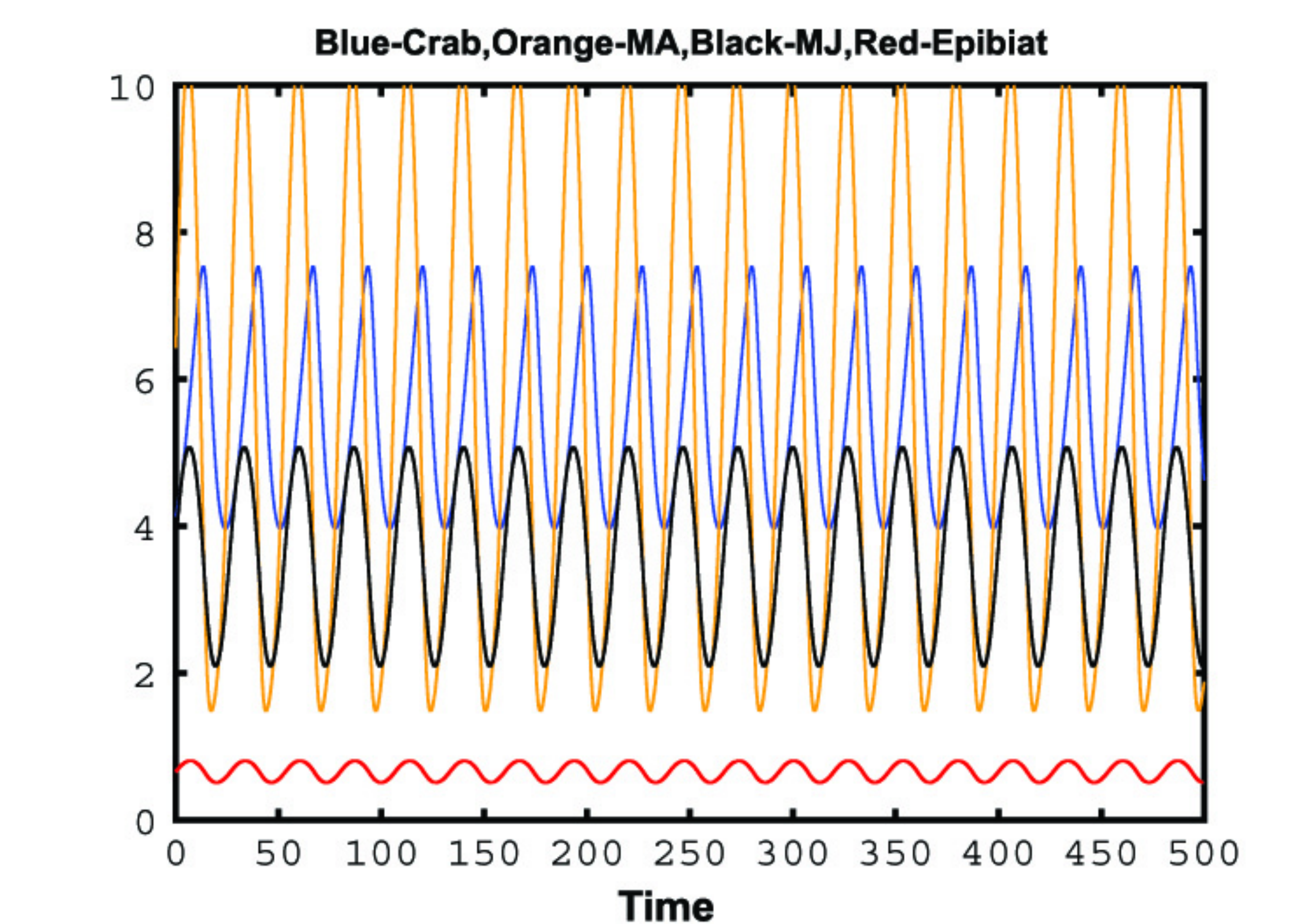}
  \caption{The biomasses of all species exhibited the periodic coexistence against the Time series is  shown for parameter set \ref{par1}. }
 \label{fig:p1}
\end{subfigure} \hfill 
\begin{subfigure}{.45\textwidth}
  \centering
  \includegraphics[angle=-90,width=.8\linewidth]{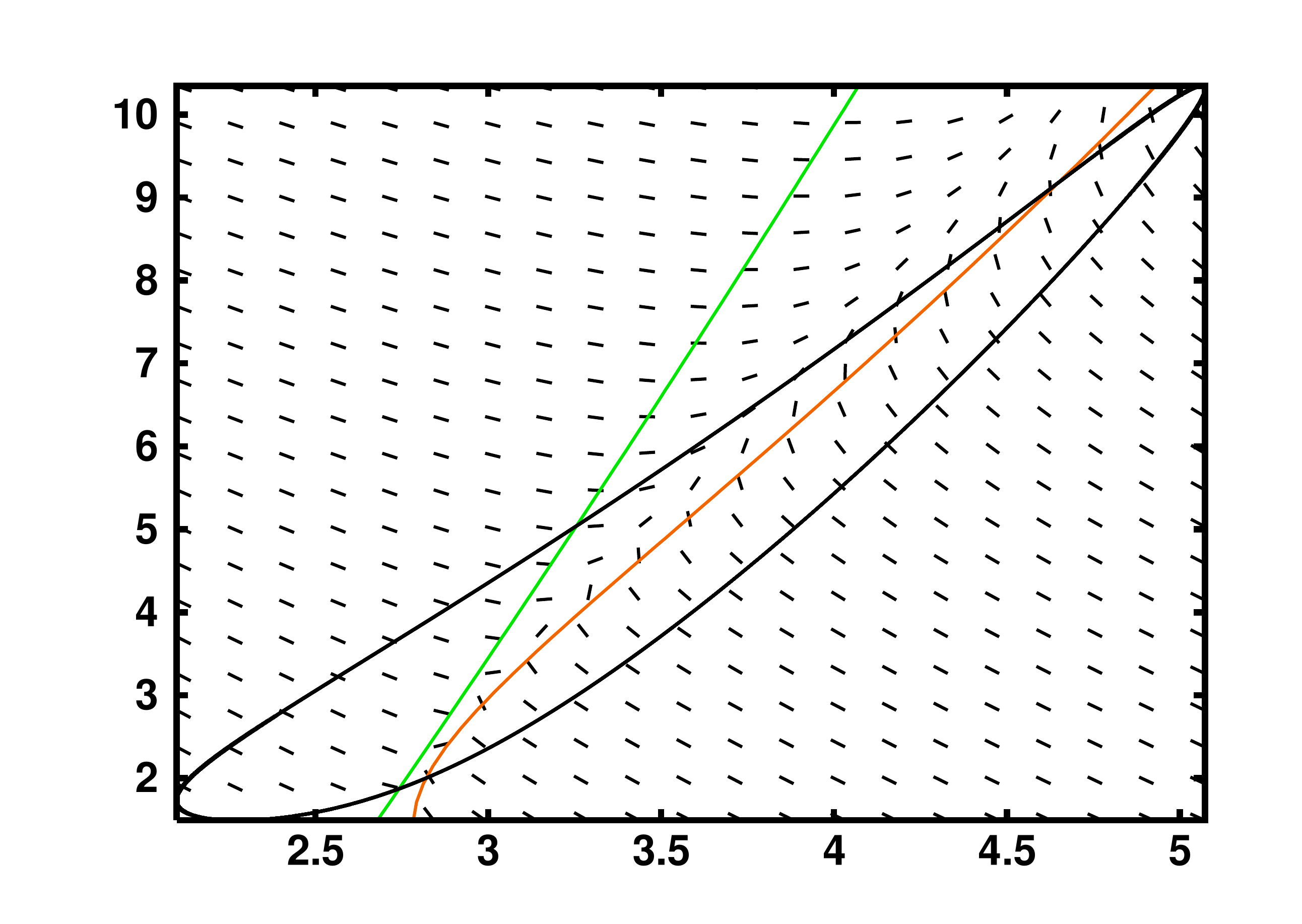}
  \caption{A stable limit cycle in the Two-dimensional phase space $M_{A},M_{J}$ for parameter set \ref{par1}}
  \label{fig:p2}
\end{subfigure}
\caption{Time-series and limit cycle in the 2-D phase space plot}
\label{fig:fig1}
\end{figure}

\begin{figure}
\begin{subfigure}{.45\textwidth}
  \centering
  \includegraphics[angle=-90,width=.8\linewidth]{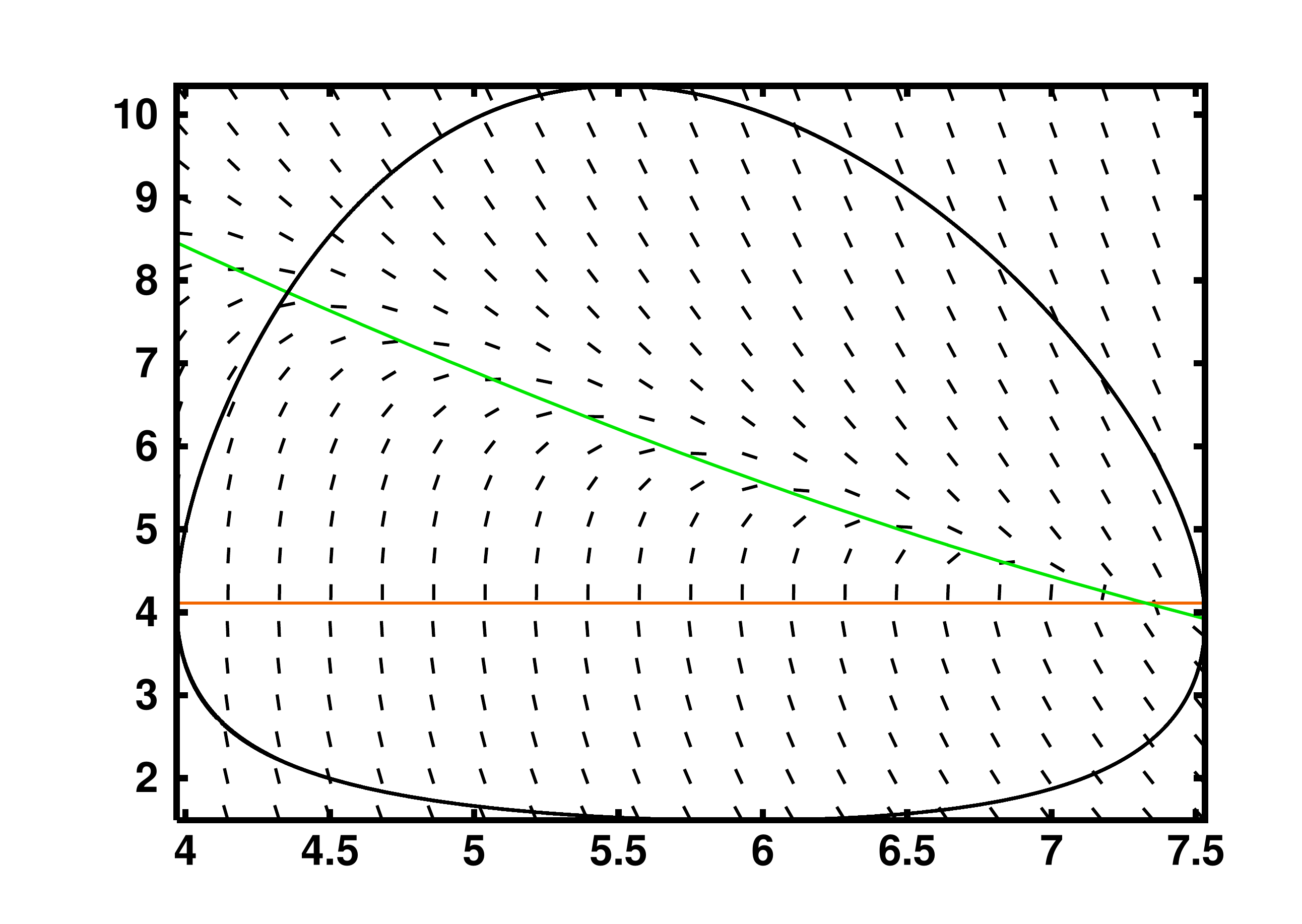}
  \caption{A stable limit cycle in the Two-dimensional phase space $C,M_{J}$ for parameter set \ref{par1}}
  \label{fig:p3}
\end{subfigure} \hfill
\begin{subfigure}{.45\textwidth}
  \centering
  \includegraphics[angle=-90,width=.8\linewidth]{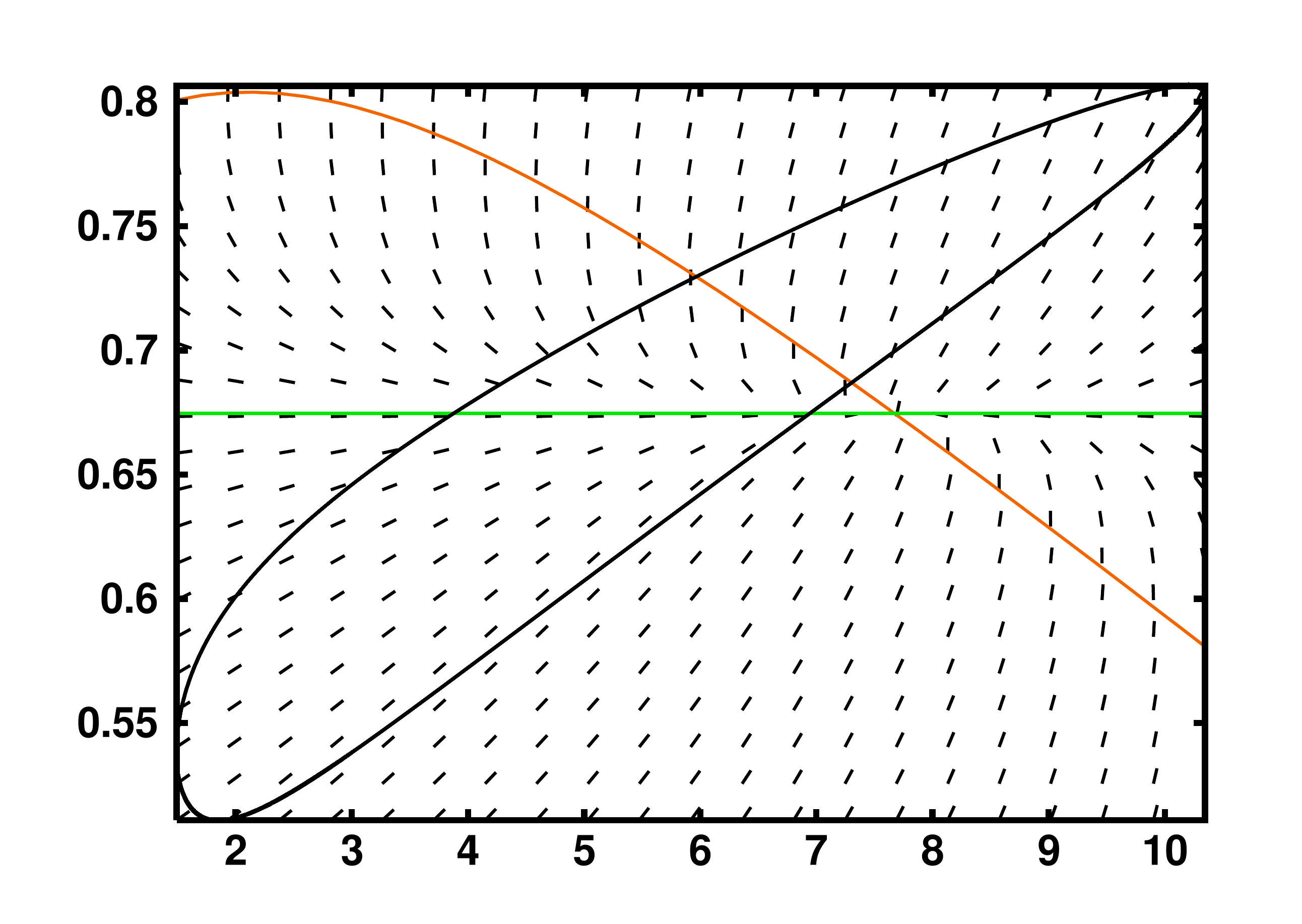}
  \caption{A stable limit cycle in the Two-dimensional phase space $E,M_{j}$ for parameter set \ref{par1}}
  \label{fig:p4}
\end{subfigure}
\caption{Stable limit cycle in the 2-D pahase space plot}
\label{fig:fig2}
\end{figure}

\begin{figure}
	\begin{subfigure}{.45\textwidth}
		\centering
		\includegraphics[angle=-90,width=.8\linewidth]{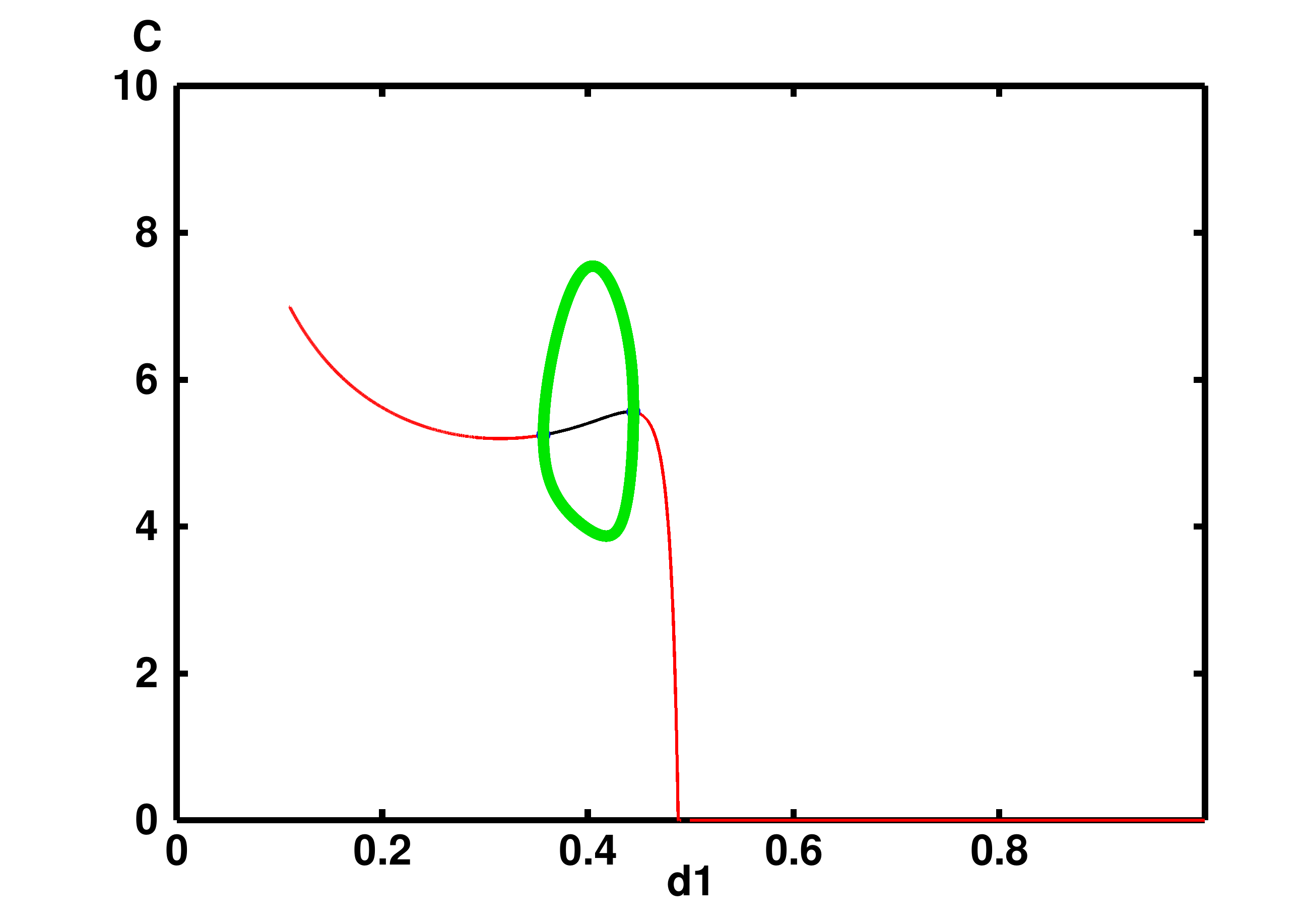}
		\caption{One-parameter bifurcation diagram with respect to parameter $d_{1}$}
		\label{fig:p5}
	\end{subfigure}\hfill
	\begin{subfigure}{.45\textwidth}
		\centering
		\includegraphics[angle=-90,width=.8\linewidth]{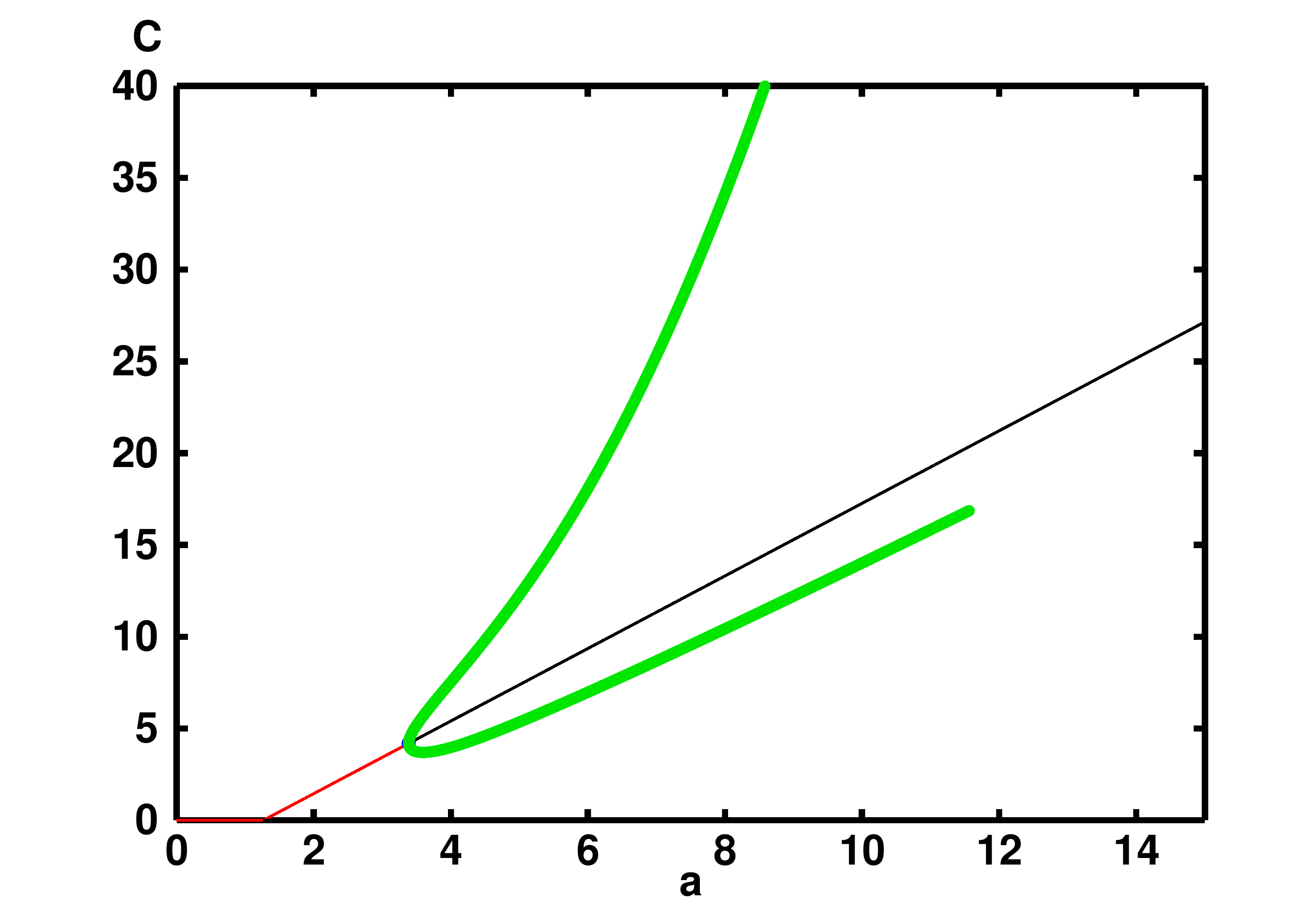}
		\caption{One-parameter bifurcation diagram with respect to parameter $a$}
		\label{fig:p6}
	\end{subfigure}
	\caption{One-parameter bifurcation diagrams to depict stablilty, Hopf Bifurcation point and periodic solutions with respect to parameter set \ref{par1}.}
	\label{fig:fig3}
\end{figure}

\begin{figure}
	\begin{subfigure}{.45\textwidth}
		\centering
		\includegraphics[angle=00,width=.8\linewidth]{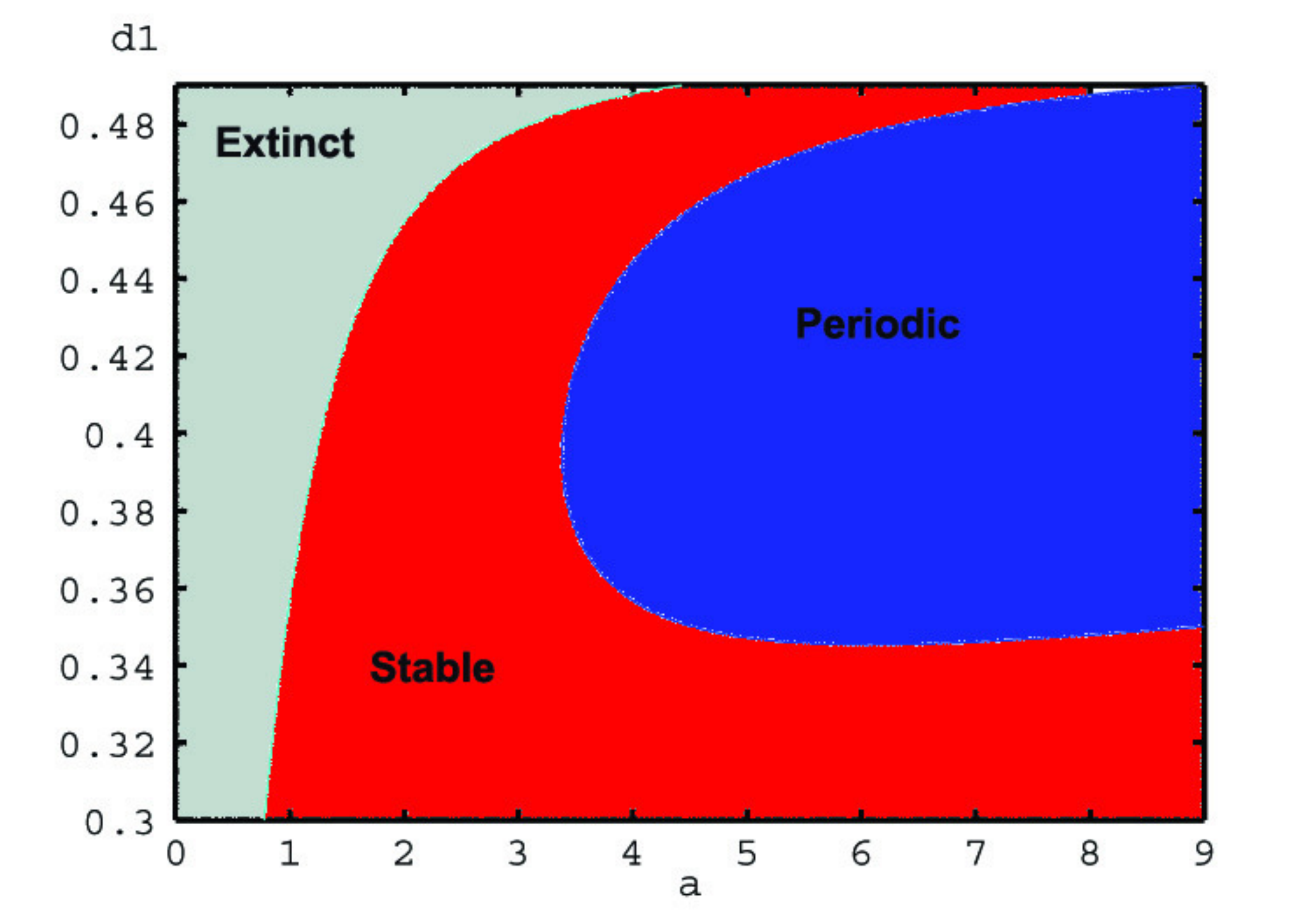}
		\caption{Two-parameters $(a,d_{1})$ bifurcation diagram with respect to parameter set \ref{par1}}
		\label{fig:p7}
	\end{subfigure} \hfill
	\begin{subfigure}{.45\textwidth}
		\centering
		\includegraphics[angle=00,width=.75\linewidth]{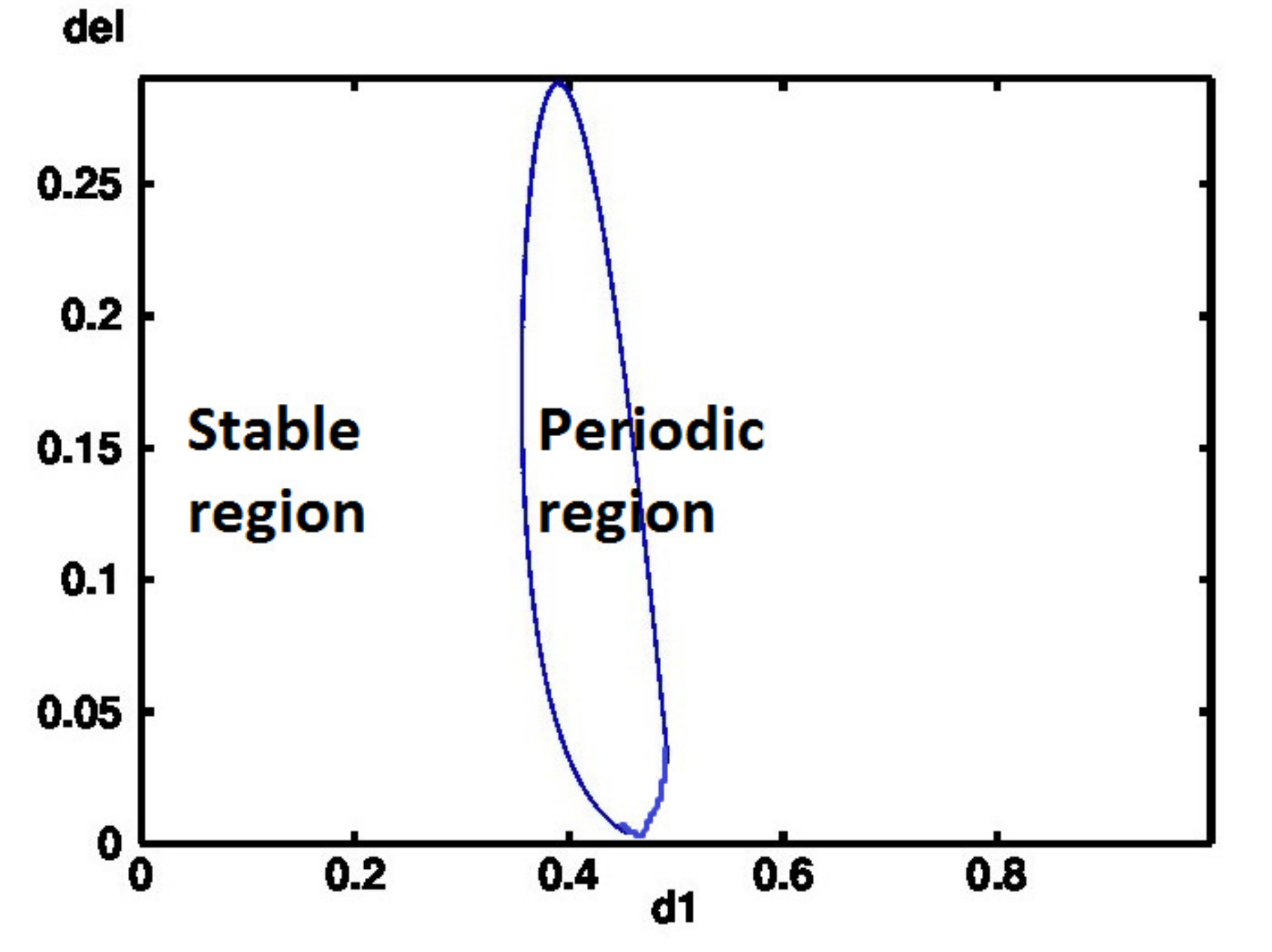}
		\caption{Two-parameters $(d_{1},\delta)$ bifurcation diagram with respect to parameter set \ref{par1}}
		\label{fig:p8}
	\end{subfigure}
	\caption{Two-parameter bifurcation diagrams to depict parameter region for the stable coexistence and periodic coexistence  with respect to parameter set \ref{par1}.}
	\label{fig:fig4}
\end{figure}


\begin{thebibliography}{10}


\bibitem{FC03}
\newblock Fisheries and Oceans Canada, 2003. 
\newblock \emph{Profile of the Blue Mussel (Mytilus edulis) Gulf Region. Policy and Economics Branch, Gulf Region},
\newblock  Department of Fisheries and Oceans, Moncton, New Brunswick, February, 2003.



\bibitem{FAO17}
\newblock \emph{Food and Aquaculture Organization of the United Nations [FAO] Cultured Aquatic Species Information Programme},
\newblock Mytilus edulis. In: FAO Fisheries and Aquaculture Department [online]: http://www.fao.org/fishery/culturedspecies/Mytilus$\_$edulis/en\#tcNA00D6. Update 1 January 2004. [Accessed 28 June 2017]. 



\bibitem{SS92}
\newblock Seed, R., Suchanek, T.H., 1992.
\newblock \emph{ Population and community ecology of Mytilus.},
\newblock In: Gosling, E. [editor], The Mussel Mytilus: ecology, physiology, genetics and culture. Developments in Aquaculture and Fisheries Science, 25. Elsevier (Amsterdam), 87-169, 1992. 


\bibitem{M83}
\newblock Menge, B.A., 1983.
\newblock \emph{ Components of predation intensity in the low zone of the New England rocky intertidal region},
\newblock Oecologia (Berl.) 58: 141-155, 1983. 


\bibitem{S85}
\newblock Suchanek, T.H.
\newblock \emph{ Mussels and their role in structuring rocky shore communities},
\newblock In: P.G. Moore and R. Seed (eds.), The Ecology of Rocky Coasts. Hodder and Stoughton, Sevenoaks, UK, 70-96, 1985. 



\bibitem{WD92}
\newblock Widdows, J., Donkin, P. 
\newblock \emph{Mussels and environmental contaminants: bioaccumulation and physiological aspects},
\newblock  In: Gosling, E. [editor], The Mussel Mytilus: ecology, physiology, genetics and culture.
Developments in Aquaculture and Fisheries Science, 25. Elsevier (Amsterdam), 383-424, 1992.



\bibitem{TN85}
\newblock Tsuchiya, M., Nishihara, M. 
\newblock \emph{Islands of Mytilus as a habitat for small intertidal animals: effect of island size on community structure},
\newblock Marine Ecology Progress Series 25:71-81, 1985.


\bibitem{SD17}
\newblock Sorte, C.J.B., Davidson, Victoria E., Franklin, M.C., Benes, Kylla M., Doellman, M.M., Etter, R.J., Hannigan, R.E., Lubchenco, J., Menge, B.A.
\newblock \emph{Long-term declines in an intertidal foundation species parallel shifts in community composition},
\newblock Global Change Biology, 23(1): 341-352, 2017. 


\bibitem{PI16}
\newblock Jiao, J., Pilyugin, S. S., Osenberg, C. W. 
\newblock \emph{Random movement of predators can eliminate trophic cascades in marine protected areas}
\newblock  Ecosphere, 7(8), 2016.


\bibitem{PID16}
\newblock Pilyugin, S. S., Medlock, J., De Leenheer, P.  
\newblock \emph{The effectiveness of marine protected areas for predator and prey with varying mobility} 
\newblock Theoretical population biology, 110, 63-77, 2016.




\bibitem{A10}
\newblock Auker, L et. al. 
\newblock \emph{The effects of Didemnum vexillum overgrowth on Mytilus edulis biology and ecology},\newblock PhD thesis, University of New Hampshire, 2010.



\bibitem{A14}
\newblock Auker, L et. al. 
\newblock \emph{Exploring biotic impacts from carcinus maenas predation and didemnum vexillum epibiosis on mytilus edilus in the gulf of maine},
\newblock Northeastern Naturalist, 21(3): 479-494, 2014.


\bibitem{DT04}
\newblock DeGraaf, J.D., Tyrrell, M.C.
\newblock \emph {Comparison of the feeding rates of two introduced crab species, Carcinus maenas and Hemigrapsus sanguineus, on the Blue Mussel, Mytilus edulis},
\newblock Northeastern Naturalist, 11(2): 163-167, 2004. 


\bibitem{BF14}
\newblock Li, B., Bewick, S., Shang, J., Fagan, W. F. 
 \newblock \emph{Persistence and spread of a species with a shifting habitat edge}, 
\newblock SIAM Journal on Applied Mathematics, 74(5), 1397-1417, 2014.










\bibitem{FB06}
\newblock Freeman, A.S., Byers, J.E., 2006. 
\newblock \emph{Divergent induced responses to an invasive predator in marine mussel populations}, \newblock Science, 313: 831-833, 2006.




\bibitem{E78}
\newblock Elner, R.W.
\newblock \emph{The mechanics of predation by the shore crab Carcinus maenas (L.) from Port Herbert, Southwestern Nova Scotia}, 
\newblock Journal of Shellfish Research, 1: 89-94, 1978. 






\bibitem{FD02}
\newblock Frandsen, R., Dolmer, P.
\newblock \emph{Effects of substrate type on growth and mortality of blue mussels (Mytilus edulis) exposed to the predator Carcinus maenas},
\newblock Marine Biology, 141(2): 253-262, 2002. 





\bibitem{DH07}
\newblock Dijkstra, J., Harris, L.G., Westerman, E.
\newblock \emph {Distribution and long-term temporal patterns of four invasive colonial ascidians in the Gulf of Maine},
\newblock Journal of Experimental Marine Biology and Ecology, 342: 61-68, 2007.

\bibitem{G07}
\newblock Negi, K.,  Gakkhar, S. 
\newblock \emph { Dynamics in a Beddington–DeAngelis prey–predator system with impulsive harvesting}., 
\newblock Ecological Modelling, 206(3-4), 421-430, 2007.



\bibitem{PBU16}
\newblock Parshad, R. D., Bhowmick, S., Quansah, E., Basheer, A., Upadhyay, R. K. 
\newblock \emph {Predator interference effects on biological control: The “paradox” of the generalist predator revisited},
\newblock Communications in Nonlinear Science and Numerical Simulation, 39, 169-184, 2016.


\bibitem{GG16}
\newblock Gupta, K., Gakkhar, S. 
\newblock \emph {The Filippov Approach for Predator-Prey System Involving Mixed Type of Functional Responses}, 
\newblock Differential Equations and Dynamical Systems, 1-21, 2016.










\bibitem{B07}
\newblock Bullard, S.G. et al.
\newblock \emph{The colonial ascidian Didemnum sp. A: current distribution, basic biology and potential threat to marine communities of the northeast and west coasts of North America}, 
\newblock Journal of Experimental Marine Biology and Ecology, 342: 99-108, 2007. 


\bibitem{KP15} Krivan, V., Priyadarshi, A.
\newblock  L-shaped prey isocline in the Gause predator-prey experiments with a prey refuge.
\newblock Journal of Theoretical Biology, 370:21-26, 2015.

\bibitem{LW99}
\newblock Laudien, J., Wahl, M.
\newblock \emph{Indirect effects of epibiosis on host mortality: seastar predation on differently fouled mussels},
\newblock Marine Ecology, 20:35-47, 1999.

\bibitem{PQB16}
\newblock  Parshad, R. D, Qansah, E., Black, K. and Beauregard, M.
\newblock Biological control via ``ecological" damping: An approach that attenuates non-target effects.
\newblock Mathematical Biosciences, 273:23-44, 2016.


\bibitem{WL98}
\newblock Wahl, M., Kröger, K., Lenz, M.
\newblock \emph{Non-toxic protection against epibiosis},
\newblock Biofouling, 12: 205-226, 1998.


\bibitem{HS93}
\newblock Harper, E.M., Skelton, P.W.
\newblock \emph{A defensive value of the thickened periostracum in the Mytiloidea}, 
\newblock Veliger, 36(1): 36-42, 1993. 


 


\bibitem{E03}
\newblock Enderlein, P., Moorthi, S., Rohrscheidt, H. and Wahl, M.
\newblock \emph{Optimal foraging versus shared doom effects: interactive influence of mussel size and epibiosis on predator preference}, 
\newblock Journal of Experimental Marine Biology and Ecology, 292: 231-242, 2003. 



\bibitem{T07}
\newblock Thornber, C. 
\newblock \emph{Associational resistance mediates predator–prey interactions in a marine subtidal system},
\newblock Marine Ecology, 28: 480-486,  2007.


\bibitem{V06}
\newblock Valeria Bers, A., D'Souza, F., Klinjstra, J., Willemsen, P. and Wahl, M.
\newblock \emph{Chemical defence in mussels: antifouling effect of crude extracts of the periostracum of the blue mussel Mytilus edulis},
\newblock  Biofouling, 22(4): 251-259, 2006. 


\bibitem{D77}
\newblock N.B. Davies.
\newblock \emph {Prey selection and the search strategy of the spotted flycatcher (Muscicapa striata): A field study on optimal foraging}, 
\newblock Animal Behaviour, 25: 1016-1033, 1977.



\bibitem{K96}
\newblock Vlastimil Krivan.
\newblock \emph {Optimal foraging and predator-prey dynamics},
\newblock Theoretical Population Biology, 49:265-290, 1996.





\bibitem{K99}
\newblock Krivan, Vlatsimil and Sikder, A. 
\newblock \emph {Optimal foraging and predator-prey dynamics II},
\newblock Theoretical Population Biology, 55: 111-126, 1999.

\bibitem{LB15}	
\newblock Pribylova, L., and Berec, L.,
\newblock  Predator interference and stability of predator-prey dynamics
\newblock Journal of Mathematical Biology, 71: 301-323, 2015.

\bibitem{BM01}
\newblock Van Baalen, Minus, et al. 
\newblock \emph{Alternative food, switching predators, and the persistence of predator-prey systems}, \newblock The American Naturalist, 157(5): 512-524, 2001. 



\bibitem{WG81}
\newblock Werner, Earl E., Gary G. Mittelbach. 
\newblock \emph{Optimal foraging: field tests of diet choice and habitat switching}, 
\newblock American Zoologist, 21(4): 813-829, 1981.


\bibitem{PL13}
\newblock Perko, Lawrence. 
\newblock \emph{Differential equations and dynamical systems},
\newblock  Springer Science \& Business Media, Vol. 7, 2013.


\bibitem{L94}
\newblock W.M. Liu.
\newblock \emph {Criterion of Hopf Bifurcations without Using Eigenvalues},
\newblock Journal of Mathematical Analysis and Applications, 182:250-256, 1994.


\bibitem{LJ07}
\newblock Lenhart, Suzanne, and John T. Workman. 
\newblock \emph{Optimal control applied to biological models},
\newblock Crc Press, 2007.


\bibitem{FR75}
\newblock W.H. Fleming and R.W. Rishel.
\newblock \emph{Deterministic and Stochastic Optimal Control}
\newblock Springer Verlag, New York, 1975. 


\bibitem{IR06}
\newblock Inderjit, Chapman, D., Ranelletti, M., Kaushik, S.
\newblock \emph{ Invasive marine algae: an ecological perspective},
\newblock Botanical Review, 72:153-178, 2006.cf



\bibitem{WL02}
\newblock Wolfe, L.M., 2002. 
\newblock \emph {Why alien invaders succeed: support for the escape-from-enemy hypothesis}, \newblock American Naturalist, 160: 705-711, 2002. 



\bibitem{CR04}
\newblock Callaway, R.M., Ridenour, W.M.
\newblock \emph {Novel weapons: invasive success and the evolution of increased competitive ability},
\newblock Frontiers in Ecology and the Environment, 2: 436-443, 2004.




\bibitem{PP02}
\newblock Pisut, D.P., Pawlik, J.R.
\newblock \emph{Anti-predatory chemical defenses of ascidians: secondary metabolites or inorganic acids? },
\newblock Journal of Experimental Marine Biology and Ecology, 270: 203-214, 2002.


\bibitem{WH95}
\newblock Wahl, M., Hay, M.E.
\newblock \emph{ Associational resistance and shared doom: effects of epibiosis on herbivory},  \newblock Oecologia, 102: 329-340, 1995.


\bibitem{B16} 
\newblock  Basheer A, Quansah E, Bhowmick S, Parshad RD.
\newblock  \emph{Prey cannibalism alters the dynamics of Holling Tanner-type predator prey models},
\newblock  Nonlinear Dynamics, Vol 85, no.4, pp 2549-2567, 2016.


\bibitem{BL16} 
\newblock Basheer A, Lyu J, Giffin A and Parshad RD.
\newblock \emph{The destabilizing effect of cannibalism in a spatially explicit three-species Age structured predator-prey model},
\newblock Complexity, 2017.


\bibitem{M95}
\newblock Moreno, C.A. 
\newblock \emph{Macroalgae as a refuge from predation for recruits of the mussel Choromytilus chorus (Molina, 1782) in Southern Chile},
\newblock Journal of Experimental Marine Biology and Ecology, 191(2): 181-193, 1995.

\bibitem{L16}
\newblock Larsen P.S., Riisgard H.U.
\newblock \emph{Growth-prediction model for blue mussels (Mytilus edulis) on future optimally thinned farm-ropes in Great Belt (Denmark)},
\newblock Journal of Marine Science and Engineering, 4, 2016.


















 %\bibitem{K07}
%\newblock Kuang, Yang.
%\newblock \emph {Some mechanistically derived population models},
%\newblock Mathematical Biosciences and Engineering, 4(4):1-11, 2007.

%\bibitem{AO08}
%\newblock Auker, L.A., Oviatt, C.A.
%\newblock \emph {Factors influencing the recruitment and abundance of Didemnum in Narragansett Bay},
%\newblock ICES Journal of Marine Science, 65 (5): 765-769, 2008. 

%\bibitem{B84}
%\newblock Boulding, E.G.
%\newblock \emph {Crab-resistant features of shells of burrowing bivalves: decreasing vulnerability by increasing handling time},
%\newblock Journal of Experimental Biology and Ecology, 76: 201-223, 1984

%\bibitem{L80}
%\newblock Lutz, R.A. 
%\newblock \emph{Mussel Culture and Harvest: A North American Perspective},
%\newblock Elsevier Science Publishers, B.V., Amsterdam, pp 305, 1980. 

%\bibitem{SC04}
%\newblock Saier, B., Chapman, A.S.
%\newblock \emph{Crusts of the alien bryozoan Membranipora membranacea can negatively impact spore output from native kelps (Laminaria longicruris)},
%\newblock Botanica Marina, 47: 265-271, 2004.

%\bibitem{TB07}
%\newblock Thieltges, D.W., Buschbaum, C.
%\newblock \emph{Mechanism of an epibiont burden: Crepidule fornicata increases byssus thread production by Mytilus edulis},
%\newblock Journal of Molluscan Studies, 73: 75-77, 2007.

%\bibitem{B07}
%\newblock Buschbaum, C., Buschbaum, G., Schrey, I. and Thieltges, D.
%\newblock \emph {Shell-boring polychaetes affect gastropod shell strength and crab predation},
%\newblock Marine Ecology Progress Series, Vol. 329: 123-130, 2007.



%\bibitem{W89}
%\newblock Wahl, M.
%\newblock \emph{Marine epibiosis. 1. Fouling and antifouling – some basic aspects},
%\newblock Marine Ecology Progress Series 58, 175-189, 1989. 

%\bibitem{VC09}
%\newblock Valentine, P.C., Carman, M.R., Dijkstra, J., Blackwood, D.S.
%\newblock \emph{Larval recruitment of the invasive colonial ascidian Didemnum vexillum, seasonal water temperatures in New England coastal and offshore waters, and implications for spread of the species},
%\newblock Aquatic Invasions, 4(1): 153-168, 2009. 

%\bibitem{J85}
%\newblock Jenkins, R.J., 1985. 
%\newblock \emph{Mussel cultivation in the Marlborough Sounds (New Zealand)},
%\newblock New Zealand Fishing Industry Board, Wellington, 2nd Edition, pp 77, 1985.


%\bibitem{LB92}
%\newblock Lambert, W.J., Levin, P.S., Berman, J.
%\newblock \emph{Changes in the structure of a New England (USA) kelp bed: the effects of an introduced species},
%\newblock Marine Ecology Progress Series 88: 303-307, 1992. 

%\bibitem{HH06}
%\newblock Hepburn, C.D., Hurd, C.L., Frew, R.D. 
%\newblock \emph{Colony structure and seasonal differences in light and nitrogen modify the impact of sessile epifauna on the giant kelp Macrocystis pyrifera (L.) C Agardh},
%\newblock Hydrobiologia, 560: 373-384, 2006. 

%\bibitem{HS81}
%\newblock Hughes, R.N., Seed, R.
%\newblock \emph {Size selection of mussels by the blue crab Callinectes sapidus: energy maximize or time minimizer? }
%\newblock Marine Ecology Progress Series 6: 83-89, 1981. 

%\bibitem{HB92}
%\newblock Hunter, R.D., Bailey, J.F., 1992. 
%\newblock \emph{Dreissena polymorpha (zebra mussel): colonization of soft substrata and some effects on unionid bivalves},
%\newblock Nautilus, 106: 60-67, 1992.

%\bibitem{CL12}
%\newblock Cebrián, E., Linares, C., Marschal, C., Garrabou, J.
%\newblock \emph {Exploring the effects of invasive algae on the persistence of gorgonian populations}, %\newblock Biological Invasions,14: 2647-2656, 2012.

%\bibitem{CC05}
%\newblock Chan, D.H.L., Chan, B.K.K.
%\newblock \emph {Effect of epibiosis on the fitness of the sandy shore snail Batillaria zonalis in Hong Kong},
%\newblock Marine Biology, 146: 695-705, 2005.


%\bibitem{EB11}
%\newblock Eschweiler, N., Buschbaum, C.
%\newblock \emph{Alien epibiont (Crassostrea gigas) impacts on native periwinkles (Littorina littorea)}, 
%\newblock Aquatic Invasions, 6: 281-290, 2011.

%\bibitem{FH09}
%\newblock Firstarter, F.N., Hidalgo, F.J., Lomovasky, B.J., Gallegos, P., Gamero, P., Iribarne, O.O. \newblock  \emph{Effects of epibiotic Enteromorpha spp. on the mole crab Emerita analoga in the Peruvian central coast}, 
%\newblock Journal of the Marine Biological Association of the United Kingdom, 89:363-370, 2009. 


%\bibitem{DR91}
%\newblock Dittman, D., Robles, C.
%\newblock \emph{Effect of algal epiphytes on the mussel Mytilus californianus},
%\newblock Ecology, 72:286-296, 1991.

%\bibitem{DB03}
%\newblock Donovan, D.A., Bingham, B.L., From, M., Fleisch, A.F., Loomis, E.S., 2003. 
%\newblock \emph{Effects of barnacle encrustation on the swimming behavior, energetics, morphometry, and drag coefficient of the scallop Chlamys hastata},
%\newblock Journal of the Marine Biological Association of the United Kingdom, 83: 813-819, 2003.

\end{thebibliography}
\end{document}